\documentclass[acmsmall,screen,nonacm]{acmart}

\AtBeginDocument{%
	}
\setcopyright{none}

\usepackage[T1]{fontenc}
\usepackage{hyperref}
\usepackage{fancyhdr}
\usepackage{graphicx}
\usepackage{amsmath}
\usepackage{stmaryrd,wasysym}
\usepackage{booktabs,comment} 
\usepackage{xcolor}

\usepackage{framed}

\usepackage{multirow}

\usepackage[lined,boxed,ruled]{algorithm2e}

\usepackage[normalem]{ulem}
\usepackage{tikz,pgffor}
\usetikzlibrary{arrows}
\usetikzlibrary{shapes}
\usetikzlibrary{calc}
\usetikzlibrary{automata}
\usetikzlibrary{positioning}

\usepackage{bbm}

\tikzstyle{proglabel}=[shape=circle,draw,inner sep=0pt,minimum size=5mm]
\tikzstyle{tran}=[draw,->,>=stealth, rounded corners]

\usepackage{listings}
\lstdefinelanguage{prog}
{
	morekeywords={if, then, else, fi, while, do, od, true, false, and, or, skip, sample, observe,return,score,normal,uniform, prob},
	sensitive = false
}

\usepackage{threeparttable}
\usepackage{wrapfig}
\usepackage{subfigure}

\usepackage{tabularx}
\usepackage{dsfont}

\usepackage{bussproofs}

\usepackage[capitalise]{cleveref}

\makeatletter
\providecommand{\bigsqcap}{
	\mathop{
		\mathpalette\@updown\bigsqcup
	}
}
\newcommand*{\@updown}[2]{
	\rotatebox[origin=c]{180}{$\m@th#1#2$}
}
\makeatother

\newenvironment{scprooftree}[1]%
{\gdef\scalefactor{#1}\begin{center}\proofSkipAmount \leavevmode}%
	{\scalebox{\scalefactor}{\DisplayProof}\proofSkipAmount \end{center} }

\newcommand\calM{\mathcal{M}}
\newcommand\calV{\mathcal{V}}
\newcommand\calU{\mathcal{U}}
\newcommand\calW{\mathcal{W}}
\newcommand\calP{\mathcal{P}}
\newcommand\calD{\mathbb{D}}
\newcommand{\expect}[1]{\mathbb{E}\left[{#1}\right]}

\newcommand{\expectdist}[2]{\mathbb{E}_{#1}\left[ {#2} \right]}

\def\solvable{\#}

\newcommand{\supp}[1]{\textsf{supp}\left({#1}\right)}

\definecolor{codegreen}{rgb}{0,0.6,0}
\definecolor{codegray}{rgb}{0.5,0.5,0.5}
\definecolor{codepurple}{rgb}{0.58,0,0.82}
\definecolor{backcolour}{rgb}{0.95,0.95,0.92}

\lstdefinestyle{myStyle}{
    belowcaptionskip=1\baselineskip,
    breaklines=true,
    frame=none,
    basicstyle=\footnotesize\ttfamily,
    keywordstyle=\bfseries\color{green!40!black},
    commentstyle=\itshape\color{purple!40!black},
    identifierstyle=\color{blue},
    backgroundcolor=\color{gray!10!white},
    numberstyle=\tiny\color{codegray},
    stringstyle=\color{codepurple},
    breakatwhitespace=false,                          
    keepspaces=true,                 
    numbers=left,       
    numbersep=5pt,                  
    showspaces=false,                
    showstringspaces=false,
    showtabs=false,                  
    tabsize=2,
}

\newcommand\doubleplus{+\kern-1.3ex+\kern0.8ex}

\newcommand{\tuple}[1]{\ensuremath{\langle #1 \rangle}}
\newcommand{\expv}{\mathbb{E}}

\newcommand{\tr}{\boldsymbol{t}}

\newcommand{\pv}{\mathbf{v}}
\newcommand{\rv}{\mathbf{r}}

\newcommand{\Rset}{\mathbb{R}}
\newcommand{\Nset}{\mathbb{N}}
\newcommand{\Zset}{\mathbb{Z}}

\newcommand{\expt}{\mathbb{E}}

\newcommand{\trans}{\mathcal{T}}

\newcommand{\pspace}{(\Omega,\mathcal{F},\probm)}
\newcommand{\probm}{\mathbb{P}}
\newcommand{\condexpv}[2]{{\expt}{\left[{#1} \mid {#2}\right]}}

\newcommand{\valueSem}[1]{\mathsf{val}_{#1}} % value (semantics)
\newcommand{\weightSem}[1]{\mathsf{wt}_{#1}} % weight (semantics)
\newcommand{\measureSem}[1]{\llbracket #1 \rrbracket}
\newcommand{\posterior}{\mathsf{posterior}}

\newcommand{\loc}{\ell}
\newcommand{\locs}{\mathit{L}}

\newcommand{\lin}{\loc_\mathrm{init}}
\newcommand{\lout}{\loc_\mathrm{out}}
\newcommand{\val}[1]{\mbox{\sl Val}_{#1}}

\newcommand{\pvars}{V_\mathrm{p}}
\newcommand{\rvars}{V_{\mathrm{r}}}

\newcommand{\sle}{\sqsubseteq}
\newcommand{\sge}{\sqsupseteq}

\newcommand{\lfp}{\mathrm{lfp}}
\newcommand{\gfp}{\mathrm{gfp}}

\newcommand{\rdvarjdis}{\mathcal D}

\newcommand{\upd}{\mbox{\sl upd}}
\newcommand{\wet}{\mbox{\sl wt}}
\newcommand{\transset}{\mathfrak T}
\newcommand{\valin}{\pv_{\mathrm{init}}}

\newcommand{\win}{w_{\mathrm{init}}}

\newcommand{\trunc}{\mathcal{B}}

\newcommand{\ewt}{\mbox{\sl ewt}}

\newcommand{\monomials}{\mathbf{M}}

\renewcommand{\paragraph}[1]{\smallskip\noindent {\em #1}}  

\newtheorem{remark}{Remark}

\begin{document}

\title{Static Posterior Inference of Bayesian Probabilistic Programming via Polynomial Solving}

\thanks{This work is a technical report for the paper with the same name to appear in the 45th ACM SIGPLAN Conference on Programming Language Design and Implementation (PLDI 2024).  \\}

\author{Peixin Wang$^{*,\dag}$}
\affiliation{
\institution{Nanyang Technological University}
  \country{Singapore}
}
\email{peixin.wang@ntu.edu.sg}

\author{Tengshun Yang$^*$}
\affiliation{            
  \institution{SKLCS, Institute of Software, University of Chinese Academy of Sciences}   
  \country{China} 
} 
\email{yangts@ios.ac.cn}

\author{Hongfei Fu$^\dag$}
\affiliation{
  \institution{Shanghai Jiao Tong University}
  \country{China}
}
\email{jt002845@sjtu.edu.cn}

\author{Guanyan Li}
\affiliation{
  \institution{University of Oxford}
  \country{United Kingdom}
}
\email{guanyan.li@cs.ox.ac.uk}

\author{C.-H. Luke Ong}
\affiliation{
  \institution{Nanyang Technological University}
  \country{Singapore}
}
\email{luke.ong@ntu.edu.sg}

\begin{abstract}
In Bayesian probabilistic programming, a central problem is to estimate the normalised posterior distribution (NPD) of a probabilistic program with conditioning via {\tt score} (a.k.a. {\tt observe}) statements. 
Most previous approaches address this problem by Markov Chain Monte Carlo and variational inference, 
and therefore could not generate guaranteed outcomes within a finite time limit. Moreover, existing methods for exact inference either impose syntactic restrictions or cannot guarantee successful inference in general.

In this work, we propose a novel automated approach to derive guaranteed bounds for NPD via polynomial solving. We first establish a fixed-point theorem for the wide class of \emph{score-at-end} Bayesian probabilistic programs that terminate almost-surely and have a single bounded score statement at program termination. Then, we propose a multiplicative variant of Optional Stopping Theorem (OST) to address \emph{score-recursive} Bayesian programs where score statements with weights greater than one could appear inside a loop. Bayesian nonparametric models, enjoying a renaissance in statistics and machine learning, can be represented by score-recursive Bayesian programs and are difficult to handle due to an integrability issue. Finally, we use polynomial solving to implement our fixed-point theorem and OST variant.
To improve the accuracy of the polynomial solving, we further propose a truncation operation and the synthesis of multiple bounds over various program inputs.  
Our approach can handle Bayesian probabilistic programs with unbounded while loops and continuous distributions with infinite supports.
Experiments over a wide range of benchmarks show that compared with the most relevant approach (Beutner \emph{et al.}, PLDI 2022) for guaranteed NPD analysis via recursion unrolling, our approach is more time efficient and derives comparable or even tighter NPD bounds. Furthermore, our approach can handle score-recursive programs which previous approaches could not. 
\end{abstract}

%%% 2012 ACM Computing Classification System (CSS) concepts
%%% Generate at 'http://dl.acm.org/ccs/ccs.cfm'.



\maketitle   

\def\thefootnote{*}\footnotetext{Equal contribution}\def\thefootnote{\arabic{footnote}}
\def\thefootnote{\dag}\footnotetext{Corresponding authors}\def\thefootnote{\arabic{footnote}}

\section{Introduction}\label{sec:intro}
Bayesian probabilistic programming~\cite{DBLP:journals/corr/abs-1809-10756,rainforth2017automating} is a programming paradigm that incorporates
Bayesian reasoning into programming languages, and aims at first modelling probabilistic models as probabilistic programs and then analyzing the models through their program representations. Compared with traditional approaches \cite{McIverM04,McIverM05,SriramCAV,ChatterjeeFG16} that specify an ad-hoc programming language, probabilistic programming languages (PPLs) \cite{DBLP:journals/corr/abs-1809-10756} provide a universal framework to perform Bayesian inference. PPLs have two specific constructs: {\tt sample} and {\tt score} \cite{borgstrom2016lambda}.\footnote{Sometimes {\tt observe} is used instead of {\tt score}~\cite{gordon2014probabilistic}, which has the same implicit effect.} The {\tt sample} construct describes the prior probabilities, allowing to draw samples from a (prior) distribution. The {\tt score} construct records the likelihood of observed data in the form of ``{\tt score}(weight)'',\footnote{The argument ``weight'' corresponds to the likelihood each time the data is observed.} and is typically used to weight the current execution in Monte Carlo simulation. 
Nowadays, Bayesian probabilistic programming has become an active research subject in statistics, machine learning and programming language communities, for which typical Bayesian programming languages include Pyro~\cite{bingham2019pyro}, WebPPL~\cite{dippl}, Anglican~\citep{DBLP:conf/pkdd/TolpinMW15}, Church~\cite{DBLP:conf/uai/GoodmanMRBT08}, etc. 

In this work, we consider the analysis of the normalised posterior distribution (NPD) in Bayesian probabilistic programs. The general statement of the problem is that: given a prior distribution $p(z)$ 
over the latent variables $z\in\Rset^n$ of interest, and a probabilistic model represented by a probabilistic program whose distribution $p(x,z)$ is obtained by observing the event $x\in\Rset^m$ with the likelihood $p({x}{\mid}{z})$, the target is to calculate the NPD $p({z}{\mid}{x})$ by Bayes' rule. 
There are two mainstream variants of the NPD conditioning. The first is soft conditioning~\cite{DBLP:conf/lics/StatonYWHK16} that 
assigns a non-negative weight to the program based on the probability (density) of
a given event occurring. The second is hard conditioning that restricts the weights to be either $0$ or $1$. 
Hard and soft conditioning are incomparable in general, as in some situations ``hard conditioning is a particular case of soft conditioning''~\cite[Page 42]{rainforth2017automating}, while in other cases hard conditioning is more general.

In the literature, there are two classes of approaches to address the 
NPD problem. The first is approximate approaches that estimate the NPD by random simulation, while the second is formal approaches that aim at deriving guaranteed bounds for NPD. 
In approximate approaches, two dominant methods are Markov chain Monte Carlo~\cite{gamerman2006markov} and variational inference~\cite{blei2017variational}. Although approximate approaches can produce approximate results efficiently, they cannot provide formal guarantee within a finite time limit. 
Moreover, as shown in ~\citet{Beutner2022b}, approximate approaches may produce inconsistent results between different simulation methods, which led the machine learning community to develop new variants~\cite{DBLP:conf/icml/MakZO22}. 
In formal approaches, there is a large amount of existing works such as $(\lambda)$PSI~\cite{DBLP:conf/cav/GehrMV16,DBLP:conf/pldi/GehrSV20}, AQUA~\cite{DBLP:conf/atva/HuangDM21}, Hakaru~\cite{DBLP:conf/flops/NarayananCRSZ16} and SPPL~\cite{DBLP:conf/pldi/SaadRM21}, aiming to derive exact inference for NPD. However, these methods are restricted to specific kinds of programs, e.g., programs with closed-form solutions to NPD or without continuous distributions, and none of them can handle probabilistic programs with unbounded while-loops/recursion. 
In recent works, \citet{DBLP:conf/nips/ZaiserMO23} and~\citet{oopsla24} used probability generating functions (PGF) to do exact inference for NPD. However, the former work cannot handle loopy programs, and both of them require a closed-form solution to NPD and only work for discrete observations. The recent work by~\citet{Beutner2022b} infers guaranteed bounds for NPD and handles unbounded recursion and continuous distributions. 
This approach relies on recursion unrolling and hence suffers from the path explosion problem.

\vspace{1ex}
\noindent{\em Challenges and gaps.} In this work, we focus on developing formal approaches to derive guaranteed bounds for NPD over loopy probabilistic programs in the setting of soft conditioning. From the literature, a main challenge is to develop new techniques that circumvent the path explosion problem from the approach~\cite{Beutner2022b}. Another challenge (and gap) is that existing approaches cannot handle the situation where {\tt score} statements with weights greater than $1$ appear inside a loop (which we refer to as \emph{score-recursive} programs), which has received significant attention in statistical phylogenetics~\cite{ronquist2021universal,treeflow}.

For score-recursive programs, the following example shows that {\tt score} inside a loop may cause an integrability issue and thus requires careful treatment. Consider a simple loop ``\textbf{while true do if prob}($0.5$) \textbf{then break else score}($3$) \textbf{fi od}''.
In each loop iteration, the loop terminates directly with probability $\frac{1}{2}$, and continues to execute a score command ``{\tt score}$(3)$'' with the same probability. 
It follows that the normalising constant in NPD is equal to $\sum_{n=1}^{\infty} \probm(T=n)\cdot 3^n=\sum_{n=1}^{\infty} (\frac{3}{2})^n=\infty$,
so that the infinity makes the posterior distribution invalid. 
This is noted in e.g., 
~\citet{DBLP:conf/lics/StatonYWHK16}  
that
unbounded weights may introduce the possibility of ``infinite model evidence errors''. To circumvent the drawback, previous results (e.g.,~\citet{borgstrom2016lambda}) allow only $1$-bounded weights. 

In probabilistic program analysis, polynomial solving~\cite{cost2019wang,DBLP:conf/pldi/WangS0CG21,DBLP:conf/cav/ChakarovS13,DBLP:journals/toplas/ChatterjeeFNH18,ChatterjeeFG16} is a well-established technique and 
naturally avoids the path explosion problem in the approach of \citet{Beutner2022b}.   
In this work, we leverage polynomial solving to address the NPD problem. 
Note that simply applying well-known polynomial solving techniques does not suffice 
for the following reasons:
(a) Polynomial solving is tight usually over a bounded region (see e.g., Weierstrass Approximation Theorem~\cite{jeffreys1988weierstrass}), and in general is not accurate if the region is unbounded;
(b) Polynomial solving synthesizes a single bound. 
However, having a single bound is not enough to get tight bounds for NPD, as one needs different bounds for different program inputs to achieve tightness in the normalisation. 

We address the challenges and gaps mentioned above. Our contributions are as follows.

\paragraph{Our contributions.}  In this work, we present the following contributions: 
\begin{itemize}
\item First, we establish a fixed-point theorem and a multiplicative variant of Optional Stopping Theorem (OST)~\cite{doob1971martingale,williams1991probability}.  
Our fixed-point theorem targets Bayesian probabilistic programs that have almost-sure termination and a single score statement at the end of the programs with a bounded score function (referred to as \emph{score-at-end} Bayesian programs), which is a wide class of Bayesian programs in the literature~\cite{DBLP:conf/cav/GehrMV16,DBLP:conf/pldi/GehrSV20,Beutner2022b}. Our OST variant targets 
\emph{score-recursive} Bayesian programs and 
addresses the integrability issue in these  programs.
\item Second, we apply polynomial solving techniques with our fixed-point theorem and OST variant. In addition to existing polynomial solving techniques, our approach improves the accuracy of the derived NPD bounds by the following: First, we propose a novel truncation operation that truncates a probabilistic program into a bounded range of program values.
Second, we devise our algorithm to synthesize multiple bounds for various program inputs. 
\end{itemize}
Experimental results show that our approach can handle a wide range of benchmarks including non-parametric examples such as Pedestrian~\cite{Beutner2022b} and score-recursive examples such as phylogenetic models~\cite{ronquist2021universal}. 
Compared with the previous approach~\cite{Beutner2022b} over score-at-end benchmarks, our approach reduces the runtime by up to $15$ times, while deriving comparable or even tighter bounds for NPD.

\paragraph{Limitations.} Our approach has the combinatorial explosion in the degree of polynomial solving. However, by our experimental results, a moderate choice of the degree (e.g., $\le 10$) suffices. Moreover, our synthesis of multiple bounds for various inputs  mitigates the combinatorial explosion. Another limitation is that in our polynomial solving, we utilize linear and semidefinite programming solvers, which may produce unsound results due to numerical errors.

\section{Motivating Examples}\label{sec:overview}
We present two motivating examples to highlight our key novelties.

\subsection{Pedestrian Random Walk}\label{sec3:pedestrian}

\begin{figure}
\lstset{language=prog}
\lstset{linewidth=5.1cm}
\begin{minipage}{0.45\textwidth}
\begin{lstlisting}[basicstyle=\small,mathescape]
    $\mbox{\sl start} := $ sample uniform$(0,3)$;
    $\mbox{\sl pos} := \mbox{\sl start}$; $\mbox{\sl dist} := 0$;
$\lin:$ $\color{blue} h(pos,dis)=a_1\cdot pos+a_2\cdot dis+a_3$
    $\color{blue} B=\{(pos,dis)\mid pos\in [0,5], dis\in [0,5]\}$ 
    $\color{blue} \calM=2.1\times 10^{-330}$
    while $\mbox{\sl pos}\ge 0$ do
       $\mbox{\sl step} :=$ sample uniform$(0, 1)$;
       if prob($0.5$) then
          $\mbox{\sl pos} := \mbox{\sl pos} - \mbox{\sl step}$
       else
          $\mbox{\sl pos} := \mbox{\sl pos} + \mbox{\sl step}$
       fi;    
       $\mbox{\sl dist} := \mbox{\sl dist} + \mbox{\sl step}$
     od;
     score($pdf$(normal($1.1$,$0.1$),$\mbox{\sl dist}$));
     $\color{blue} \mbox{polynomial approximation } g$
     return $\mbox{\sl start}$
$\lout:$
\end{lstlisting}
\caption{A Pedestrian Random Walk}
\label{fig:pedestrian-program}
\end{minipage}
\begin{minipage}{0.45\textwidth}
\begin{lstlisting}[basicstyle=\small,mathescape]
    $\mbox{\sl lambda} := $ sample uniform$(0,2)$;
    $\mbox{\sl time} := 10$; $\mbox{\sl amount} := 0$;
    $\lin:$ while $\mbox{\sl time}\ge 0$ do
       $\mbox{\sl wait} :=$ sample uniform$(0, 0.5)$;
       $\mbox{\sl time}:=\mbox{\sl time}-\mbox{\sl wait}$;
       if prob($0.5\cdot \mbox{\sl lambda}$) then
         $\mbox{\sl birth} :=$ sample uniform($0,0.01$);
         $\mbox{\sl amount} := \mbox{\sl amount} + \mbox{\sl birth}$;
         score($1.1$)
       fi;
     od;
     return $\mbox{\sl lambda}$
    $\lout:$
\end{lstlisting}
\caption{A Phylogenetic Birth Model}
\label{fig:phylogenetic}
\end{minipage}
\end{figure}

Consider the pedestrian random walk example~ \cite{DBLP:conf/esop/MakOPW21} in \cref{fig:pedestrian-program}. In this example, a pedestrian is lost on the way home, and she only knows that she is at most $3$ km away from her house. Thus, she starts to repeatedly walk a uniformly random distance of at most $1$ km in either direction of the road with equal probability, until reaching her house. Upon the arrival, an odometer tells that she has walked $1.1$ km in total. However, this odometer was once broken and the measured distance is normally distributed around the true distance with a standard deviation of $0.1$ km. We want to infer the posterior distribution of the starting point. 

This example is modeled as a non-parametric probabilistic program whose number of loop iterations is unbounded, where the blue part is the annotations used for the Bayesian inference of this example. In the program, the variables $start$, $pos$, $step$, $dis$ represent the starting point, the current position, the distance walked in the next step and the travelled distance so far of the pedestrian, respectively. 
As the program variables $start,step$ simply receive samples, the key program variables are $pos,dist$. 

This program terminates with probability $1$ and has a score statement at termination with a bounded score function indicated by the probability density function $pdf(\textbf{normal}(1.1,0.1),dist)$ of the normal distribution with mean $1.1$ and deviation $0.1$. Note that $pdf(\textbf{normal}(1.1,0.1),dist)=pdf(\textbf{normal}(dist,0.1),1.1)$. Bayesian probabilistic programs that have score statements at termination widely exist in the literature~\cite{DBLP:conf/cav/GehrMV16,DBLP:conf/pldi/GehrSV20,Beutner2022b}, and we call them \emph{score-at-end} programs. We propose a novel approach for Bayesian inference over such programs via fixed-point conditions. 

Our approach decomposes the Bayesian inference into the computation of expected weights, i.e., expected score values from various initial program inputs. For this example, we have the initial value for the variable $pos$ ranges over $[0,3]$, and the initial value for the variable $dist$ fixed to be $0$. We partition the range of initial values into multiple pieces (e.g., dividing $[0,3]$ into $[0,0.1], [0.1,0.2],\dots,[2.9,3]$), and solve the expected weights for each piece. Within each piece, we establish a polynomial template $h$ for the upper bound to be solved and use fixed-point theory to solve the template. 
A template for this example at the entry point of the loop
is given by the linear template $h(pos,dis)=a_1\cdot pos+a_2\cdot dis+a_3$ (see the annotations), and the prefixed-point condition to solve $h$ as an upper bound is given by
\begin{align*}
&\ewt(h):=0.5\cdot\expectdist{step}{[pos-step\ge 0]\cdot h(pos-step,dis+step)+[pos-step<0]\cdot g(dis+step)}\\ & \quad{}+ 0.5\cdot\expectdist{step}{h(pos+step,dis+step)}\le h(pos,dis).
\end{align*} 
where the left-hand-side $\ewt(h)$ of the inequality expresses the expected value of the template after one loop iteration, for which the expectation is taken w.r.t the sampling of the $step$ variable. 
Note that here we use a polynomial $g$ to approximate the density function of the normal distribution to allow a uniform polynomial reasoning. 

For this example, simply solving the polynomial template $h$ does not suffice, as we have find that polynomial solving produces trivial constant bounds even with high degrees. Hence, we consider a novel truncation operation that truncates the probabilistic program into a bounded range, so that we can utilize the strong approximation ability of polynomials over bounded ranges. In this example, we can choose the bounded range to be $B=\{(pos,dis)\mid pos\in [0,5], dis\in [0,5]\}$, so that the program state space is partitioned into sets of states within $B$ and outside $B$. The execution of the program is also changed in the sense that once the execution jumps out of the bounded range, the program halts immediately. In conjunction with the bounded range, we associate a polynomial $\calM$ that over-approximates the expected weights outside the bounded range. In this example, we have that when jumping out of the bounded range, either $dist\ge 5$ or $pos \ge 5$, and in both cases we have $dist\ge 5$ (as the pedestrian needs to travel at least the maximal $pos$ in the walk). Hence, we can choose $\calM=2.1\times 10^{-330}$ since $pdf(\textbf{normal}(1.1,0.1),dist)\le 2.1\times 10^{-330}$ when $dist\ge 5$ according to its monotonicity. 
Given the bounded range $B$ with the over-approximation $\calM$, we then solve the template $h$ by the following (informal) modified prefixed-point conditions: (a) When within the bounded range $B$, we have $\ewt(h) \le h$; (b) When outside the bounded range $B$, we have $\calM \le h$.  

The example is handled in~\cite{Beutner2022b} by exhaustive recursion unrolling that has the path-explosion problem. 
Our approach circumvents path explosion and derives comparable bounds to the approach in~\citet{Beutner2022b} with runtime two-thirds of that of~\citet{Beutner2022b}.

\subsection{Phylogenetic Birth Model}\label{sec3:phylogenetic}
	
Consider a simplified version of the phylogenetic birth model \cite{ronquist2021universal}, where a species arises with a birth-rate $lambda$, and it propagates with a simplified constant likelihood of $1.1$ at some time interval. For simplicity, we assume constant weights that can be viewed as over-approximation for a continuous density function. This example can be modelled as a probabilistic loop in~\cref{fig:phylogenetic}. In this program, the variables $lambda, time, amount, wait$ stand for the birth rate of the species, the remaining propagation time, the current amount of the species and the propagation time to be spent, respectively. The variable $lambda$ is associated with a prior distribution, and the NPD problem is to infer its posterior distribution given the species evolution described by the loop.

The main difficulty to analyze the NPD of this example is that its loop body includes a score statement {\tt score($1.1$)} with the constant score function greater than $1$. We call such programs \emph{score-recursive}. 
As stated previously, this incurs an integrability issue that cannot be solved by previous approaches. 
To address this difficulty, we propose a novel multiplicative variant of Optional Stopping Theorem (OST) that allows a stochastic process to scale by a multiplicative factor during its evolution. Based on the OST variant, we apply polynomial solving with truncation as in our fixed-point approach. 

 In this example, the key program variables are $lambda,time$ as $lambda$ affects the probabilistic branches and $time$ is included in the loop guard, while $wait,birth$ simply receive samples and $amount$ does not affect the control flow of the program.
We perform truncation operation by restricting the behaviour of the program within a bounded range such as $\{(lambda,time)\mid lambda\in [0,2]$ and $time\in [0,10]\}$, and over-approximate the expected weights outside the bounded range by an interval bound of polynomial functions 
derived from our OST variant and polynomial-template method but without truncation.
Our experimental result on this example 
shows that the derived bounds match the simulation result with $10^6$ samples.

\section{Preliminaries}\label{sec:prelim}
We first review basic concepts from probability theory, then present our Bayesian probabilistic programming language, and finally define the normalised posterior distribution (NPD) problem.
We denote by $\Nset$, $\Zset$ and $\Rset$ the sets of all natural numbers, integers, and real numbers, respectively.

\subsection{Basics of Probability Theory}

We recall several basic concepts and refer to standard textbooks (e.g.  \cite{pollard2002user,williams1991probability}) for details.  

Given a probability space $\pspace$, a \emph{random variable} is an $\mathcal{F}$-measurable function $X: \Omega \rightarrow \Rset \cup \{+\infty,-\infty\}$ where the measurable space over $\Rset \cup \{+\infty,-\infty\}$ is taken as the Borel space. 
The \emph{distribution function} $F$ of $X$ is given by $F(x) = \probm (\{\omega: X(\omega) \leq x\})$. A non-negative Borel measurable function $f$ is a \emph{density function} of $X$ if it satisfies $F(x) = \int_{-\infty}^{x} f(t) \mathrm{d} t$. 
The \emph{expectation} of a random variable $X$, denoted by $\expv(X)$, is the Lebesgue integral of $X$ w.r.t. $\probm$, i.e., $\int_{\Omega} X\,\mathrm{d}\probm$. A \emph{filtration} of $\pspace$ is an infinite sequence $\{ \mathcal{F}_n \}_{n=0}^{\infty}$ of $\sigma$-algebras such that for every $n\ge 0$, the triple $(\Omega, \mathcal{F}_n, \probm)$ is a probability space and $\mathcal{F}_n \subseteq \mathcal{F}_{n+1} \subseteq \mathcal{F}$. A \emph{stopping time} w.r.t. $\{ \mathcal{F}_n \}_{n=0}^{\infty}$ is a random variable $T: \Omega \rightarrow \Nset \cup \{0, \infty\}$ such that for every $n \geq 0$, the event \{$T \leq n$\} is in $\mathcal{F}_n$. 
Recall the Borel measurable space $(\Rset^n,\Sigma_{\Rset^n})$ where $\Sigma_{\Rset^n}$ is the $\sigma$-algebra generated by the open subsets in $\Rset^n$.

A \emph{discrete-time stochastic process} is a sequence $\Gamma = \{X_n\}_{n=0}^\infty$ of random variables in $\pspace$. The process $\Gamma$ is \emph{adapted} to a filtration $\{ \mathcal{F}_n \}_{n=0}^{\infty}$, if for all $n \geq 0$, $X_n$ is a random variable in $(\Omega, \mathcal{F}_n, \probm)$. A %discrete-time 
stochastic process $\Gamma=\{X_n\}_{n=0}^\infty$ adapted to a filtration $\{\mathcal{F}_n\}_{n=0}^\infty$ is a \emph{martingale} (resp., \emph{supermartingale}, \emph{submartingale})
if for all $n \geq 0$, $\expv(|X_n|)<\infty$ and it holds almost surely 
that $\condexpv{X_{n+1}}{\mathcal{F}_n}=X_n$ (\mbox{resp., } $\condexpv{X_{n+1}}{\mathcal{F}_n}\le X_n$, $\condexpv{X_{n+1}}{\mathcal{F}_n}\ge X_n$).
Applying martingales to the formal analysis of probabilistic programs is a well-studied technique~\cite{SriramCAV,ChatterjeeFG16,ChatterjeeNZ2017}.

\subsection{Bayesian Probabilistic Programs}

The syntax of our Bayesian probabilistic programming language (PPL) is given in \cref{fig:syntax}, where $c,c_1,c_2\in\Rset$ are real constants, $p\in (0,1]$ and the metavariables $S$, $B$ and $E$ stand for statements, boolean and arithmetic expressions, respectively.  
Our PPL is imperative with the usual conditional, loop, sequential and probabilistic branching structures as well as the following new structures: (a)~sample constructs of the form ``$\textbf{sample}\  D$'' that sample a value from a prescribed distribution $D$ (e.g., normal distribution, uniform distribution, etc.) over $\mathbb{R}$; (b)~score statements of the form ``\textbf{score}($EW$)'' that weight the current execution with a value expressed by $EW$, where $\textit{pdf}(D,x)$ is the value of the probability density function w.r.t. the distribution $D$ at $x$.
We also have return statements (i.e., \textbf{return}) that return the value of a program variable. Note that although probabilistic branches can be derived from sampling of Bernoulli distributions, we include probabilistic branches here to have specific algorithmic treatment for probabilistic control flows.

\begin{figure}
%\vspace{-5ex}
	\footnotesize{
		\begin{align*}
		S &::= \textbf{skip} \mid x:=ES  \mid \textbf{score}(EW) \mid \textbf{return}\ x \mid S_1;S_2 \\
		& \mid \mbox{\textbf{while}}\, B \, \text{\textbf{do}} \, S \, \text{\textbf{od}} \mid \mbox{\textbf{if}} \, \mbox{$B$}\,\mbox{\textbf{then}} \,  S_1 \, \mbox{\textbf{else}} \,S_2 \,\mbox{\textbf{fi}}\mid \mbox{\textbf{if}} \, \mbox{\textbf{prob}($p$)}\,\mbox{\textbf{then}} \,  S_1 \, \mbox{\textbf{else}} \,S_2 \,\mbox{\textbf{fi}} \\
		B&::=\textbf{true}\mid \textbf{false} \mid\neg B\mid B_1\, \textbf{and} \, B_2 \mid  B_1\, \textbf{or} \, B_2 \mid  E_1\le E_2\mid E_1\ge E_2\\
        E&::= x\mid c\mid E_1+E_2\mid E_1-E_2\mid E_1*E_2  \ \ \ 
		ES::= E\mid  \textbf{sample}\  D \\
        EW&::= E\mid  \textit{pdf}(D,x) \ \ \ 
		D::=  \textbf{normal}(c_1,c_2)\mid \textbf{uniform}(c_1,c_2)\mid \cdots 
		\end{align*}
	}
	\caption{Syntax of Our Probabilistic Programming Language}
	\label{fig:syntax}
 %\vspace{-4ex}
\end{figure}

In our PPL, we distinguish two disjoint sets of variables in a program: (i) the set $\pvars$ of \emph{program variables} whose values are determined by assignments (i.e., the expressions at the RHS of ``:="); (ii)~the set $\rvars$ of \emph{sampling variables} whose values are independently sampled from prescribed probability distributions each time they are accessed (i.e., each ``$\textbf{sample}\ D$" is a sampling variable). 
A \emph{valuation} on a set $V$ of variables is a function $\pv: V \rightarrow \Rset$ that assigns a real value to each variable in $V$. 
A \emph{program} (resp. \emph{sampling}) valuation is a valuation on $\pvars$ (resp. $\rvars$).
The set of program (resp. \emph{sampling}) valuations is denoted by $\val{\mathrm{p}}$ (resp. $\val{\mathrm{r}}$), respectively.
For the sake of convenience, we fix the notations in the following way: we always use $\pv\in\val{\mathrm{p}}$ to denote a program valuation, and $\rv\in\val{\mathrm{r}}$ to denote a sampling valuation.

\begin{example}\label{ex:pedestrian-program}
\cref{fig:pedestrian-program} shows a Bayesian probabilistic program written in our PPL. In this program, the set of program variables is $\pvars=\{start,pos,dis,step\}$, and the set of sampling variables is $\rvars=\{ \textbf{sample uniform}(0,3), \textbf{sample uniform}(0,1) \}$. At the execution of $\textbf{sample uniform}(0,3)$,  
it samples a value uniformly from $[0, 3]$ and assigns it to the variable $start$ in the initialization. During the loop iteration, each time $\textbf{sample uniform}(0,1)$ is executed, it samples a value uniformly from $[0,1]$ and assigns the value to the variable $step$. 
\qed

\begin{figure}
\lstset{language=prog}
\lstset{linewidth=5.1cm}
\begin{minipage}{0.45\textwidth}
	\begin{center}
\includegraphics[width=0.78\textwidth]{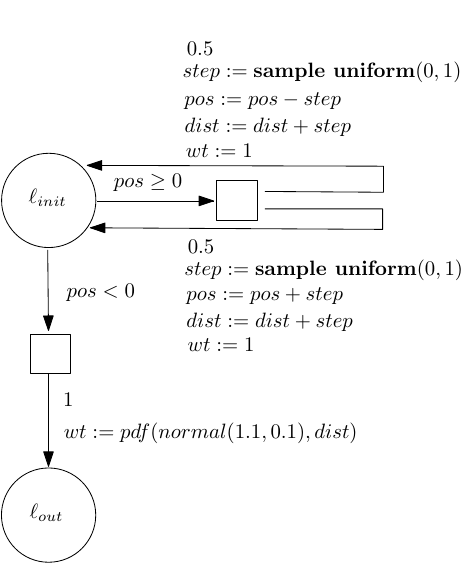}
\end{center}
\caption{The WPTS of Pedestrian}
\label{fig:pedestrian-wpts}
\end{minipage}
\begin{minipage}{0.45\textwidth}
	\begin{center}
\includegraphics[width=0.8\textwidth]{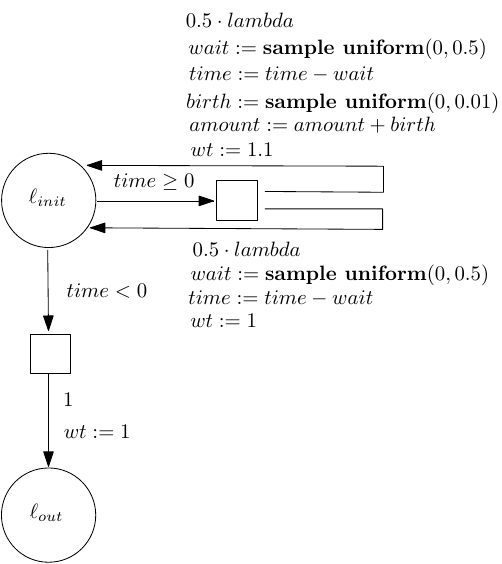}
\end{center}
\caption{The WPTS of Phylogenetic Model}
\label{fig:phylogenetic-wpts}
\end{minipage}
\end{figure}
\end{example}

Below we present the semantics for our PPL. In the literature, existing semantics are either measure-based~\cite{DBLP:conf/lics/StatonYWHK16,LeeYRY20} or sampling-based~\cite{DBLP:conf/esop/MakOPW21,Beutner2022b}. To facilitate the development of our approach, we consider the \emph{transition-based} semantics~\cite{DBLP:conf/cav/ChakarovS13,DBLP:conf/popl/ChatterjeeFNH16} so that  
each probabilistic program is transformed into an equivalent form of \emph{weighted probabilistic transition system} (WPTS). A WPTS extends a PTS  ~\cite{DBLP:conf/cav/ChakarovS13,DBLP:conf/popl/ChatterjeeFNH16} with weights and an initial probability distribution. 

\begin{definition}[WPTS]\label{def:wpts}
	A \emph{weighted probabilistic transition system} (WPTS) $\Pi$
	is a tuple
\begin{equation}\label{eq:wpts} 
\tag{\dag}
\Pi = (\pvars, \rvars,  L,\lin,\lout,\mu_{\mathrm{init}}, \rdvarjdis,\transset)%\win)
\end{equation}
for which:
	\begin{itemize}
		\item
		$\pvars$ and $\rvars$ are finite disjoint sets of \emph{program} and \emph{sampling} variables.
    \item $\locs$ is a finite set of \emph{locations} 
  with special locations $\lin,\lout\in \locs$. Informally, a location corresponds to a cut point in a Bayesian probabilistic program, $\lin$ is the initial location and $\lout$ represents program termination. 
		\item
		$\mu_{\mathrm{init}}$ is the \emph{initial probability distribution} over $\mathbb{R}^{|\pvars|}$ with a bounded support (denoted by $\supp{\mu_{\mathrm{init}}}$), 
  We call each $\pv\in\supp{\mu_{\mathrm{init}}}$ an \emph{initial program valuation}.

        \item $\rdvarjdis$ is a function that assigns a probability distribution $\rdvarjdis(r)$ to each 
  $r \in \rvars$. We abuse the notation so that $\rdvarjdis$ also denotes the joint distribution of all independent variables $r\in \rvars$. 
		\item 
		$\transset$ is a finite set of \emph{transitions} where
		each transition $\tau \in \transset$ is a tuple $\langle \loc, \phi, F_1,\dots,F_k \rangle$ such that (a) $\loc\in L$ is the \emph{source location}, 
(b) $\phi$ is the \emph{guard condition} which is a logical formula over program variables $\pvars$, 
and (c) each $F_j:=\langle \loc'_j, p_j, \upd_j,\wet_j \rangle$ is called a \emph{weighted fork} for which (i) $\loc'_j\in L$ is the \emph{destination location} of the fork, (ii) $p_j\in (0,1]$ is the probability of occurrence of this fork, (iii) $\upd_j:\Rset^{|\pvars|} \times \Rset^{|\rvars|} \rightarrow \Rset^{|\pvars|}$ is an {\em update function} that takes as inputs the current program and sampling valuations and returns an updated program valuation, and (iv) $\wet_j:\Rset^{|\pvars|} \times \Rset^{|\rvars|}\to [0,\infty)$ is a \emph{score function} that gives the likelihood weight of this fork depending on the current program and sampling valuations.	
\end{itemize}
\end{definition}

In a WPTS, update functions correspond to assignment statements to program variables, and score functions correspond to the cumulative multiplicative weight of a basic block of statements from the score statements in the block. 
Note that if there is no score statement in the block, then the score function of the block is constantly $1$. 
We also assume the initial probability distribution to have \emph{bounded} support as our approach solves optimization problems over bounded sets of initial program valuations.

We always assume that a WPTS $\Pi$ is \emph{deterministic} and \emph{total}, i.e., (1) there is no program valuation that simultaneously satisfies the guard conditions of two distinct transitions from the same source location, and (2) the disjunction of the guard conditions of all the transitions from any source location is a tautology. 
The transformation from a probabilistic program into its WPTS can be done in a straightforward way (see e.g.~\cite{DBLP:journals/toplas/ChatterjeeFNH18,DBLP:conf/cav/ChakarovS13}).

\begin{example}\label{ex:pedestrian-semantics} 
\cref{fig:pedestrian-wpts} shows the WPTS of the program in \cref{fig:pedestrian-program} which has two locations $\lin,\lout$. 
The value of $step$ is initialised to $0$. The initial probability distribution $\mu_{\mathrm{init}}$ over $\Rset^{|\pvars|}$ is determined by the joint distribution of $(start,pos,dis,step)$ where $start\sim uniform(0,3)$ and $pos,dis,step$ observe the Dirac measures $Dirac(\{start\})$, $Dirac(\{0\})$ and $Dirac(\{0\})$, respectively, e.g., the probability of the event ``$dis\in\{0\}$'' equals $1$. 
The circle nodes represent locations and square nodes model the forking behavior of transitions. An edge entering a square node is labeled with the guard condition of its respective transition, while an edge entering a circle node stands for a fork, which is associated with its probability, update functions and score functions. The WPTS of the program in \cref{fig:phylogenetic} is analogously given in \cref{fig:phylogenetic-wpts}. 
\footnote{Here we omit the update functions if the values of program variables are unchanged.} 
\qed
\end{example}

Below we specify the semantics of a WPTS. Consider a WPTS $\Pi$ in the form of \eqref{eq:wpts}. Given a program valuation $\mathbf{v}$ and a guard condition $\phi$ over variables $\pvars$, we say that $\mathbf{v}$ \emph{satisfies} $\phi$ (written as $\mathbf{v}\models\phi$) if $\phi$ holds when the variables in $\phi$ are substituted by their values in $\mathbf{v}$. 

A \emph{state} is a pair $\Xi=(\loc, \pv)$ where $\loc \in L$ (resp. $\pv \in \Rset^{|\pvars|}$) represents the current location (resp. program valuation), respectively, while a \emph{weighted state} is a triple $\Theta=(\loc, \pv, w)$ where $(\loc, \pv)$ is a state and $w\in [0,\infty)$ represents the cumulative multiplicative likelihood weight.

The semantics of $\Pi$ is formalized by the infinite sequence $\Gamma=\{\widehat{\Theta}_n=(\widehat{\loc}_n,\widehat{\pv}_n,\widehat{w}_n)\}_{n\ge 0}$ where each $(\widehat{\loc}_n,\widehat{\pv}_n,\widehat{w}_n)$ is the random weighted state at the $n$-th execution step of the WPTS such that $\widehat{\loc}_n$ (resp. $\widehat{\pv}_n$, $\widehat{w}_n$) is the random variable for the location (resp. the random program valuation, the random variable for the multiplicative likelihood weight) at the $n$-th step, respectively. 
The sequence $\Gamma$ starts with the initial random weighted state 
$\widehat{\Theta}_0=(\widehat{\loc}_0,\widehat{\pv}_0,\widehat{w}_0)$ such that $\widehat{\loc}_0$ is constantly $\lin$, $\widehat{\pv}_0\in \supp{\mu_\mathrm{init}}$ is sampled from the initial distribution $\mu_\mathrm{init}$ and the initial weight $\widehat{w}_0$ is constantly set to $1$.\footnote{This follows the traditional setting in e.g.~\cite{Beutner2022b}.} 
Then, given the current random weighted state $\widehat{\Theta}_n=(\widehat{\loc}_n,\widehat{\pv}_n,\widehat{w}_n)$ at the $n$-th step, the next random weighted state $\widehat{\Theta}_{n+1}=(\widehat{\loc}_{n+1},\widehat{\pv}_{n+1},\widehat{w}_{n+1})$ is determined by:
(a) If $\widehat{\loc}_n=\lout$, then $(\widehat{\loc}_{n+1}, \widehat{\pv}_{n+1},\widehat{w}_{n+1})$ takes the same weighted state as $(\widehat{\loc}_n,\widehat{\pv}_n,\widehat{w}_n)$ (i.e., the next weighted state stays at the termination location $\lout$);
(b) Otherwise, $\widehat{\Theta}_{n+1}$ is determined by the following procedure:
\begin{itemize}
\item First, since the WPTS $\Pi$ is deterministic and total, we take the unique transition $\tau=\langle \hat{\loc}_n,\phi,F_1,\dots, F_k \rangle$ such that $\hat{\pv}_n\models\phi$. 
\item Second, we choose a fork $F_j=\langle \loc'_j, p_j,\upd_j,\wet_j\rangle$ with the probability $p_j$.
\item Third, we obtain a sampling valuation $\rv\in \supp{\rdvarjdis}$ by sampling each $r \in \rvars$ independently from the probability distribution $\rdvarjdis(r).$
\item Finally, the value of the next random weighted state $(\widehat{\loc}_{n+1}, \widehat{\pv}_{n+1},\widehat{w}_{n+1})$ is determined as that of 
$(\loc'_j, \upd_j(\hat{\pv}_n,\rv),\widehat{w}_n\cdot \wet_j(\widehat{\pv}_n,\rv))$. 
Note that the weight is obtained in the style of the multiplicative score.
\end{itemize}
Unlike several semantics (such as~\citet{DBLP:conf/lics/StatonYWHK16}) that integrates score statements with sampling, our semantics separates them and have a special construct $\textbf{score}(-)$ for score statements.

Given the semantics, a \emph{program run} of the WPTS $\Pi$ is a concrete instance of $\Gamma$, i.e., an infinite sequence $\omega=\{\Theta_n\}_{n\ge 0}$ of weighted states where each $\Theta_n=(\loc_n,\pv_n,w_n)$ is the concrete weighted state at the $n$-th step in this program run with location $\loc_n$, program valuation $\pv_n$ and cumulative multiplicative likelihood weight $w_n$. A state $(\loc,\pv)$ is called \emph{reachable} if there exists a program run $\omega=\{\Theta_n\}_{n\ge 0}$ such that $\Theta_n=(\loc,\pv,w)$ for some $n$ and $w$.

\begin{example}\label{ex:pedestrian-run}
Consider the WPTS in \cref{ex:pedestrian-semantics}. 
Suppose the initial program valuation is $(1,1,0,0)$ which means that the initial values of $start,pos,dis,step$ are $1,1,0,0$, respectively. Then starting from the initial weighted state $(\lin,(1,1,0),1)$, a program run 
could be 
\small{
\[
(\lin,(1,1,0,0),1)\to (\lin,(1,0.5,0.5,0.5),1)\to (\lin,(1,-0.1,1.1,0.6),1)\to (\lout,(1,-0.1,1.1,0.6),3.9894).
\]
}
\noindent

  After the final execution, the program valuation becomes $(1, -0.1, 1.1,0.6)$ and the loop terminates. We capture the likelihood weight with $\textit{pdf}(\textbf{normal}(1.1,0.1),1.1) = \frac{1}{\sqrt{2\pi} \times 0.1} e^{-\frac{0}{2 \times 0.1^2}} = 3.9894$ and multiply it to the  current weight (i.e., 1). Thus the final weight of this run is $3.9894$. \qed
\end{example}

Given an initial program valuation $\valin$ of a WPTS, one could construct a probability space over the program runs via their probabilistic execution described above and standard constructions such as general state space Markov chains~\cite{meyn2012markov}. We denote the probability measure in this probability space by $\probm_{\valin}(-)$ and the expectation operator by $\expectdist{\valin}{-}$.

\subsection{Normalised Posterior Distribution}\label{sec2:NPD}

Below we fix a WPTS $\Pi$ in the form of \eqref{eq:wpts}.
The \emph{termination time} of the WPTS $\Pi$ is the random variable $T$ given by
$T(\omega):=\text{min}\{n\in\Nset\mid \loc_n=\lout\}$ for every program run  $\omega=\{(\loc_n,\pv_n,w_n)\}_{n\ge 0}$
where $\text{min}\,\emptyset:=\infty$. That is, $T(\omega)$ is the number of steps a program run $\omega$ takes to reach the termination location $\lout$. The WPTS $\Pi$ is \emph{almost-surely terminating} (AST) if $\probm_{\valin}(T<\infty)=1$ for all initial program valuations $\valin\in \supp{\mu_{\mathrm{init}}}$.  
Given a designated initial program valuation $\valin$ and a measurable subset $\calU\in\Sigma_{\Rset^{|\pvars|}}$, the \emph{expected weight} $\measureSem{\Pi}_{\valin}(\calU)$ restricted to $\calU$ is defined as $\measureSem{\Pi}_{\valin}(\calU):=\expectdist{\valin}{[\widehat{\pv}_T\in \calU]\cdot\widehat{w}_T}$ where $[-]$ is the Iverson bracket such that $[\phi]=1$ if $\phi$ holds and $[\phi]=0$ otherwise. 
(Recall that $\widehat{\pv}_T$ and $\widehat{w}_T$ are the random vector and variable of the program valuation and the multiplicative likelihood weight at termination, respectively.) 
If $\calU=\Rset^{|\pvars|}$, then $\measureSem{\Pi}_{\valin}(\Rset^{|\pvars|})$ is called the \emph{unrestricted expected weight}, or simply \emph{expected weight} for short.
The normalised posterior distribution (NPD) is defined as follows.  
\begin{definition}[NPD]\label{def:npd}
The \emph{normalised posterior distribution} (NPD) $\posterior_\Pi$ of $\Pi$ is defined by:
\begin{align*}
\posterior_{\Pi}(\calU):=\measureSem{\Pi}(\calU)/Z_\Pi\mbox{ for all measurable subsets } \calU\in \Sigma_{\Rset^{|\pvars|}},   
\end{align*}	
where $\measureSem{\Pi}(\calU):=\int_{\calV} \measureSem{\Pi}_{\pv}(\calU)\cdot \mu_{\mathrm{init}}(\mathrm{d} \pv)$ is the \emph{unnormalised posterior distribution} w.r.t. $\calU$ with $\calV:=\supp{\mu_{\mathrm{init}}}$, and $Z_\Pi:=\measureSem{\Pi}(\Rset^{|\pvars|})$ is the \emph{normalising constant}.  
$\Pi$ is \emph{integrable} if $0<Z_{\Pi}<\infty$. 

\end{definition}

In this work, we consider the automated interval bound analysis for the NPD of a WPTS. Formally, we aim to derive a tight interval $[l,u]\subseteq [0,\infty)$ for an integrable WPTS $\Pi$ and any measurable set $\calU\in\Sigma_{\Rset^{|\pvars|}}$ such that $l\le \posterior_{\Pi}(\calU) \le u$. To achieve this, we consider bounds on $\measureSem{\Pi}(\calU),Z_\Pi$.

\begin{framed}{\textbf{NPD Bounds}.}
Assume we have two intervals $[l_\calU,u_\calU],[l_Z,u_Z]\subseteq [0,\infty)$ such that the unnormalised posterior distribution $\llbracket \Pi\rrbracket (\calU)\in [l_\calU,u_\calU]$ and the normalising constant $Z_\Pi\in [l_Z,u_Z]$. If $\Pi$ is integrable, 
then we have the NPD $\posterior_{\Pi}(\calU)\in [\frac{l_\calU}{u_Z},\frac{u_\calU}{l_Z}]$.
\end{framed}

To analyze the quantity $\measureSem{\Pi}(\calU)$ derived from expected weights restricted to $\calU$, 
we construct a new WPTS $\Pi_\calU$ from the original $\Pi$ and a measurable set $\calU\in \Sigma_{\Rset^{|\pvars|}}$.  
Consider a probabilistic program $P$ and its WPTS $\Pi$, given a measurable set $\calU\in\Sigma_{\Rset^{|\pvars|}}$, we construct a new program $P_\calU$ by adding a conditional branch of the form ``\textbf{if} $\pv_T\notin\calU$ \textbf{then} \textbf{score}($0$) \textbf{fi}'' immediately after the termination of $P$ and obtain the WPTS $\Pi_\calU$ of $P_\calU$. Thus, $\Pi_\calU$ simply restricts the program valuation at termination to $\calU$ when taking the final cumulative weight. 
In this way, we transform the analysis of $\measureSem{\Pi}(\calU)$ into that of 
$\measureSem{\Pi_\calU}(\Rset^{|\pvars|})$ in the following, since $\measureSem{\Pi}(\calU)=\measureSem{\Pi_\calU}(\Rset^{|\pvars|})$. Therefore, from~\cref{def:npd}, we can reduce the bound analysis of NPD to that of unrestricted expected weights in the rest of the paper. See details in \cref{app:sec2-prop}.

\section{Theoretical Approaches}\label{sec:math}

Below we present two approaches for deriving bounds of unrestricted expected weights over Bayesian probabilistic programs. 

\subsection{The Fixed-Point Approach}\label{sec:fixed-point}

Our fixed-point approach targets \emph{score-at-end} Bayesian programs that terminate almost-surely and have a single score statement with a bounded weight at the termination. 
We formally define the notion of score-at-end programs directly over WPTS's. A WPTS $\Pi$ is \emph{score-at-end} if: (i) $\Pi$ has AST; (ii) there is exactly one transition $\tau$ in the WPTS that involves the destination location $\lout$, and this transition $\tau$ takes the form $\tau=\langle \loc, \mathbf{-}, F\rangle$ with the single weighted fork $F=\langle \lout, 1, \mbox{\sl id}, wt\rangle$ where $\mbox{\sl id}$ is the identity function and $wt$ is \emph{score-bounded} by a constant $M>0$, i.e., $wt\in [0, M]$; and (iii) all transitions other than the transition $\tau$ have the score function constantly equal to $1$. Such programs widely exist in the literature~\cite{DBLP:conf/cav/GehrMV16,DBLP:conf/pldi/GehrSV20,Beutner2022b}.

To demonstrate the approach, we recall several basic concepts in lattice theory (see Appendix~\ref{app:fixed-point-materials} for details). 
Given a complete lattice $(K, \sle)$ and a function $f: K \to K$, the \emph{supremum} of $K$ is denoted by $\bigsqcup K$, while the \emph{infimum} of $K$ is denoted by $\bigsqcap K$. An element $k \in K$ is called a \emph{fixed-point} if $f(k) = k.$ The \emph{least} (resp. \emph{greatest}) \emph{fixed-point}  of $f$, denoted by $\lfp f$ (resp. $\gfp f$),  is the fixed-point that is no greater (resp. smaller) than every fixed-point under $\sle$. 
Moreover, $k$ is a \emph{prefixed-point} if $f(k) \sle k$ and a \emph{postfixed-point} if $k\sle f(k)$. 
We apply Tarski's fixed-point theorem as follows. 

\begin{theorem}[\textit{Tarski}~\cite{KnasterTarski}]\label{thm:tarski}
Let $(K, \sle)$ be a complete lattice and $f:K \to K$ be a monotone function. Then, both $\lfp\ f$ and $\gfp\ f$ exist. Moreover,
$\textstyle \lfp\ f  = \mathop{\bigsqcap} \left\{x\ |\ f(x)\sle x\right\}\mbox{ and }\gfp\ f = \mathop{\bigsqcup} \left\{x\ |\ x\sle f(x)\right\}$. 
\end{theorem}

Based on \cref{thm:tarski}, we present our fixed-point approach for score-at-end WPTS's. Below we fix a score-at-end WPTS $\Pi$. 
We define $\Lambda$ as the set of states $(\loc, \pv)$ where $\loc$ is a location and $\pv$ is a program valuation. Given a maximum finite value $M\in [0,\infty)$, we define a \emph{state function} as a function $h:\Lambda\to [-M,M]$ such that for all $\pv\in\Rset^{|\pvars|}$, $h(\lout,\pv)\in [0,M]$.
The intuition of a state function $h$ is that $h(\loc,\pv)$ gives an estimation of the expected weight when the WPTS $\Pi$ starts with the initial location $\loc$ and the initial program valuation $\pv$.
We denote the set of all state functions with the maximum value $M$ by $\mathcal{K}_M$. We also use the usual partial order $\le$ on $\mathcal{K}_M$ that is defined in the pointwise fashion, i.e., for any $h_1,h_2\in \mathcal{K}_M$, $h_1\le h_2$ iff $h_1(\loc,\pv)\le h_2(\loc,\pv)$ for all $(\loc,\pv)\in\Lambda$. It is straightforward to verify that  $(\mathcal{K}_M,\le)$ is a complete lattice. 

To connect the complete lattice $(\mathcal{K}_M,\le)$ with expected weights, we define the \emph{expected-weight function} $\mbox{\sl ew}_\Pi$ by $\mbox{\sl ew}_\Pi(\lin,\pv):=\llbracket \Pi\rrbracket_{\pv}(\Rset^{|\pvars|})$, and omit the subscript $\Pi$ if it is clear from the context. Informally, $\mbox{\sl ew}_\Pi(\lin,\pv)$ is the expected weight of the program starting from an initial program valuation $\pv$. 
Note that if the weight at termination is bounded by a maximum value $M$, then we have $\mbox{\sl ew}_\Pi\in \mathcal{K}_M$. 
We consider the following higher-order function over the complete lattice $(\mathcal{K}_M,\le)$. 

\begin{definition}[Expected-Weight Transformer]\label{def:ewt}
Given a finite maximum value $M\in [1,\infty)$, the \emph{expected-weight transformer} $\ewt_\Pi:\mathcal{K}_M\to \mathcal{K}_M$ is the higher-order function such that for each state function $h\in \mathcal{K}_M$ and state $(\loc,\pv)$, if  $\tau = \langle \loc, \phi, F_1,\dots,F_k \rangle$ is the unique transition that satisfies $\pv\models\phi$ and $F_j=\langle \loc'_j,p_j,\upd_j,\wet_j\rangle$ for each $1\le j\le k$, then we have that
\begin{equation}\label{eq:ewt}
\ewt_\Pi(h)(\loc,\pv)\,:=\, \begin{cases}  \sum_{j=1}^k p_j\cdot \expectdist{\rv}{\wet_j(\pv,\rv)\cdot h(\loc'_j,\upd_j(\pv,\rv))} & \mbox{if } \loc\neq \lout \\
1 & \mbox{otherwise} 
\end{cases}\enskip. 
\end{equation}
In \eqref{eq:ewt}, the expectation $\expectdist{\rv}{-}$ is taken over a sampling valuation $\rv$ that observes the distribution $\rdvarjdis$.
\end{definition}
Informally, given a state function $h$, the expected-weight transformer $\ewt_\Pi$ computes the expected weight $\ewt_\Pi(h)$ after one step of WPTS transition. Note that in \eqref{eq:ewt}, the weight $\wet_j(\pv,\rv)$ equals $1$
when the location $\loc$ does not refer to the score statement at the end of the program, as we consider that the WPTS $\Pi$ is score-at-end. This implies that $\ewt_\Pi$ is indeed a higher-order operator for the complete lattice $(\mathcal{K}_M,\le)$ when the maximum value $M$ is a bound for the weights in the score statement. By the monotonicity of expectation, we have that $\ewt_\Pi$ is monotone. 

We will omit the subscript $\Pi$ in $\ewt_\Pi(h)$ if it is clear from the context. In the next definition, we define potential weight functions that are prefixed/postfixed points of the operator $\ewt_\Pi(h)$.

\begin{definition}[Potential Weight Functions]\label{def:puwf}
    A \emph{potential upper weight function} (PUWF) is a function $h:\locs{}\times\val{V_p} \rightarrow\Rset$ that has the following properties:
	\begin{itemize}
		\item[\emph{(C1)}] for all reachable states $(\loc,\pv)$ with $\loc\neq\lout$, we have $\ewt({h})(\loc,\pv) \le h(\loc,\pv)$;
		\item[\emph{(C2)}] for all reachable states $(\loc,\pv)$ such that $\loc=\lout$, we have $h(\loc,\pv)=1$.
	\end{itemize}
	Analogously, a \emph{potential lower weight function} (PLWF) is a function $h:\locs{}\times\val{V_p} \rightarrow\Rset$
	that satisfies the conditions (C1') and (C2), for which the condition (C1') is almost the same as (C1) except for that ``$\ewt({h})(\loc,\pv) \le h(\loc,\pv)$'' is replaced with ``$\ewt({h})(\loc,\pv) \ge h(\loc,\pv)$''. 
\end{definition}

Informally, a PUWF is a state function that satisfies the prefixed-point condition of $\ewt$ at non-terminating locations, and equals one at termination. A PLWF is defined similarly by using the postfixed-point condition. 
The main theorem below states that potential weight functions serve as bounds for expected weights. The proof establishes that the fixed point of $\ewt_\Pi$ is unique and equals $\mbox{\sl ew}_\Pi$ 
when the WPTS $\Pi$ has AST, and applies Tarski's Fixed Point Theorem to use prefixed and post-fixed points to derive upper and lower bounds. See Appendix~\ref{app:fixedpoint} for the detailed proof.

\begin{theorem}[Fixed-Point Approach]\label{thm:fix-point-bounds}

$\llbracket \Pi\rrbracket_{\valin} (\Rset^{|\pvars|})\le h(\lin,\valin)$ (resp. $\llbracket \Pi\rrbracket_{\valin} (\Rset^{|\pvars|})\ge h(\lin,\valin)$) for any bounded PUWF (resp. PLWF) $h$ over $\Pi$ and initial state $(\lin,\valin)$.
\end{theorem}

\begin{remark}%[Generality of our fixed-point approach]
Although our fixed-point approach targets score-at-end programs, it can be extended to Bayesian programs such that in any execution of the program, only a bounded number of {\tt score}  statements are executed and they are all bounded. 
This is because one can find a bound for the overall weight of the boundedly many score functions to apply our fixed point theorem.  The intuition here is that our fixed point approach can handle Bayesian programs whose multiplicative cumulative score values are always bounded by some constant. Especially, this approach can handle programs where all execution cycles are score-bounded by one. \qed
\end{remark}

\subsection{The OST Approach}\label{sec:ostapproach}

Below we propose a novel approach to address the NPD problem of Bayesian programs that have score statements with weights greater than $1$ inside loop bodies. We refer to such programs \emph{score-recursive}, and they have received  significant attention recently (such as the phylogenetic birth model in Section~\ref{sec3:phylogenetic}). 
We say that a WPTS is \emph{score-recursive} if it has a cycle of transitions on which there is a score statement with score function taking a value greater than $1$. Our fixed-point approach requires bounded score, and therefore cannot handle all score-recursive programs. 

We first derive a novel multiplicative variant of the classical Optional Stopping Theorem (OST) that tackles the multiplicative feature of score statements, then applies our OST variant to potential weight functions (\cref{def:puwf}) to address the NPD problem. The classical OST requires bounded changes of the weight value, and thus cannot handle score-recursive programs.  

Our OST variant applies to bounded-update score-recursive WPTS's. Informally, a WPTS has the bounded-update property if the change of values of the program variables is bounded by a global constant for every transition in the WPTS. This property is fulfilled in many realistic probabilistic models as the change of value of a variable in a single step is often bounded.
Formally, a WPTS $\Pi$ has the \emph{bounded-update} property if there exists a real constant $\varkappa>0$ such that for every reachable state $(\loc,\pv)$ and fork $F_j=\langle \loc'_j,p_j,\upd_j,\wet_j \rangle$ from a transition with the source location $\loc$, we have that $\forall \rv\in\supp{\rdvarjdis}\,~\forall x\in \pvars,~~ |\upd_j(\pv,\rv)(x)-\pv(x)|\le \varkappa$. The OST variant is given as follows.

\begin{theorem}[OST Variant]\label{thm:ost-variant}
Let $\{X_n\}_{n=0}^\infty$ be a supermartingale adapted to a filtration 
$\mathcal{F}=\{\mathcal{F}_n\}_{n=0}^\infty$, and $\kappa$ be a stopping time w.r.t. the filtration $\mathcal{F}$. 
Suppose that there exist positive real numbers $b_1,b_2,c_1,c_2,c_3$ such that $c_2>c_3$ and
\begin{itemize}
\item[(A1)] $\probm(\kappa>n) \leq c_1 \cdot e^{-c_2 \cdot n}$ for sufficiently large $n \in \Nset$, and
\item[(A2)] for all $n \in \Nset$, $\left\vert X_{n+1}-X_n \right\vert \le b_1\cdot n^{b_2}\cdot e^{c_3\cdot n}$ holds almost surely. 
\end{itemize}
Then we have that 
$\expv\left(|X_\kappa|\right)<\infty$ and %$\expv\left(X_\kappa\right) = \expv(X_0)$ 
$\expv\left(X_\kappa\right)\le\expv(X_0)$.
\end{theorem}

Our OST variant relaxes the classical OST that we allow the next random variable $X_{n+1}$ to be bounded by that of $X_n$ with a multiplicative factor $e^{c_3}$. The intuition is to  
cancel the multiplicative factor with the exponential decrease in $\probm(\kappa>n) \leq c_1 \cdot e^{-c_2 \cdot n}$.
We note that the exponential decrease is essential to cancel multiplicative scaling in score statements, as shown in \emph{challenges and gaps} in~\cref{sec:intro}.  
The proof resembles \cite[Theorem 5.2]{cost2019wang} and is relegated to \cref{app:ost-variant-proof}.

Below we show how our OST variant can be applied to handle bounded-update score-recursive programs. Fix a bounded-update score-recursive WPTS $\Pi$ in the form of \eqref{eq:wpts}. We reuse the expected-weight transformer defined in  \cref{def:ewt} and  potential weight functions given in \cref{def:puwf}. The difference with our fixed-point approach is that we no longer require that the score function $\wet_j$ in \eqref{eq:ewt} equals one for locations that do not lead to termination.

\begin{theorem}[OST Approach]\label{thm:puwf-normalizing}
Let $\Pi$ be a bounded-update score-recursive WPTS. 
Suppose that there exist real numbers $c_1>0$ and $c_2>c_3>0$ such that 
\begin{itemize}
\item[(E1)] $\probm(T>n) \leq c_1 \cdot e^{-c_2 \cdot n}$ for sufficiently large $n\in\Nset$, and 
\item[(E2)] for each score function $\wet$ in $\Pi$, we have $|\wet|\le e^{c_3}$. 
\end{itemize}
 Then for any polynomial PUWF (resp. PLWF) $h$ over $\Pi$, we have that \  $\llbracket \Pi\rrbracket_{\valin} (\Rset^{|\pvars|})\le h(\lin,\valin)$ (resp. $\llbracket \Pi\rrbracket_{\valin} (\Rset^{|\pvars|})\ge h(\lin,\valin)$) for any initial state $(\lin,\valin)$, respectively.  
\end{theorem}

\begin{proof}[Proof Sketch.]
For upper bounds, we define the stochastic process $\{X_n\}_{n=0}^\infty$ as $X_n:=h(\loc_n,\pv_n)$ where $(\loc_n,\pv_n)$ is the program state at the $n$-th step of a program run. Then we construct another stochastic process $\{Y_n\}_{n=0}^\infty$ such that $Y_n:=X_n\cdot \prod_{i=0}^{n-1} W_i$ where $W_i$ is the weight at the $i$-th step of the program run. We consider the termination time $T$ of $\Pi$ and prove that $\{Y_n\}_{n=0}^\infty$ satisfies the prerequisites of our OST variant (\cref{thm:ost-variant})
by matching (A1) with (E1) and (A2) with (E2). 
Then by \cref{thm:ost-variant}, we obtain that $\expect{Y_T}\le \expect{Y_0}$. By (C2) in \cref{def:puwf}, we have that $Y_T=h(\loc_T,\pv_T)\cdot \prod_{i=0}^{T-1} W_i=\widehat{w}_T$. Thus, we have that $\llbracket \Pi\rrbracket_{\valin} (\Rset^{|\pvars|})=\expectdist{\valin}{\widehat{w}_T}=\expect{\prod_{i=0}^{T-1} W_i}\le \expect{Y_0}=h(\lin,\valin)$. Lower bounds are derived similarly. The detailed proof is relegated to~\cref{app:ost}.
\end{proof}

\begin{remark}%[Non-score-recursive WPTS's]
Our OST approach can handle programs with unbounded score values as caused by weights greater than $1$ that appear inside loop bodies, but with prerequisites in Theorem~\ref{thm:puwf-normalizing} (i.e., bounded update, (E1) and (E2)) to ensure the integrability to apply our OST variant. Note that our OST approach can directly handle non-score-recursive WPTS's 
as in such WPTS the score value  inside loops is bounded by one. \qed
\end{remark}

\section{Algorithmic Approaches}\label{sec:algorithm}
In this section, we present algorithms for our fixed-point and OST approaches. 
Recall the task is to compute bounds for unrestricted expected weights over Bayesian probabilistic programs. 

\subsection{Fixed-point Approach for Score-at-end Programs}

Our algorithm for the fixed-point approach solves a polynomial template $h$ w.r.t the fixed point conditions in \cref{thm:fix-point-bounds}. The details of our algorithm (inputs and stages) are as follows.

\smallskip
\noindent\textbf{Inputs:} The inputs include a score-at-end WPTS $\Pi=(\pvars, \rvars,  L,\lin,\lout,\mu_{\mathrm{init}}, \rdvarjdis,\transset)$ parsed from a Bayesian probabilistic program $P$ written in our PPL, and extra parameters $d,m$, for which $d$ is the degree of the polynomial to be solved and $m$ is the number of partitions that divides the support $\supp{\mu_{\mathrm{init}}}$ of the initial distribution uniformly into $m$ pieces for which our algorithm solves a polynomial for each piece. Note that whether a WPTS is score-at-end or not can be done by checking the score statement and applying techniques from \cite{DBLP:conf/cav/ChakarovS13,DBLP:conf/popl/ChatterjeeFNH16} to check AST. 

\smallskip
\noindent\textbf{Stage 1: Pre-processing.} Our algorithm has the following pre-processing to obtain auxiliary information for the input WPTS. 

\smallskip
\noindent{\em Invariant}: To have an over-approximation of the set of reachable program states, our algorithm leverages external invariant generators (such as~\cite{SankaranarayananSM04}) to generate numerical invariants for the WPTS. We denote the generated invariant at each location $\loc$ by $I(\loc)$ and treat each invariant $I(\loc)$ directly as the set of program valuations satisfying $I(\loc)$. 

\smallskip
\noindent{\em Bounded range}: Our algorithm calculates a bounded subset $B$ of program valuations encoded as a logical formula $\Phi_B$ (so that $B=\{\pv\mid \pv\models\Phi_B\}$) by heuristics. A simple heuristic here would be to run the WPTS for a small number of transitions and determine the bounded range of each program variable as the range covered by these transitions.
The input WPTS will be truncated onto the bounded range $B$  to increase the accuracy of the polynomial solving. 
Moreover, our algorithm calculates an extended bounded range $B'$ of $B$ by examining all transitions from $B$, i.e., for all $\pv\in B$, all locations $\loc\ne \lout$ and all weighted forks $F=\langle \loc'_j, p_j, \upd_j,\wet_j \rangle$ in some transition with source location $\loc$ and all $\rv\in \supp{\rdvarjdis}$, we have that $\upd_j(\pv,\rv)\in B'$. 
Note that the exact choice of the bounded range is irrelevant to the soundness of our fixed-point approach.

\smallskip
\noindent {\em Polynomial approximation}: In the case that the score function $g$ at the termination is non-polynomial, our algorithm leverages external polynomial interpolators to calculate a piecewise polynomial approximation of $g$ over $B'$ with an error bound 
$\epsilon> 0$ such that $|g(\pv)-g'(\pv)|\le \epsilon$ for all $\pv\in B'$. 
Our algorithm then replaces $g$ with $g'$ to avoid non-polynomial arithmetic.
We prove that the replacement is sound up to the additive error $\epsilon$, see~\cref{app:error-analysis}. 

\begin{example}\label{ex:pedes-algo-pre}
Recall the Pedestrian example in \cref{sec3:pedestrian} and its WPTS $\Pi$ in \cref{fig:pedestrian-wpts}. We choose the algorithm parameters as $d=1$ and $m=30$ to exemplify our algorithm.
We derive an invariant $I$ simply from the loop guard so that $I(\lin)=pos\ge 0$ and $I(\lout)=pos<0$. 
As the program variables $start,step$ simply receive samples, we only consider the key program variables $pos,dist$.
The bounded range is denoted by $B=\{(pos,dis)\mid pos\in [0,5], dis\in [0,5]\}$,
and the extended bounded range is given by $B'=\{(pos,dis)\mid pos\in [-1,6],dis\in [0,6]\}$.
Since the program is score-at-end and its score function $g$ at the termination is non-polynomial (i.e., $g(dis)=pdf(normal(1.1,0.1),dis)$), we choose a polynomial approximation $g'$ of $g$ with the error bound $\epsilon=10^{-5}$. \qed
\end{example}

\smallskip
\noindent\textbf{Stage 2: Partition.}
Our algorithm splits the set $\calV:=\supp{\mu_{\mathrm{init}}}$ of initial program valuations uniformly into $m\ge 1$ disjoint partitions $\calV_1,\dots,\calV_m$ and construct a set $\calW=\{\pv_1,\dots,\pv_m\}$ such that each $\pv_i\in\calV_i$. Our approach tackles each $\calV_i$ ($1\le i\le m$) separately to obtain polynomial bounds $l_i,u_i$ such that $l_i(\pv)\le \measureSem{\Pi}_{\pv}(\Rset^{|\pvars|})\le u_i(\pv)$ for all $\pv\in\calV_i$. It follows that the interval bound for $\llbracket \Pi \rrbracket (\Rset^{|\pvars|})$ 
can be derived by integrals of polynomial bounds over all $\calV_i$'s, that is, 
    \begin{align*}\label{eq:interval-analysis}
        \tag{$\clubsuit$}
        \sum_{i=1}^m \int_{\calV_i} l_i(\pv) \mu_{\mathrm{init}}(\mathrm{d} \pv) \le \llbracket \Pi \rrbracket (\Rset^{|\pvars|})= \int_{\calV} \measureSem{\Pi}_{\pv}(\Rset^{|\pvars|})\cdot \mu_{\mathrm{init}}(\mathrm{d} \pv)\le \sum_{i=1}^m \int_{\calV_i} u_i(\pv) \mu_{\mathrm{init}}(\mathrm{d} \pv) .
    \end{align*}

\begin{example}\label{ex:pedes-algo-part}
We obtain the set $\calV=\{(pos,dis)\mid pos\in [0,3],dis=0\}$, and partition $\calV$ uniformly into $m=30$ disjoint subsets on the dimension $pos$, i.e., $\calV_1=\{(pos,dis)\mid pos\in [0,0.1],dis=0\},\dots,\calV_{30}=\{(pos,dis)\mid pos\in [2.9,3],dis=0\}$. We calculate the midpoints of the dimension $pos$ for all $\calV_i$'s, and construct the set $\calW=\{(0.05,0),(0.15,0),\dots,(2.95,0)\}$.\qed
\end{example}

\smallskip
\noindent\textbf{Stage 3: Truncation.} 
Our algorithm performs a truncation operation to improve the accuracy. Intuitively, the truncation operation restricts the program values into the bounded range $B$ calculated from the pre-processing, and over-approximates the expected weight outside the bounded range by a truncation approximation. 
A \emph{truncation approximation} is a function $\calM:\mathbb{R}^{|\pvars|}\to [0,\infty)$ such that each $\calM(\pv)$ is intended to be an over- or under-approximation of the expected weight 
$\llbracket \Pi\rrbracket_{\pv} (\Rset^{|\pvars|})$ 
outside the bounded range $B$.  
The truncation operation is given as follows. 

\begin{definition}[Truncation Operation]\label{def:truncation}
Given the bounded range $B$ and a truncation approximation $\calM$,the \emph{truncated} WPTS $\Pi_{B,\calM}$ is defined as
	$\Pi_{\trunc,\calM}:=( \pvars, \rvars, L\cup\{\#\}, \lin, \lout,\mu_{\mathrm{init}},\rdvarjdis, \transset_{\trunc,\calM})$
	where $\#$ is a fresh termination location and the transition relation $\transset_{\trunc,\calM}$ is given by
	\begin{align*}
	&\transset_{\trunc,\calM}:=\{\langle \loc, \phi\wedge \Phi_\trunc, F_1, \dots, F_k \rangle\mid \langle \loc, \phi, F_1,\dots, F_k \rangle\in\transset\mbox{ and } \loc\ne \lout\}\\
 &\quad\cup \{\langle \loc, \phi\wedge (\neg\Phi_\trunc), F^{\calM,\sharp}_1, \dots, F^{\calM,\sharp}_k \rangle\mid \langle \loc, \phi, F_1,\dots, F_k \rangle\in\transset\mbox{ and } \loc\ne \lout\}\tag{\ddag}\\
 &\quad\cup\{\langle \lout, \mathbf{true}, F_{\lout}\rangle, \langle \sharp, \mathbf{true}, F_\sharp\rangle \}
\end{align*}
for which (a) we have $F_\loc:=\langle \loc,1,\mbox{\sl id},\overline{\mathbf{1}} \rangle$ ($\loc\in\{\lout,\sharp\}$) where $\mbox{\sl id}$ is the identity function and $\overline{\mathbf{1}}$ is the constant function that always takes the value $1$, and (b) for a weighted fork $F=\langle \loc', p, \mbox{\sl upd}, \wet\rangle$ in the original WPTS $\Pi$ we have $F^{\calM,\sharp}:=F$ if $\loc'=\lout$ and $F^{\calM,\sharp}:=\langle \sharp, p, \mbox{\sl upd}, \calM\rangle$ otherwise.
\end{definition}
Informally, we obtain the truncated WPTS by restraining each transition to the bounded range $B$ and redirecting 
all transitions jumping out of the bounded range but not into the location $\lout$ to the fresh termination location $\sharp$. We add the self-loop $\langle \sharp, \mathbf{true}, F_\sharp\rangle$ to ensure determinism and totality. 

We call a truncation approximation $\calM$ \emph{upper} (reps. \emph{lower}) if for all reachable state $(\loc,\pv)$ such that $\pv\not\in B$, it holds that $\llbracket \Pi\rrbracket_{\pv}(\Rset^{|\pvars|})\le \calM(\pv)$ (resp. $\llbracket \Pi\rrbracket_{\pv}(\Rset^{|\pvars|})\ge \calM(\pv)$). 
We prove  the correctness that $\llbracket \Pi\rrbracket_{\valin} (\Rset^{|\pvars|})\le \llbracket \Pi_{B,\calM}\rrbracket_{\valin}(\Rset^{|\pvars|})$ for all initial program valuations $\valin$ if $\calM$ is an upper truncation approximation, and $\llbracket \Pi\rrbracket_{\valin} (\Rset^{|\pvars|})\ge \llbracket \Pi_{\trunc,\calM}\rrbracket_{\valin}(\Rset^{|\pvars|})$ if $\calM$ is lower. 
See details in~\cref{app:truncaion}. 

To calculate truncation approximations, our algorithm takes a bound $M$ for the score function and has $0$ and $M$ as the trivial lower and upper truncation approximations. To sharpen the truncation approximations, our algorithm either utilizes the monotonicity of program variables to tighten the estimation of the values of the score function at the termination, or derives polynomial truncation approximations by applying polynomial solving to our fixed-point approach without truncation.

\begin{example}\label{ex:pedes-algo-trunc}
We choose the bounded range 
$[0,5]\times [0,5]$ that specifies $[0,5]$ for both the variable $pos,dis$. The truncation approximations are set to $\calM_{\mathrm{up}}=2.1\times 10^{-330}$ and $\calM_{\mathrm{low}}=0$. The values $2.1\times 10^{-330},0$ are calculated by the monotonicity of the density function $pdf(normal(1.1,0.1),dis)$ and the values of $dis$ at $5,6$, respectively. We obtain two truncated WPTS's $\Pi_{\trunc,\calM_{\mathrm{up}}}$ and $\Pi_{\trunc,\calM_{\mathrm{low}}}$. 
\qed
\end{example}

\begin{remark}
Recall that in our OST approach, we require the bounded update of program variables. This rules out samplings from 
unbounded probability distributions such as normal distribution. By using truncation, we could extend our approach to unbounded update by a suitable over-approximation function $\calM$ that handles unbounded update jumping outside the truncated range. The truncation operation could also handle hierarchical models such as $\textbf{score}(normal(0, 0.1 + Y), dist)$ with  $Y$ sampled from the beta distribution $\textbf{beta}(0,1)$ by a truncation that includes the hierarchical variable $Y$, and unbounded initial distribution by truncating the distribution into a bounded distribution with suitable approximations.  \qed
\end{remark}

\smallskip
\noindent\textbf{Stage 4: Polynomial Solving. }
Our algorithm establishes $d$-degree polynomial templates for $\Pi_{\trunc,\calM_{\mathrm{up}}}$ and 
$\Pi_{\trunc,\calM_{\mathrm{low}}}$, and derives polynomial upper and lower bounds for the expected weights $\llbracket \Pi_{\trunc,\calM_{\mathrm{up}}}\rrbracket_{\pv_i}(\Rset^{|\pvars|})$ and $\llbracket \Pi_{\trunc,\calM_{\mathrm{low}}}\rrbracket_{\pv_i}(\Rset^{|\pvars|})$ for each initial program valuation $\pv_i\in\calV_i$ in $\calW$ by solving the templates w.r.t the PUWF and PLWF constraints (i.e., (C1), (C2), (C1') from \cref{def:puwf}), respectively. 
The correctness of this stage follows from \cref{thm:fix-point-bounds} and that polynomials are bounded over a bounded range. 
Below we present the details in \textbf{Step A1} -- \textbf{A3}. 
We focus on $\Pi_{\trunc,\calM_{\mathrm{up}}}$ and polynomial upper bounds. The case of lower bounds 
considers $\Pi_{\trunc,\calM_{\mathrm{low}}}$ and is similar.

\smallskip
\noindent\textbf{Step A1.} In this step, for each location $\loc\not\in \{\lout,\sharp\}$, our algorithm
sets up a $d$-degree polynomial template $h_\loc$ over the program variables $\pvars$. Each template is a summation of all monomials in the program variables of degree no more than $d$, for which each monomial is multiplied with an unknown coefficient.    
For $\loc\in \{\lout,\sharp\}$, our algorithm assumes $h_\loc\equiv 1$.

\smallskip
\noindent\textbf{Step A2.} In this step, our algorithm establishes constraints for the templates $h_\loc$'s from (C1), (C1') in \cref{def:puwf} (as (C2) is satisfied directly by the form of $h_\loc$). 
For every location $\loc \in\locs\setminus\{\lout,\sharp\}$, 
we have the following relaxed constraints of (C1) to synthesize a PUWF over $\Pi_{\trunc,M_{up}}$: 
\begin{itemize}
\item[(D1)] 
For every program valuation $\pv\in I(\loc) \cap B$, we have that $\ewt(h)(\loc,\pv) \le h(\loc, \pv)$.  
\item[(D2)] For every 
program valuation $\mathbf{v}\in I(\loc) \cap (B'\setminus B)$, we have that $\calM_{\mathrm{up}}(\pv)\le h(\loc,\pv)$.  
\end{itemize}

For lower bounds over $\Pi_{\trunc,M_{low}}$, our algorithms have the relaxed PLWF constraints (D1') and (D2') which are obtained from (D1) and resp. (D2) by replacing ``$\ewt(h)(\loc,\pv) \le h(\loc, \pv)$'' with ``$\ewt(h)(\loc,\pv) \ge h(\loc, \pv)$'' in (D1) and resp. ``$\calM_{\mathrm{up}}(\pv)\le h(\loc,\pv)$'' with ``$\calM_{\mathrm{low}}(\pv)\ge h(\loc,\pv)$'' in (D2), respectively. 
We have that (D1) and (D2) together ensure (C1)  since $\calM_{\mathrm{up}}(\pv)\le h(\loc,\pv)$ implies that $\ewt(h)(\loc,\pv) \le h(\loc, \pv)$ for every location $\loc \in\locs\setminus\{\lout,\sharp\}$ and program valuation $\pv\in I(\loc) \cap (B'\setminus B)$. The same holds for (D1') and (D2'). 

Note that in (D1), the calculation of $\ewt(h)(\loc,\pv)$ has the piecewise nature that different sampling valuations $\rv$ may cause the next program valuation to be within or outside the bounded range, and to satisfy or violate the guards of the transitions in the WPTS. In our algorithm, we have a refined treatment for (D1) that enumerates all possible situations for a sampling valuation $\rv$ that satisfy different guards in the calculation of $\ewt(h)(\loc,\pv)$, for which we use an SMT solver (e.g., Z3~\cite{Z3paper}) to compute the situations. As for (D2), we use (D2) to avoid handling the piecewise feature in the computation of $\ewt(h)(\loc,\pv)$ from within/outside the bounded range (i.e., the computation is a direct computation over a single-piece polynomial), so that the amount of computation of $\ewt(h)(\loc,\pv)$ is reduced by ignoring the piecewise feature. The use of (D2) to reduce the computation follows from the extended bounded range in the pre-processing. 
The same holds for (D1') and (D2').

\begin{example}\label{ex:pedes-algo-A3}
Recall \cref{ex:pedes-algo-pre,ex:pedes-algo-part,ex:pedes-algo-trunc}, we set up the $1$-degree template $h_{\lin}(pos,dis)=a_1\cdot pos+a_2\cdot dis+a_3$ withe unknown coefficients $a_1,a_2,a_3\in\Rset$ and $h_{\loc}(pos,dis)=1$ for $\loc=\{\lout,\sharp\}$.
Then the bounded range $\calD_1=I(\lin)\cap B$ in (D1) is defined such that $ pos\in [0,5]\wedge dis\in [0,5]$. Consider to derive upper bounds from $\Pi_{\trunc,\calM_{\mathrm{up}}}$. We make a fine-grained treatment for (D1) by splitting the range $\calD_1$ and enumerating all possible situations for the sampling valuation $\textbf{sample uniform}(0,1)$ that the next program valuation satisfies or violates the loop guard ``$pos\ge 0$''. In detail, $\calD_1$ is split into $\calD_{11}=\{(pos,dis)\mid pos\in [0,1), dis\in [0,5]\}$ and $\calD_{12}=\{(pos,dis)\mid pos\in (1,5], dis\in [0,5]\}$, where $\calD_{11}$ stands for the situation that with different sampling valuations the next program valuation may satisfy the loop guard (so that the next location is $\lin$) or violate the loop guard (so that the next location is directed to $\lout$), and  $\calD_{12}$ stands for the situation that the next program valuation will definitely satisfy the loop guard and the next location is $\lin$.

We first show the PUWF constraints for $\Pi_{\trunc,\calM_{\mathrm{up}}}$. Recall that the program has two probabilistic branches with probability $0.5$. When the current program valuation is in $\calD_{11}$, we observe that 
(a) if the loop takes the branch $pos:=pos+step$, then the next value of $pos$ remains to be non-negative and the loop continues, and 
(b) if the loop takes the branch $pos:=pos-step$, 
then the next value of $pos$ either satisfies or violates the loop guard, depending on the exact value of $step\in [0,1]$. 

Then we have the constraint (D1.1) over $\calD_{11}$ that has expectation over a piecewise function on $step$ derived from whether the loop terminates in the next iteration or not, as follows:
\begin{itemize}
\item[(D1.1)] $\forall pos,dis.$ $pos\in [0,1)\wedge dis\in [0,5] \Rightarrow$
	\begin{align*}
	&0.5\cdot\expectdist{step}{[pos-step\ge 0]\cdot h_{\lin}(pos-step,dis+step)+[pos-step<0]\cdot g'(dis+step)}\\ +& 0.5\cdot\expectdist{step}{h_{\lin}(pos+step,dis+step)}\le h_{\lin}(pos,dis).
	\end{align*} 
  \end{itemize}
When the current program valuation is in $\calD_{12}$, we observe that the next program valuation is guaranteed to satisfy the loop guard, and hence we have the constraint (D1.2) over $\calD_{12}$ as follows: 
\begin{itemize}
      \item[(D1.2)] $\forall pos,dis.$ $pos\in (1,5]\wedge dis\in [0,5]\Rightarrow$
	\begin{align*}
	0.5\cdot\expectdist{step}{h_{\lin}(pos-step,dis+step)}+0.5\cdot\expectdist{step}{h_{\lin}(pos+step,dis+step)}\le h_{\lin}(pos,dis).	
	\end{align*}
  \end{itemize}
The range $\calD_2=I(\lin) \cap (B'\setminus B)$ in (D2) is represented by the disjunctive formula $\Phi:=(pos\in [0,6]\wedge dis\in [5,6])\vee (pos\in [5,6]\wedge dis\in [0,6])$. For the program valuation in $\calD_2$, its next location is $\loc_\sharp$. Recall the truncation approximation $\calM_{\mathrm{up}}=2.1\times 10^{-330}$. 
From (D2), we have two constraints from the disjunctive clauses in $\Phi$ as follows:
\begin{itemize}
    \item[(D2.1)] $\forall pos,dis.$ $pos\in [0,6]\wedge dis\in [5,6]\Rightarrow 2.1\times 10^{-330}\le h_{\lin}(pos,dis)$. 
	\item[(D2.2)] $\forall pos,dis.$ $pos\in [5,6]\wedge dis\in [0,6]\Rightarrow 2.1\times 10^{-330}\le h_{\lin}(pos,dis)$. 
\end{itemize}
For $\Pi_{\trunc,\calM_{\mathrm{low}}}$, the PLWF constraints (D1'.1) to (D2'.2) are obtained from the PUWF constraints above by replacing ``$\le$'' with ``$\ge$'' and $2.1\times 10^{-330}$ with $0$. \qed
\end{example}

\smallskip
\noindent\textbf{Step A3.} In this step, for every initial program valuation $\pv_i\in\calV_i$ in $\calW$ where $\calV_i$'s and $\calW$ are obtained from \textbf{Stage 2}, our algorithm solves the unknown coefficients in the templates $h_\loc$ ($\loc\in \locs\setminus \{\lout,\sharp\}$) via the well-established methods of Putinar's Positivstellensatz~\cite{putinar} or Handelman's Theorem~\cite{handelman1988representing}.
In detail, our algorithm minimizes the objective function $h_{\lin}(\pv_i)$ for each $\pv_i$ in $\calW$ which subjects to the PUWF constraints from the previous step to derive polynomial upper bounds over $\Pi_{\trunc,\calM_{\mathrm{up}}}$. For lower bounds, we use PLWF and consider to maximize the objective function. 

Note that the PUWF or PLWF constraints from \textbf{Step A2} can be represented as a conjunction of formulas in the form $\forall\pv\in \calP.(\mathbf{g}(\pv)\ge 0)$ where the set $\calP$ is defined by a conjunction of polynomial inequalities in the program variables and $\mathbf{g}$ is a polynomial over $\pvars$ whose coefficients are affine expressions in the unknown coefficients from the templates, and such formulas can be guaranteed by the sound forms of Putinar's Positivstellensatz and Handelman's Theorem. The application of Putinar's Positivstellensatz results in semidefinite constraints and can be solved by semidefinite programming (SDP), while the application of Handelman's Theorem is restricted to the affine case (i.e., every condition and assignment in the WPTS or the original program is affine) and leads to linear constraints and can be solved by linear programming (LP). 

We refer to~\cref{app:putinar,app:Handelman} for the details on the application of Putinar's Positivstellensatz and Handelman's Theorem. 

\begin{example}\label{ex:pedes-algo-A4}
Recall \cref{ex:pedes-algo-part,ex:pedes-algo-trunc,ex:pedes-algo-A3}, we have that $\calV_1=\{(pos,dis)\mid pos\in [0,0.1],dis=0\},\dots,\calV_{30}=\{(pos,dis)\mid pos\in [2.9,3],dis=0\}$ and the set $\calW=\{(0.05,0),(0.15,0),\dots,(2.95,0)\}$. 
Pick a point $\pv_1=(0.05,0)\in\calV_1$ from $\calW$, and to synthesize the upper bound for $\llbracket \Pi_{\trunc,\calM_{\mathrm{up}}}\rrbracket_{\pv_1}(\Rset^{|\pvars|})$,
we solve the following optimization problem whose objective function is $h_{\lin}(0.05,0)$, i.e., 
$$\textstyle \textbf{Min}~~~~~ 0.05\cdot a_1+a_3 \qquad \textstyle \textbf{s.t.}
 ~~~~~\text{constraints (D1.1)--(D2.2)}$$
Dually, the lower bound for $\llbracket \Pi_{\trunc,\calM_{\mathrm{low}}}\rrbracket_{\pv_1}(\Rset^{|\pvars|})$ is solved by 
$\textstyle \textbf{Max}~~~~~ 0.05\cdot a_1+a_3~\textstyle \textbf{s.t.} ~\text{(D1'.1)--(D2'.2)}$.
Although the constraints are universally quantified, the universal quantifiers can be soundly (but not completely) removed and relaxed into semidefinite constraints over the unknown coefficients $a_i$'s by applying Putinar's Positivstellensatz, where
we over-approximate all strict inequalities (e.g., ''$<$'') by non-strict ones (e.g., ``$\le$''). Then we call an SDP solver to solve the two optimization problems and find the solutions of $a_i$'s, which will generate two polynomial bound functions $Up_1,Lw_1$ such that $\llbracket \Pi_{\trunc,\calM_{\mathrm{up}}}\rrbracket_{\pv}(\Rset^{|\pvars|}) \le Up_1(\pv)+\epsilon$ and $Lw_1(\pv)-\epsilon \le \llbracket \Pi_{\trunc,\calM_{\mathrm{low}}}\rrbracket_{\pv}(\Rset^{|\pvars|})$ for all program valuation $\pv\in\calV_1$ with the polynomial approximation error  $\epsilon=10^{-5}$.    \qed
\end{example}

\smallskip
\noindent\textbf{Stage 5: Integration.}
As a consequence of \textbf{Stage 4}, our algorithm obtains
polynomial upper bounds  $Upper=\{Up_1,\dots,Up_m\}$ for the expected weights $\{ \llbracket \Pi_{\trunc,\calM_{\mathrm{up}}} \rrbracket_{\pv_1}(\Rset^{|\pvars|}), \dots,  \llbracket \Pi_{\trunc,\calM_{\mathrm{up}}} \rrbracket_{\pv_m}(\Rset^{|\pvars|}) \}$ w.r.t. the $m$ partitions $\calV_i,\dots,\calV_m$. 
Then our algorithm integrates these polynomial upper bounds to derive the upper bound for $\llbracket \Pi_{\trunc,\calM_{\mathrm{up}}} \rrbracket (\Rset^{|\pvars|})$.  
In detail, we have that
\begin{align}\label{eq:nc-upper}
\textstyle\llbracket \Pi_{\trunc,\calM_{\mathrm{up}}} \rrbracket (\Rset^{|\pvars|})\le \sum_{i=1}^m \int_{\calV_i} Up_i(\pv) \cdot \frac{1}{m} \mathrm{d}\pv=:u
\end{align}
Similarly,
the lower bound for $\llbracket \Pi_{\trunc,\calM_{\mathrm{low}}} \rrbracket (\Rset^{|\pvars|})$ is given by
\begin{align}\label{eq:nc-lower}
\textstyle l:=\sum_{i=1}^m \int_{\calV_i} Lw_i(\pv) \cdot \frac{1}{m} \mathrm{d}\pv \le \llbracket \Pi_{\trunc,\calM_{\mathrm{low}}} \rrbracket (\Rset^{|\pvars|})  
\end{align}
with the polynomial lower bounds $Lower=\{Lw_1,\dots,Lw_m\}$ generated in \textbf{Stage 4}. Note that $[l,u]$ is also the interval bound for the original $\llbracket \Pi \rrbracket (\Rset^{|\pvars|})$ as $\llbracket \Pi_{B,\calM_{\mathrm{low}}}\rrbracket_{\valin}(\Rset^{|\pvars|}) \le \llbracket \Pi\rrbracket_{\valin} (\Rset^{|\pvars|})\le \llbracket \Pi_{B,\calM_{\mathrm{up}}}\rrbracket_{\valin}(\Rset^{|\pvars|})$ for all initial program valuations $\valin$ (see~\cref{app:truncaion}).

If the score function $g$ at the termination is non-polynomial, our algorithm integrates the polynomial approximation error $\varsigma=\mathrm{volume}(\calV)\cdot\epsilon$ caused by its polynomial approximation $g'$ to the two bounds, i.e, $l'=l-\varsigma$ and $u'=u+\varsigma$, where $\mathrm{volume}(\calV)$ is the volume of $\calV$ and $\epsilon$ is the error bound. In practice, to ensure the tightness of the interval bounds, we can control the amount of $\epsilon$ so that the approximation error $\varsigma$ is at least one magnitude smaller than the values of $l,u$. 

The correctness of our algorithm is stated as follows. 
The pseudo code is in~\cref{alg:fixed-point}. 

\begin{algorithm}
    \SetKwData{Left}{left}\SetKwData{This}{this}\SetKwData{Up}{up}
    \SetKwFunction{Union}{Union}\SetKwFunction{FindCompress}{FindCompress}
    \SetKwInOut{Input}{Input}\SetKwInOut{Output}{Output}

    \Input{A score-at-end WPTS $\Pi$, Parameters $d, m$}
     \Output{The upper bound $u$ and the lower bound $l$ for $\llbracket\Pi \rrbracket (\Rset^{|\pvars|})$}
    \BlankLine
    \textbf{Pre-processing:} 
    \begin{enumerate}
        \item Invariant generation $I$; 
        \item Bounded range $B$ and extended bounded range $B'$;
        \item Polynomial approximation for non-polynomial score function at termination.
    \end{enumerate}
    \textbf{Partition:} Splits the set $\calV:=\supp{\mu_{\mathrm{init}}}$ into $m\ge 1$ disjoint partitions $\calV_1,\dots,\calV_m$ and construct a set $\calW=\{\pv_1,\dots,\pv_m\}$ \textbf{s.t.} $\pv_i\in\calV_i$. \\
    \textbf{Truncation:}  
    \begin{enumerate}
        \item Restricts the program values into the bounded range $B$;
        \item Over (resp. under) -approximate the expected weight outside the bounded range $B$ with truncation approximation $\calM_{\mathrm{up}}$ (resp. $\calM_{\mathrm{low}}$);
        \item Construct the truncated WPTS $\Pi_{B,\calM_{\mathrm{up}}}$ (resp. $\Pi_{B,\calM_{\mathrm{low}}}$) for upper (resp. lower) bounds.
    \end{enumerate}
    \textbf{Polynomial Solving:} 
    Establish $d$-degree template $h$ and constraints over $h$, \\
    \For{each $\loc$}{
        Establish constraints (D1) (within truncate range $I(\loc) \cap B$) and (D2) (outside truncated range  $I(\loc) \cap (B'\setminus B)$) \tcc*{upper bound} 
        Establish constraints (D1') (within truncate range $I(\loc) \cap B$) and (D2') (outside truncated range $I(\loc) \cap (B'\setminus B)$) \tcc*{lower bound} 
    }
    Optimize $h_{\lin}(\pv_i)$ for each $\pv_i$ in $\calW$ with the constraints above, and produce polynomial upper and lower bounds $Upper=\{Up_1,\dots,Up_m\}$, $Lower=\{Lw_1,\dots,Lw_m\}$. \\
    \textbf{Integration:} Integrate the upper and lower bounds,\\
    $\textstyle \llbracket \Pi_{\trunc,\calM_{\mathrm{up}}} \rrbracket (\Rset^{|\pvars|})\le \sum_{i=1}^m \int_{\calV_i} Up_i(\pv) \mathrm{d}\pv=:u$,\quad
    $\llbracket \Pi_{\trunc,\calM_{\mathrm{low}}} \rrbracket (\Rset^{|\pvars|}) \ge \sum_{i=1}^m \int_{\calV_i} Lw_i(\pv) \mathrm{d}\pv=:l$
    \caption{The Algorithm for Fixed-Point Approach}   \label{alg:fixed-point}
\end{algorithm}

\begin{theorem}[Soundness]\label{thm:soundness}
	If our algorithm finds valid solutions for the unknown coefficients in the polynomial templates, then it returns correct interval bounds for $\llbracket \Pi \rrbracket (\Rset^{|\pvars|})$. 
\end{theorem}
\begin{proof}
Let $h_\loc$'s ($\loc\in\locs$) be the solved polynomial templates. Since the extended range $B'$ is bounded, we can find a bound $M$ such that all the values these $h_\loc$'s take fall in $[-M,M]$ within $B'$. By applying \cref{thm:fix-point-bounds} to the truncated WPTS and \cref{eq:interval-analysis}, we obtain the desired result. 
\end{proof}

\subsection{OST Approach for Bounded-Update Score-Recursive Programs}

The algorithm for our OST approach over bounded-update score-recursive programs follows similar stages as for our fixed-point approach. 
The main difference lies at the pre-processing and truncation stages. In the pre-processing stage, the difference includes the following:
\begin{itemize}
\item To apply \cref{thm:puwf-normalizing}, our approach checks the prerequisites (E1) and (E2) by external approaches. For example, the condition (E1) can be checked by existing approaches in concentration and tail bound analysis~\cite{DBLP:journals/toplas/ChatterjeeFNH18,DBLP:conf/pldi/WangS0CG21}, and the condition (E2) by SMT solvers. 
\item Our current algorithmic approach does not tackle score statements with non-polynomial score functions inside a loop body. In general, non-polynomial score functions inside the loop body can also be handled by piecewise polynomial approximation. A direct way is to provide upper and lower polynomial bounds for the original non-polynomial function, which avoids the calculation of the propagation of approximation errors along loop iterations.
The same applies to sampling variables inside loop bodies with non-polynomial probability density (or mass) functions.
\end{itemize}
In the truncation, the difference is that our OST approach derives polynomial truncation approximations 
by applying polynomial solving to our OST approach without truncation. 
The soundness of the algorithm (that produces correct bounds for expected weights) follows directly from~\cref{thm:puwf-normalizing}.

\section{Experimental Results}\label{sec:experiment}
In this section, we present the experimental evaluation of our approach over a variety of benchmarks.
First, we show that our approach can handle novel examples that cannot be addressed by existing tools such as \cite{DBLP:conf/cav/GehrMV16,DBLP:conf/pldi/GehrSV20,DBLP:conf/atva/HuangDM21,DBLP:conf/flops/NarayananCRSZ16}. 
Then we compare our approach with the state-of-the-art tool GuBPI \cite{Beutner2022b} over score-at-end Bayesian programs. (Note that GuBPI could only handle score-at-end programs.) Finally, even though the problem of path probability estimations is not the focus of our work, we demonstrate that our approach works well for this problem, for which we also compare the performance of our approach with GuBPI.
We implement our algorithms in Matlab. All results are obtained on an Intel Core i7 (2.3 GHz) machine with 16 GB of memory, running MS Windows 10. 

\subsection{Experimental Setup}\label{sec:6-1}

Since our approach is completely orthogonal to GuBPI, we conservatively have the experimental setup in order to minimize the advantage of the external inputs to our algorithms in the comparison.

\noindent\textbf{Inputs.} All benchmarks are in the form of a single while loop. We set 
two locations for each benchmark, i.e., $\lin$ for the entry of the while loop and $\lout$ for termination.
We denote the loop guard by $\phi$. 
We implement a parser from probabilistic programs into WPTS's in F\#. We conduct our experiments with various parameter combinations of $d, m$ and present the results in \cref{table:1,table:2,table:3}. 

\smallskip
\noindent{\textbf{Stage 1: Pre-processing.}} The pre-processing is illustrated as follows. 

\smallskip\noindent {\em Invariant:} We take the conservative setting that simply derives the invariant from the loop guard $\phi$ so that $I(\lin)=\phi$ and $I(\lout)=\neg\phi$.

\smallskip
\noindent {\em Bounded range:} To minimize the advantage of the choice of the bounded range to our approach, we have a conservative setting that has the bounded range to cover a majority part of program executions. Having a smaller bounded range would result in more accurate results, as polynomial solving is more accurate over a smaller bounded range. 
In detail, for each program variable $x$, we have $B''_x$ as the interval that is the projection of the support of the initial distribution onto the variable $x$. 
Then we choose a large deviation $\delta$ for all $B''_x$ ($x\in \pvars$) to get intermediate intervals $IB_x$, i.e., each ${IB}_x$ is given by ${IB}_x:=[\zeta_1-\delta,\zeta_2+\delta]$ where $[\zeta_1,\zeta_2]=B''_x$. The final bounded range $B$ is given as the intersection of the Cartesian product of all $IB_x$'s and the loop guard. 

\smallskip
\noindent {\em Polynomial approximation:} In the case that the input is a score-at-end Bayesian program and the score function is non-polynomial, we use the polynomial interpolator in Matlab to obtain a piecewise polynomial approximation for the score function. We do not have polynomial approximations in our OST approach since the benchmarks considered do not have non-polynomial score functions. 

\smallskip
\noindent{\textbf{Stage 2: Partition.}} We partition the set of initial valuations uniformly into $m$ disjoint subsets $\mathcal{V}_1,\dots,\mathcal{V}_m$, and choose the midpoints $v_i$ of each partition $\mathcal{V}_i$. 

\smallskip
\noindent{\textbf{Stage 3: Truncation.}} 
Our approach calculates the truncation approximations as described in~\cref{sec:algorithm}. For score-at-end programs, our approach either gets them by the direct bounds from the score function, further improves them by heuristics such as monotonicity, or derives polynomial truncation approximations by applying polynomial solving to our fixed-point approach without truncation. For score-recursive programs, our approach gets the truncation approximations by directly applying polynomial solving under our OST variant without truncation.

\smallskip
\noindent\textbf{Stage 4: Polynomial solving.} In applying the Positivstellensatz's, we use the LP solver in Matlab (resp. Mosek \cite{mosek}) for solving linear (resp. semidefinite) programming, respectively.  

\subsection{Results}

We focus on unbounded while loops as they distinguish our approach with previous approaches most significantly, and compare with the most relevant tool GuBPI~\cite{Beutner2022b}. To compare our approach with GuBPI fairly, our measurement of the time cost of our approach includes the time taken to generate all the extra inputs. 
 
\smallskip 
\noindent{\em NPD - Novel Examples.} 
We consider $10$ novel examples adapted from the literature, where all $7$ examples with prefix ``\textsc{Pd}'' or ``\textsc{RdWalk}'' are from \cite{Beutner2022b}, the two ``RACE'' examples are from \cite{DBLP:conf/pldi/WangS0CG21}, and the last example is from statistical phylogenetics   \cite{ronquist2021universal} (see also \cref{sec:overview}). Concretely, the ``RACE(V2)'' and ``BIRTH'' examples are both score-recursive probabilistic programs with weights greater than $1$, and thus their integrability condition should be verified by the existence of suitable concentration bounds (see~\cref{thm:puwf-normalizing}); other examples are score-at-end probabilistic while loops with unsupported types of scoring by previous tools (e.g., polynomial scoring \textbf{score}($y$) where $y$ is a single-variable polynomial). Therefore, no existing tools w.r.t. NPD can tackle these novel examples. The results are reported in \cref{table:1}, where the first column is the name of each example, the second column contains the parameters of each example used in our approach (i.e., the degree $d$ of the polynomial template, the number $m$ of partitions and the bounded range of program variables), the third column is the used solver, and the fourth and fifth columns correspond to the runtime of upper and lower bounds computed by our approach, respectively.
Our runtime is reasonable, that is, most examples can obtain tight bounds within $100$ seconds, and the simulation results by Pyro~\cite{bingham2019pyro} ($10^6$ samples per case) match our derived bounds. 
We display part of the comparison in \cref{fig:results1}, see \cref{app:experiments} for other figures.

\begin{table*}
%	\vspace{-1.7em}
	\caption{Results for Novel Examples}
	\label{table:1}
	%\begin{footnotesize}
	\resizebox{\textwidth}{!}{
		\begin{threeparttable}
			\begin{tabular}{|c|c|c|c|c|}
				\hline
				\multicolumn{1}{|c|}{\multirow{2}{*}{\textbf{Benchmark}}}  &
				\multicolumn{1}{c|}{\multirow{2}{*}{\textbf{Parameters}}}      &
				\multicolumn{1}{c|}{\multirow{2}{*}{\textbf{Solver}}}      &
				\multicolumn{1}{c|}{\textbf{Upper}}      &
				\multicolumn{1}{c|}{\textbf{Lower}}   \\ \cline{4-5}
				\multicolumn{1}{|c|}{} & \multicolumn{1}{c|}{}  &  \multicolumn{1}{c|}{}  & \multicolumn{1}{c|}{\textbf{ Time (s)}} & 
				\multicolumn{1}{c|}{\textbf{ Time (s)}}                    \\ \hline \hline
				\multirow{1}{*}{\textsc{Pd(v1)}} 
				& $d=6$, $m = 60$,$pos, dis \in [0,5]$   & SDP & $54.65$  & $52.39$        \\
				\hline
				\multirow{1}{*}{\textsc{Race(v1)}} 
				&$d=6$,$m = 40$,$h,t\in [0,5]$ & LP &$87.43$ & $86.27$  \\  
				\hline
				\multirow{1}{*}{\textsc{Race(v2)}\tnote{*}}
				&$d=6$,$m = 40$, $h,t\in [0,5]$ &  LP & $81.19$ & $81.18$   \\   \hline 
				\multirow{1}{*}{\textsc{RdWalk(v1)}} 
				& $d=6$,$m = 60$, $x,y\in [0,5]$ &  LP  & $46.65$  & $47.72$   \\ 
				\hline
				\multirow{1}{*}{\textsc{RdWalk(v2)}} 
				& $d=6$,$m = 60$, $x,y\in [0,5]$ &  LP  & $97.45$  & $103.65$   \\ 
				\hline
				\multirow{1}{*}{\textsc{RdWalk(v3)}} 
				& $d=6$,$m = 60$, $x,y\in [0,5]$ &  LP & $48.61$  & $49.39$   \\ 
				\hline
				\multirow{1}{*}{\textsc{RdWalk(v4)}} 
				& $d=6$, $m = 60$,$x,y\in [0,5]$ &  LP & $98.75$  & $98.21$    \\ 
				\hline
				\multirow{1}{*}{\textsc{PdMB(v3)}}
				& $d=4$, $m = 60$,$pos,dis\in [0,5]$ &   LP  & $15.26$  & $14.57$   \\  
				\hline
				\multirow{1}{*}{\textsc{PdMB(v4)}}  
				& $d=4$,$m = 60$, $pos,dis\in [0,5]$ &  LP   & $16.07$  & $16.12$   \\  
				\hline
				\multirow{1}{*}{\textsc{Birth}\tnote{*}}& $d=6$, $m = 40$, $lambda \in [0,2], time\in [0,10]$ & LP 	& $14.72$ & $16.66$          \\  
				\hline
			\end{tabular}
			
			\begin{tablenotes}
				\footnotesize
				\item[*] It is a score-recursive probabilistic program with weights greater than $1$.
			\end{tablenotes}
	\end{threeparttable}}
	%	\end{footnotesize}
\vspace{-3ex}
\end{table*}

\begin{figure}[t]
\begin{minipage}{0.48\textwidth}
	\centering
	\subfigure[\textsc{Pd(v1)}]{
	\resizebox{0.45\textwidth}{!}{
    	\includegraphics[width=2.5in,height=2in]{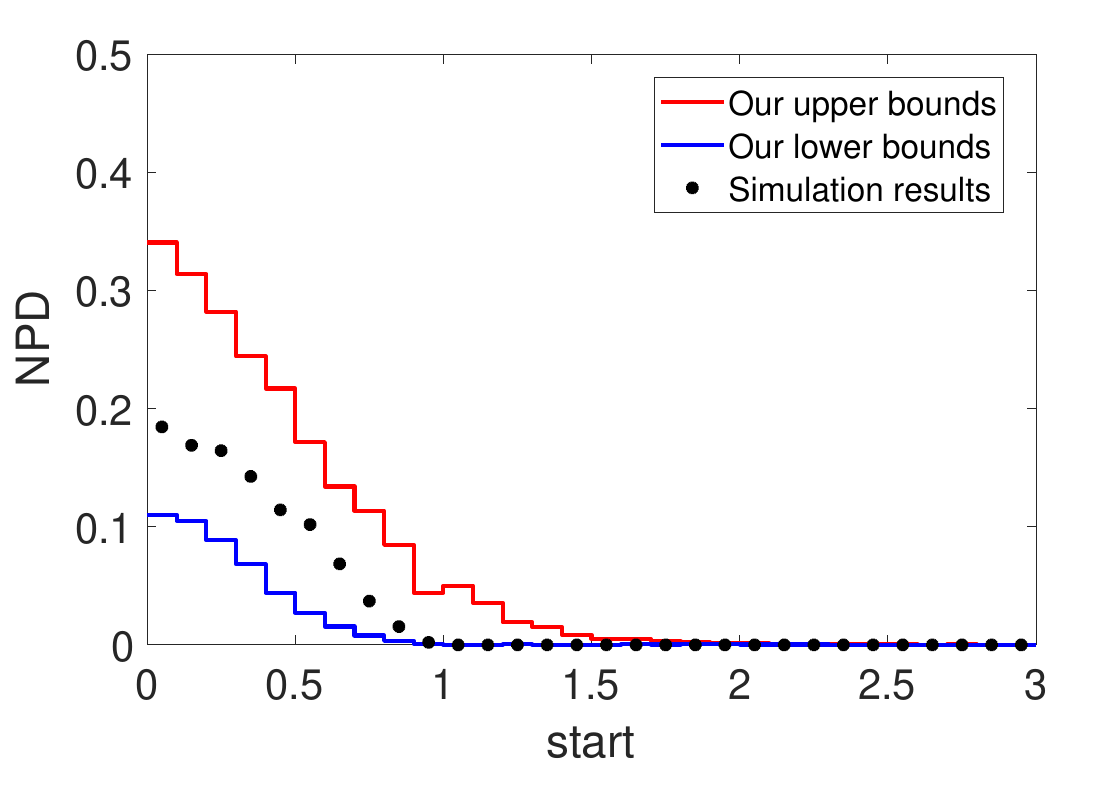}
    }
    }
    \hfill
    \subfigure[\textsc{Race(v2)}]{
        \resizebox{0.45\textwidth}{!}{
        \includegraphics[width=2.5in,height=2in]{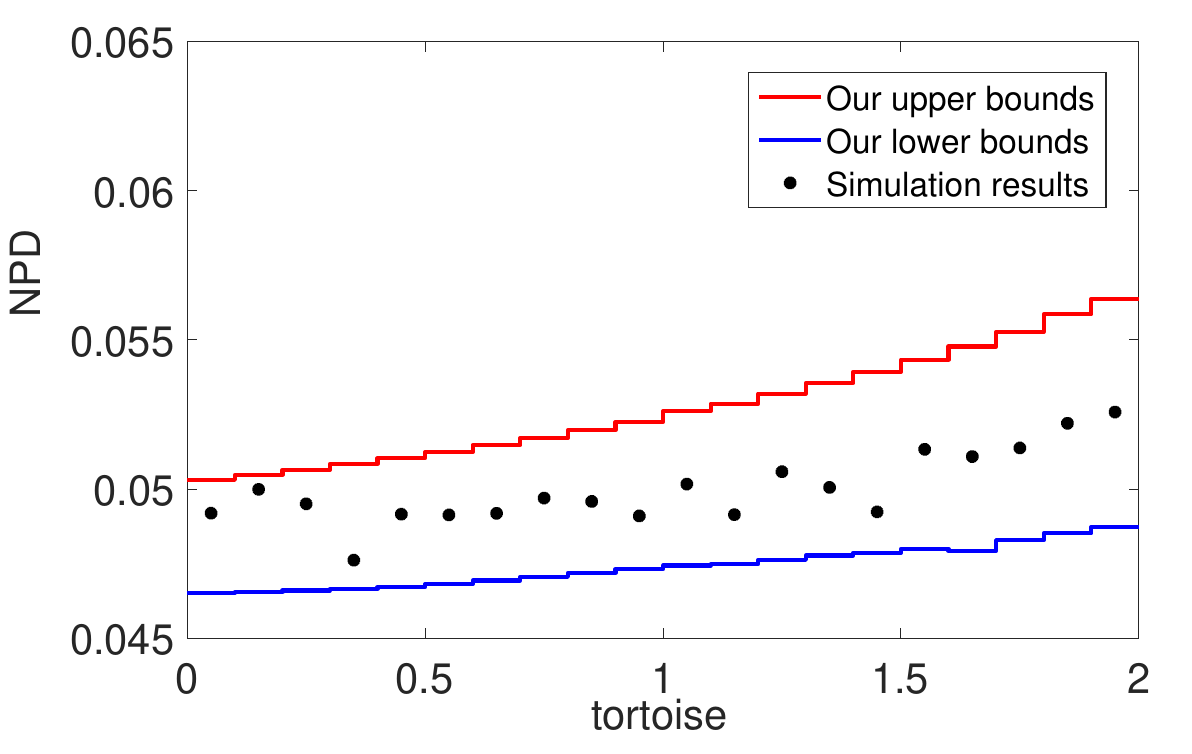}
        }
    }
    \\
	\subfigure[\textsc{RdWalk(v1)}]{
			\resizebox{0.45\textwidth}{!}{
	\includegraphics[width=2.5in,height=2in]{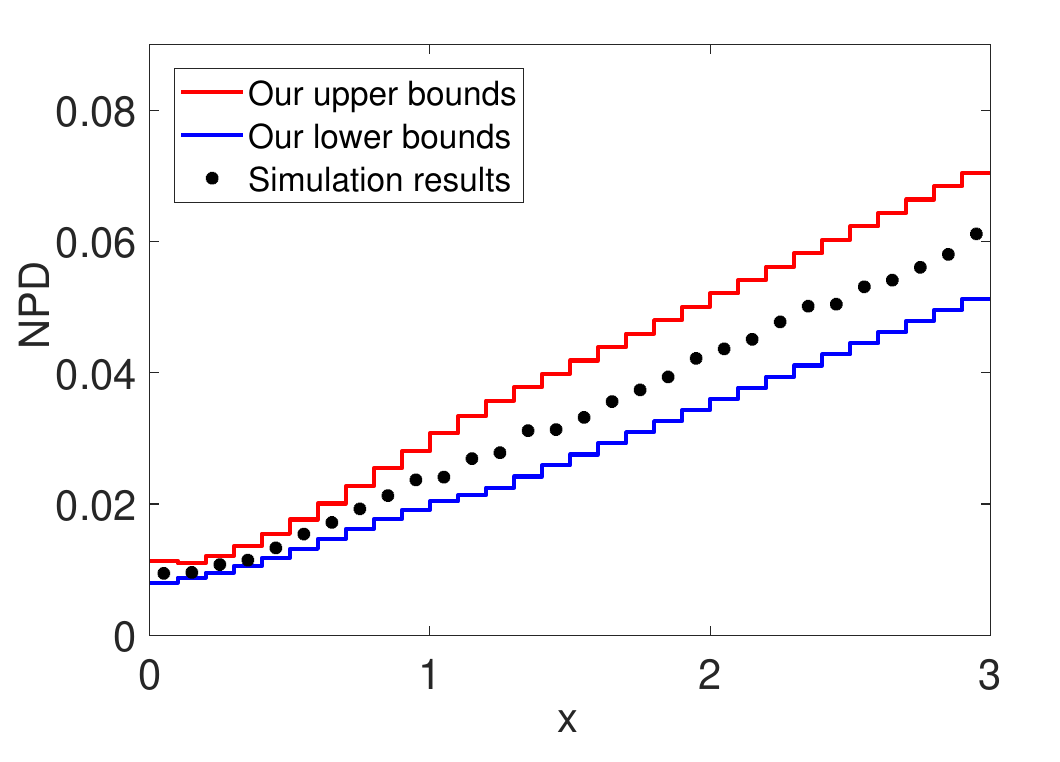}
    }
    }
    \hfill
	\subfigure[\textsc{PdMB(v3)}]{  % Multi-branches
			\resizebox{0.45\textwidth}{!}{
	\includegraphics[width=2.5in,height=2in]{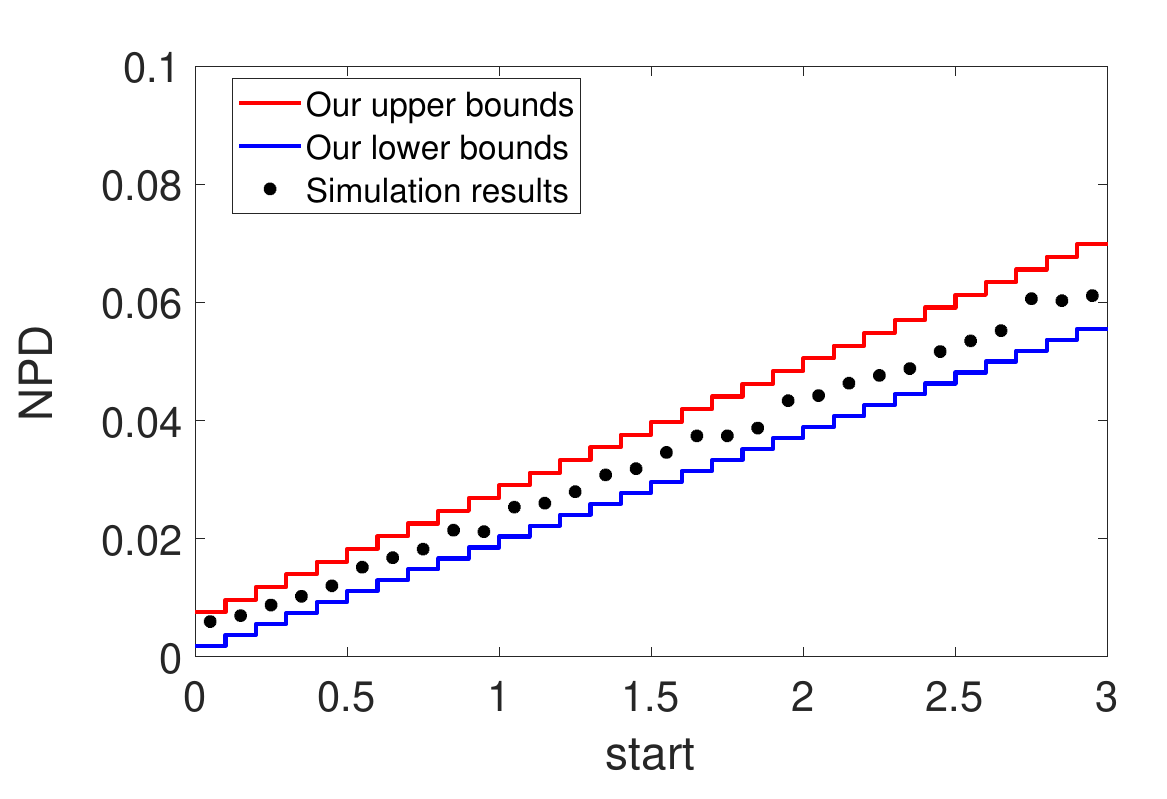}
}
}
\caption{NPD Bounds of Novel Examples}
\label{fig:results1}
\end{minipage}
\hfill
\begin{minipage}{0.48\textwidth}
	\centering
	\subfigure[\textsc{Pd}]{
		\resizebox{0.45\textwidth}{!}{
			\centering
			\includegraphics[width=2.5in,height=2in]{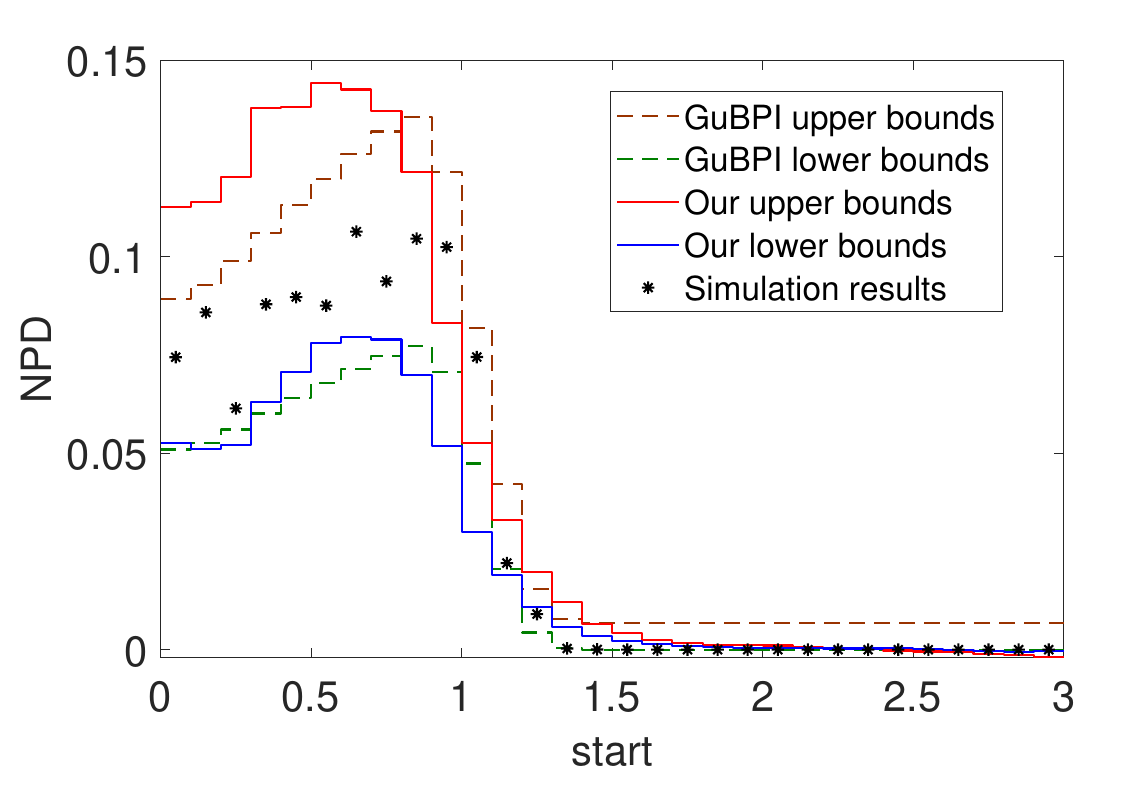}
		}
	}
	\hfill
	\subfigure[\textsc{PdLD}]{ 
		\resizebox{0.45\textwidth}{!}{
			\centering
			\includegraphics[width=2.5in,height=1.9in]{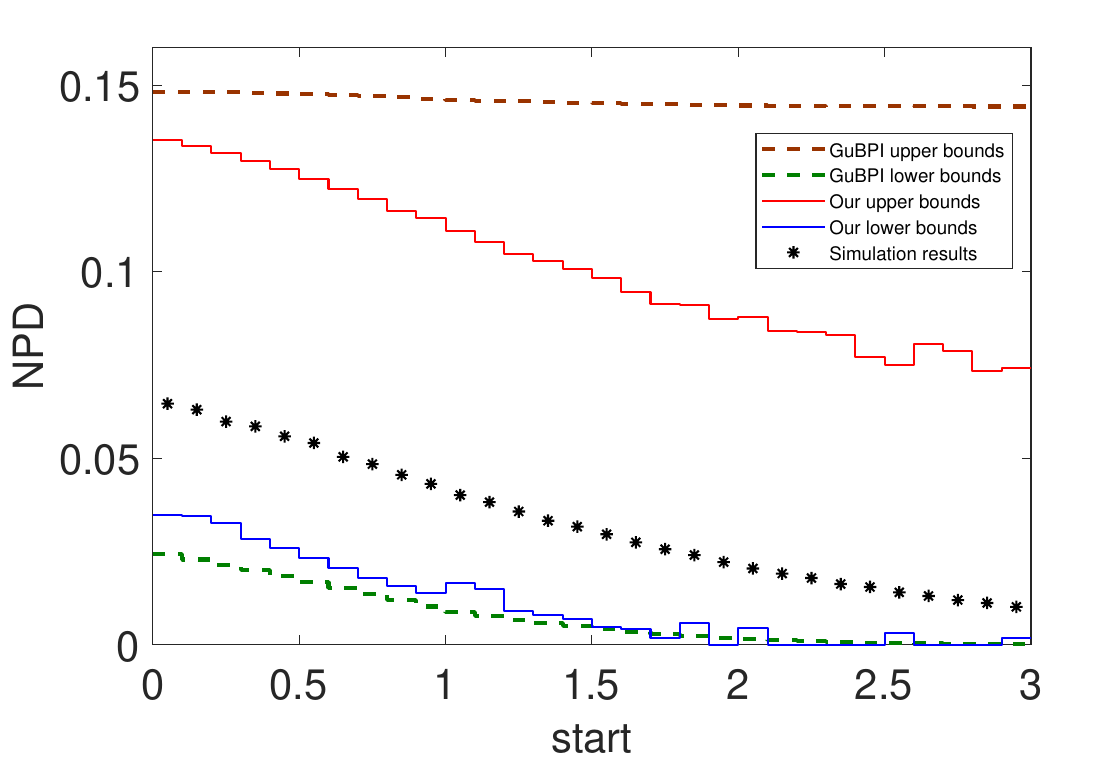}
		}
	}
    \\
	\subfigure[\textsc{PdBeta(v1)}]{
		\centering
		\resizebox{0.45\textwidth}{!}{
			\includegraphics[width=2.5in,height=2in]{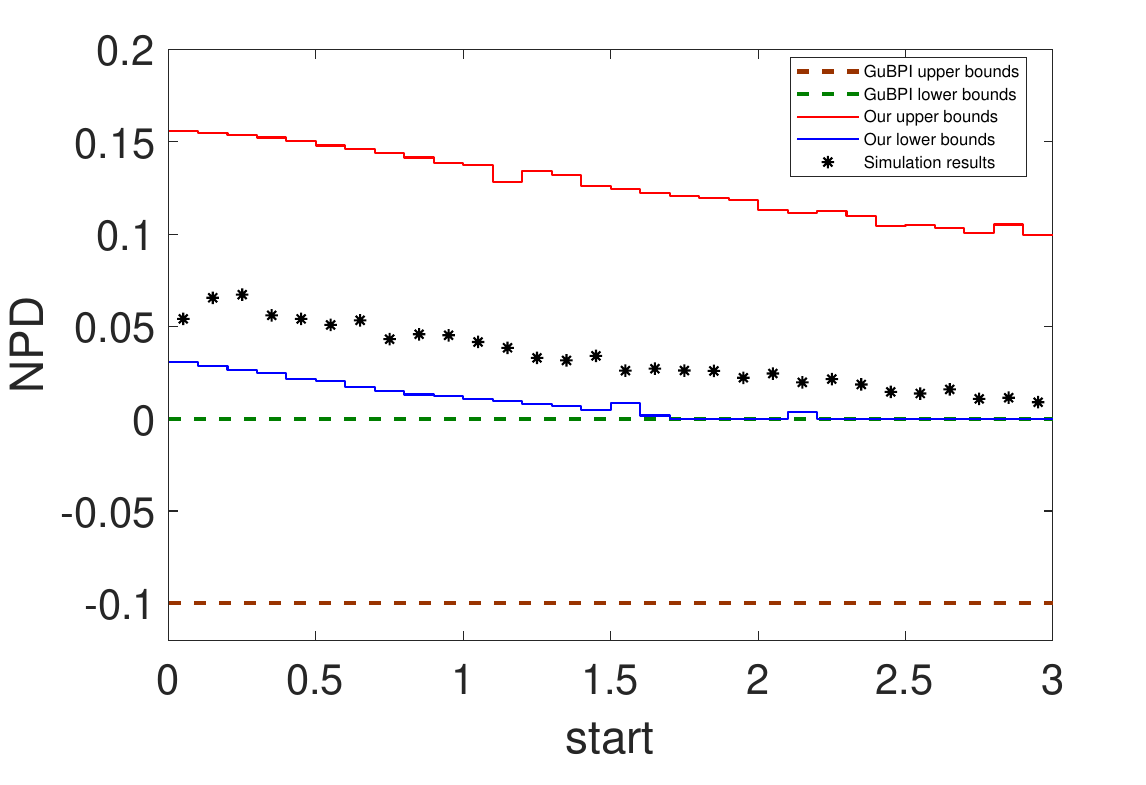}
		}
	}
	\hfill
	\subfigure[\textsc{PdBeta(v2)}]{
		\centering
		\resizebox{0.45\textwidth}{!}{
			\includegraphics[width=2.5in,height=2in]{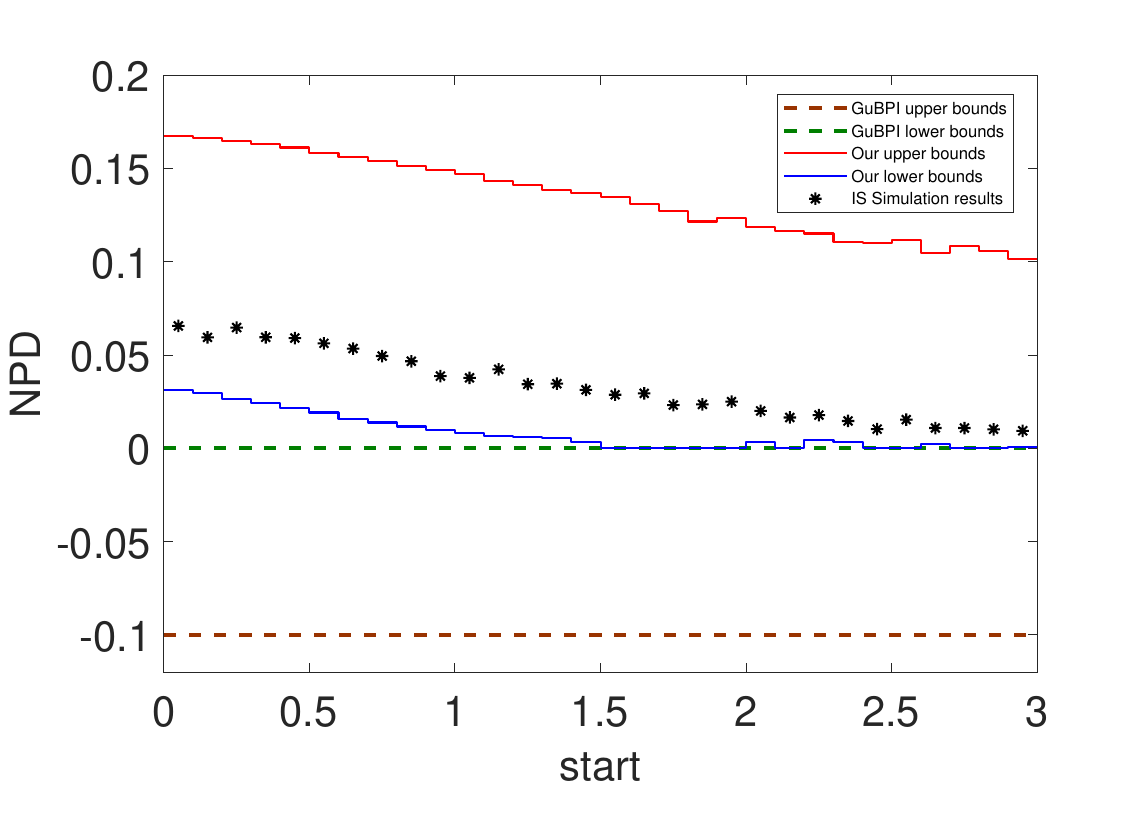}
		}
	}
	\caption{NPD Bounds of Comparison}
	\label{fig:results2}
\end{minipage}

\begin{tablenotes}
    \item The {\color{red} red} and the {\color{blue} blue} lines mark the upper and lower bounds of our results; the black bold stars mark the simulation results; the {\color{brown} brown} and {\color{teal} green} dotted lines mark the upper and lower bounds generated by GuBPI (we denote by $-0.1$ the infinity bounds).
\end{tablenotes}

\vspace{-3ex}
\end{figure}

\smallskip 
\noindent {\em NPD - Comparision with GuBPI \cite{Beutner2022b}.} Since the parameters used in GuBPI and our approach are completely different, it is infeasible to compare the two approaches directly. Instead, we choose the parameters to our algorithms that can achieve at least comparable results with GuBPI. The main parameters are shown in \cref{table:2}.
We consider the Pedestrian example ``\textsc{Pd}'' from \cite{Beutner2022b} (see also \cref{sec3:pedestrian}), and its variants. For the variants of ``\textsc{Pd}'', we enlarged the standard deviation of the observed normal distribution to be $5$ \textbf{in all $6$ variants} whose prefix name are ``\textsc{Pd}''; for the four ``\textsc{PdBeta}'' examples, we also adjust the original uniform sampling $\mathbf{uniform}(0,1)$ in the loop body by different beta distributions. The main purpose to introduce these variants is to test the robustness of our approach.
The last example is from \cite{DBLP:conf/cav/GehrMV16}. 

We report the results in \cref{table:2} whose layout is similar to \cref{table:1} except that the column ``\#'' displays whether or not the bounds are trivial, i.e., $[0,\infty]$.
We also compare our results with GuBPI's and simulation results ($10^6$ samples per case), and show part of the comparison in \cref{fig:results2}, see \cref{app:experiments} for other figures.  Our runtime is up to $15$ times faster than GuBPI while we can still obtain tighter or comparable bounds for all examples.  Specifically, for the first example ``\textsc{Pd}'', our upper bounds are a bit higher than GuBPI's when the value of $start$ falls into $[0,0.7]$ (which is not suprising as the deviation of the normal distribution in this example is quite small, i.e., $0.1$, and our approach constructs over-approximation constraints while GuPBI uses recursion unrolling to search for the feasible space exhaustively), but our lower bounds are greater than GuBPI's, and our NPD bounds are tighter in the following.\footnote{When the value of $start$ approaches $3$, our NPD bounds is close to zero, but the upper bounds may be lower than zero, which is caused by numerical issues of semi-definite programming. The problem of numerical issues is orthogonal to our work and remains to be addressed in both academic and industrial fields.} Note that the simulation results near $1$ deviates largely from our bounds, for which a possible reason is that \textsc{Pd} is a difficult example whose simulation results can have high variances.
For all $6$ variants of ``\textsc{Pd}'',
%where the deviation of the normal distribution is enlarged, 
our NPD bounds are tighter than GuBPI's, in particular, our upper bounds are much lower than GuBPI's. For the four ``\textsc{PdBeta}'' examples, we found that GuBPI produced zero-valued unnormalised lower bounds, and its results w.r.t. NPD are trivial, i.e., $[0,\infty]$. However, we can produce non-trivial results and our runtime is at least 2 times faster than GuBPI. We believe that the main reason why our approach outperforms GuBPI is that GuBPI has widening that may lose precision, while our approach uses polynomial solving with truncation to achieve better precision.

\begin{table*}
%	\vspace{-1.7em}
	\caption{Comparison with GuBPI}
	
	\label{table:2}
	\resizebox{\textwidth}{!}{
		%\begin{footnotesize}
		\begin{threeparttable}
			\begin{tabular}{|c|c|c|c|c|c|c|}
				\hline
				\multicolumn{1}{|c|}{\multirow{2}{*}{\textbf{Benchmark}}}  &
				\multicolumn{4}{c|}{\multirow{1}{*}{\textbf{Our Tool}}}      &
				\multicolumn{2}{c|}{\multirow{1}{*}{\textbf{GuBPI}}}    \\ \cline{2-7}
				\multicolumn{1}{|c|}{} &   \multicolumn{1}{c|}{\textbf{Parameters}}  & \multicolumn{1}{c|}{\textbf{ Solver}} & 
				\multicolumn{1}{c|}{\textbf{ Time (s)}} &
				\multicolumn{1}{c|}{\solvable} &
				\multicolumn{1}{c|}{\textbf{ Time (s)}}   &    
				\multicolumn{1}{c|}{\solvable}            \\ \hline \hline
                \multirow{1}{*}{\textsc{Pd}} & $d=10$, $m = 60$, $pos, dis \in [0,5]$   &  SDP & $3176.685$   & $\bullet$ & $5266.063$   &   $\bullet$    \\ 
				\hline
				\multirow{1}{*}{\textsc{PdLD}} & $d=6$, $m = 60$, $pos, dis \in [0,5]$   &  LP & $41.99$ & $\bullet$ & $648.151$   &   $\bullet$   \\ 
				\hline
				\multirow{1}{*}{\textsc{PdBeta(v1)}} & $d=6$, $m = 60$,$pos, dis \in [0,5]$   &  LP & $99.86$ & $\bullet$ & $645.055$   &   $\circ$    \\ 
				\hline
				\multirow{1}{*}{\textsc{PdBeta(v2)}} & $d=6$,$m = 60$, $pos, dis \in [0,5]$   &  LP & $228.43$ & $\bullet$ & $653.237$   &   $\circ$   \\ 
				\hline
				\multirow{1}{*}{\textsc{PdBeta(v3)}} & $d=6$,$m = 60$, $pos, dis \in [0,5]$   & LP & $101.36$ & $\bullet$ & $657.645$   &   $\circ$  \\ 
				\hline
				\multirow{1}{*}{\textsc{PdBeta(v4)}} & $d=6$,$m = 60$, $pos, dis \in [0,5]$   &  LP & $208.86$ & $\bullet$ & $686.207$   &   $\circ$   \\ 
				\hline
				\multirow{1}{*}{\textsc{PdMB(v5)}} & $d=6$, $m = 60$,$pos, dis \in [0,5]$   &  LP & $88.41$  & $\bullet$ & $391.772$    &   $\bullet$    \\ 
				\hline
				\multirow{1}{*}{\textsc{Para-recur}} & $d=8$, $m = 60$, $p\in [0,1]$   & LP & $36.61$  & $\bullet$ & $253.728$    &   $\bullet$   \\ 
				\hline
			\end{tabular}
			%\end{footnotesize}
			\begin{tablenotes}
				\footnotesize
				\item[*] $\circ$ marks the trivial bound $[0, \infty]$, while $\bullet$ marks the non-trivial ones.
			\end{tablenotes}
			
	\end{threeparttable}}
%\vspace{-3ex}
\end{table*}

\smallskip 
\noindent {\em Path Probability Estimation.} 
We consider five recursive examples in \cite{Beutner2022b,DBLP:conf/pldi/SankaranarayananCG13}, which were also cited from the PSI repository \cite{DBLP:conf/cav/GehrMV16}. Since all five examples are non-paramteric and with unbounded numbers of loop iterations, PSI cannot handle them as mentioned in \cite{Beutner2022b}.  We estimated the path probability of certain events, i.e., queries over program variables, and thus constructed a new bounded range for each query $Q$ by the conjunction of the corresponding $B$ (see \textbf{Stage 3} in~\cref{sec:6-1}) and query $Q$. The results are reported in \cref{table:3} where the second column corresponds to different queries, and the parameters in the third column are the degree $d$, the number $m$ and the bounded range, respectively.
For the first three examples, we obtained tighter lower bounds than GuBPI and same upper bounds, while our runtime is at least 2 times faster than GuBPI. Moreover, we found a potential error of GuBPI. That is, the fourth example ``\textsc{cav-ex-5}" in \cref{table:3} is an AST program with no scores, which means its normalising constant should be exactly one. However, the upper bound of the normailising constant obtained by GuBPI is smaller than $1$ (i.e., $0.6981$). A simulation using $10^6$ samples yielded the results that fall within our bounds but violate those by GuBPI. Thus, GuBPI possibly omitted some valid program runs of this example and produced wrong results.\footnote{We reported this error to the authors of GuBPI, and the bug was fixed afterward.}  
All our results match the simulation ($10^6$ samples per case).

\begin{table*}
%	\vspace{-1.7em}
	\caption{Results for Path Probability Estimation}
	
	\label{table:3}
	\resizebox{\textwidth}{!}{
		\begin{threeparttable}
			%\begin{footnotesize}
			\begin{tabular}{|c|c|c|c|c|c|c|c|}
				\hline
				\multicolumn{1}{|c|}{\multirow{2}{*}{\textbf{Benchmark}}}  &
				\multicolumn{1}{c|}{\multirow{2}{*}{\textbf{Query}}}      &
				\multicolumn{3}{c|}{\multirow{1}{*}{\textbf{Our Tool}}}      &
				\multicolumn{2}{c|}{\multirow{1}{*}{\textbf{GuBPI}}}  &
				\multicolumn{1}{c|}{\multirow{2}{*}{\textbf{Simul}}} 
				\\ \cline{3-7}
				\multicolumn{1}{|c|}{} & \multicolumn{1}{c|}{} & 
    \multicolumn{1}{c|}{\textbf{Parameters}}  & 
    % \multicolumn{1}{c|}{\textbf{ Solver}} & 
				\multicolumn{1}{c|}{\textbf{ Time (s)}} &
				\multicolumn{1}{c|}{\textbf{ Bounds}} &
				\multicolumn{1}{c|}{\textbf{ Time (s)}}   &    
				\multicolumn{1}{c|}{\textbf{ Bounds}} &            \\ \hline \hline
				\multirow{2}{*}{\textsc{cav-ex-7}}&Q1 & $ 6, 1, [0,30], [0,4] $&    $15.062$  &  $[0.9698,1.0000]$ &  $38.834$   & $[0.7381,1.0000]$&   $0.9938$\\ \cline{2-8}
				&Q2 & $ 6,1, [0,40], [0,4] $&    $16.321$   & $[0.9985,1.0000]$  &  $37.651$   & $[0.7381,1.0000]$ &  $0.9993$ \\ 
				\hline
				\multirow{2}{*}{\textsc{AddUni(L)}}&Q1 & $ 6,1,[0,10], [0,1] $&     $8.85$  &  $[0.9940,1.0000]$ &  $21.064$   & $[0.9375,1.0000]$&   $0.9991$\\ \cline{2-8}
				&Q2 & $ 6,1, [0,15], [0,1] $&   $8.80$   & $[0.9995,1.0000]$  &  $14.941$   & $[0.9375,1.0000]$ &  $0.9999$ \\ 
				\hline
				\multirow{1}{*}{\textsc{RdBox}}&Q1 & $ 4,1, [-0.8,0.8], [0,10] $&     $25.87$  &  $[0.9801,1.0000]$ &  $173.535$   & $[0.9462,1.0000]$&  $0.9999$ \\
				\hline
				\multirow{2}{*}{\textsc{cav-ex-5}  \tnote{*}}&Q1 & $ 6,1, [20,\infty], [0,10] $&     $33.17$  &  $[0.8251,0.9351]$ &  $229.623$   & $[0.5768,0.6374]$&   $0.9098$ \\ \cline{2-8}
				&Q2 & $ 6,1, [20,\infty], [0,20] $&    $38.373$   &  $[0.9405,1.0000]$ &  $224.504$   & $[0.5768,0.6375]$ &  $0.9645$  \\ 
				\hline
				\multirow{2}{*}{\textsc{GWalk}  \tnote{**}}&Q1 & $ 8,1, [1,\infty], [0,0.1] $&    $7.255$  &  $[0.0023,0.0023]$ &  $33.246$   & $[0.0023,0.0024]$&  $0.0023$ \\ \cline{2-8}
				&Q2 & $ 8, 1, [1,\infty], [0,0.2] $&   $8.197$   & $[0.0025,0.0025]$  &  $31.728$   & $[0.0025,0.0025]$ &  $0.0025$  \\ 
				\hline
			\end{tabular}
			%\end{footnotesize}
			\begin{tablenotes}
				\footnotesize
				\item[*] GuBPI's result contradicts ours, and we found GuBPI produces wrong results for this example.
				\item[**] As we care about path probabilities, we compared bounds of unnormalised distributions for this example (the NPD can be derived in the same manner above).
			\end{tablenotes}
			
	\end{threeparttable}}
\vspace{-2ex}
\end{table*}

\section{Related Works}
Below we compare our results with the most related work in the literature.

\paragraph{Approximate methods.} Statistical approaches such as MCMC~\cite{rubinstein2016simulation,gamerman2006markov},  variational inference~\cite{blei2017variational} and 
Hausdorff measure~\cite{DBLP:conf/aistats/RadulA21}
cannot provide formal guarantees for the outcomes w.r.t. posterior distributions in a finite time limit. 
The work \cite{DBLP:journals/corr/abs-2101-01502} proposes a novel sampling framework by combining control-data separation and logical condition propagation, which is actually approximate methods. 
In contrast, our approach has formal guarantees on the derived NPD bounds.

\paragraph{Guaranteed NPD inference.}  Most works on guaranteed inference (such as  $(\lambda)$PSI~\cite{DBLP:conf/cav/GehrMV16,DBLP:conf/pldi/GehrSV20}, AQUA~\cite{DBLP:conf/atva/HuangDM21}, Hakaru~\cite{DBLP:conf/flops/NarayananCRSZ16}, SPPL~\cite{DBLP:conf/pldi/SaadRM21}, etc.) are restricted to specific kinds of programs, e.g., programs with closed-form solutions to NPD or without continuous distributions, and none of them can handle unbounded while-loops/recursion.
The most relevant work ~\cite{Beutner2022b} infers the NPD bounds by recursion unrolling. Our approach circumvents the path explosion problem from recursion unrolling by polynomial solving, and outperforms this approach over various benchmarks.  
Several recent works~\cite{oopsla24,DBLP:conf/nips/ZaiserMO23} consider efficient posterior inference in Bayesian Probabilistic Programming. Note that PGF-based inference methods~\cite{oopsla24,DBLP:conf/nips/ZaiserMO23} can only be applied in discrete probabilistic programs. \cite{oopsla24} presents an extended denotational semantics for discrete Bayesian probabilistic while loops and performs exact inference for loop-free programs and a syntactic class of AST programs, while~\cite{DBLP:conf/nips/ZaiserMO23}  calculates the PGFs in Bayesian probabilistic programs efficiently using approximation and automatic differentiation, but cannot handle while loop. Our methods can handle continuous Bayesian probabilistic programs and derive tight bounds for posterior distributions.

\paragraph{Static analysis of probabilistic programs.} In recent years, there have been an abundance of works on the static analysis of probabilistic programs. Most of them address fundamental aspects such as 
termination~\cite{DBLP:conf/cav/ChakarovS13,DBLP:conf/popl/ChatterjeeFNH16,DBLP:conf/vmcai/FuC19}, sensitivity~\cite{DBLP:journals/pacmpl/BartheEGHS18,DBLP:journals/pacmpl/WangFCDX20}, expectation~\cite{DBLP:conf/pldi/NgoC018,cost2019wang,DBLP:conf/tacas/BatzCJKKM23}, tail bounds~\cite{kura2019tail,DBLP:conf/pldi/Wang0R21,wang2022tail}, assertion probability~\cite{DBLP:conf/pldi/SankaranarayananCG13,DBLP:conf/pldi/WangS0CG21}, etc. Compared with these results, we have that: 
(a) Our work focuses on normalised posterior distribution in Bayesian probabilistic programming, and hence is an orthogonal objective. 
(b) Although our algorithms follow the previous works on polynomial template solving~\cite{DBLP:conf/cav/ChakarovS13,cost2019wang,DBLP:journals/toplas/ChatterjeeFNH18,ChatterjeeFG16}, we have a truncation operation to increase the accuracy which to our best knowledge is novel.
(c) Our approach extends the classical OST as the previous works~\cite{cost2019wang,DBLP:conf/pldi/Wang0R21} do, but we consider a multiplicative variant, while the work ~\cite{cost2019wang} considers only additive variants.

\begin{acks}
We thank the anonymous reviewers of PLDI 2024 for their valuable comments and helpful suggestions. This work was supported by the National Research Foundation, Singapore, under its RSS Scheme (NRF-RSS2022-009), the Engineering and Physical Sciences
Research Council (EP/T006579/1), and the National Natural Science Foundation of China (NSFC) under Grant No. 62172271. 
\end{acks}

\clearpage
\bibliographystyle{ACM-Reference-Format}
\bibliography{references}

\clearpage
\appendix
\section{Supplementary Material for Section~\ref{sec:prelim}}\label{app:prelim}

\subsection{Basics of Probability Theory}
A \emph{measurable space} is a pair $(U,\Sigma_U)$ where $U$ is a nonempty set and $\Sigma_U$ is a $\sigma$-algebra on $U$, i.e., a family of subsets of $U$ such that $\Sigma_U\subseteq \mathcal{P}(U)$ contains $\emptyset$ and is closed under complementation and countable union. Elements of $\Sigma_U$ are called \emph{measurable} sets. A function $f$ from a measurable space $(U_1,\Sigma_{U_1})$ to another measurable space $(U_2,\Sigma_{U_2})$ is \emph{measurable} if $f^{-1}(A)\in\Sigma_{U_1}$ for all $A\in\Sigma_{U_2}$.

A \emph{measure} $\mu$ on a measurable space $(U,\Sigma_U)$ is a mapping from $\Sigma_U$ to $[0,\infty]$ such that (i) $\mu(\emptyset)=0$ and (ii) $\mu$ is countably additive:
for every pairwise-disjoint set sequence $\{A_n\}_{n\in\Nset}$ in $\Sigma_U$, it holds that $\mu(\bigcup_{n\in\Nset}A_n)=\sum_{n\in\Nset}\mu(A_n)$. We call the triple $(U,\Sigma_U,\mu)$ a \emph{measure space}. 
If $\mu(U)=1$, we call $\mu$ a \emph{probability measure}, and $(U,\Sigma_U,\mu)$ a \emph{probability space}.
The Lebesgue measure $\lambda$ is the unique measure on $(\Rset,\Sigma_{\Rset})$ satisfying $\lambda([a,b))=b-a$ for all valid intervals $[a,b)$ in $\Sigma_{\Rset}$. For each $n\in\Nset$, we have a measurable space $(\Rset^n,\Sigma_{\Rset^n})$ 
and
a unique product measure $\lambda_n$ on $\Rset^n$ satisfying $\lambda_n(\prod_{i=1}^n A_i)=\prod_{i=1}^n \lambda(A_i)$ for all $A_i\in\Sigma_{\Rset}$.

The \emph{Lebesgue} integral operator $\int$ is a partial operator that maps a measure $\mu$ on $(U,\Sigma_U)$ and a real-valued function $f$ on the same space $(U,\Sigma_U)$ to a real number or infinity, which is denoted by $\int f \mathrm{d}\mu$ or $\int f(x)\mu(\mathrm{d}x)$. 
The detailed definition of Lebesgue integral is somewhat technical, see \cite{rankin1968real,rudin1976principles} for more details. 
Given a measurable set $A\in\Sigma_U$, the integral of $f$ over $A$ is defined by $\int_A f(x)\mu(\mathrm{d} x):=\int f(x) \cdot [x\in A] \mu(\mathrm{d}x)$ where $[-]$ is the Iverson bracket such that $[\phi]=1$ if $\phi$ is true, and $0$ otherwise. If $\mu$ is a probability measure, then we call the integral as the \emph{expectation} of $f$, denoted by $\expectdist{x\sim\mu;A}{f}$, or $\expv[f]$ when the scope is clear from the context.

For a measure $v$ on $(U,\Sigma_U)$, a measurable function $f:U\to \Rset_{\ge 0}$ is the \emph{density} of $v$ with respect to $\mu$ if $v(A)=\int f(x)\cdot [x\in A] \mu(\mathrm{d} x)$ for all measurable $A\in\Sigma_U$, and $\mu$ is called the \emph{reference measure} (most often $\mu$ is the Lebesgue measure). Common families of probability distributions on the reals, e.g., uniform, normal distributions, are measures on $(\Rset,\Sigma_{\Rset})$. Most often these are defined in terms of probability density functions with respect to the Lebesgue measure. That is, for each $\mu_D$ there is a measurable function $\text{pdf}_D:\Rset\to\Rset_{\ge 0}$ that determines it: $\mu_D(A):=\int_A \text{pdf}_D (\mathrm{d}\lambda) $. As we will see, density functions such as $\text{pdf}_D$ play an important role in Bayesian inference.

Given a probability space $\pspace$, a \emph{random variable} is an $\mathcal{F}$-measurable function $X: \Omega \rightarrow \Rset \cup \{+\infty,-\infty\}$. The expectation of a random variable $X$, denoted by $\expv(X)$, is the Lebesgue integral of $X$ w.r.t. $\probm$, i.e., $\int X\,\mathrm{d}\probm$. A \emph{filtration} of $\pspace$ is an infinite sequence $\{ \mathcal{F}_n \}_{n=0}^{\infty}$ such that for every $n\ge 0$, the triple $(\Omega, \mathcal{F}_n, \probm)$ is a probability space and $\mathcal{F}_n \subseteq \mathcal{F}_{n+1} \subseteq \mathcal{F}$. A \emph{stopping time} w.r.t. $\{ \mathcal{F}_n \}_{n=0}^{\infty}$ is a random variable $T: \Omega \rightarrow \Nset \cup \{0, \infty\}$ such that for every $n \geq 0$, the event \{$T \leq n$\} is in $\mathcal{F}_n$. 

A \emph{discrete-time stochastic process} is a sequence $\Gamma = \{X_n\}_{n=0}^\infty$ of random variables in $\pspace$. The process $\Gamma$ is \emph{adapted} to a filtration $\{ \mathcal{F}_n \}_{n=0}^{\infty}$, if for all $n \geq 0$, $X_n$ is a random variable in $(\Omega, \mathcal{F}_n, \probm)$. A discrete-time stochastic process $\Gamma=\{X_n\}_{n=0}^\infty$ adapted to a filtration $\{\mathcal{F}_n\}_{n=0}^\infty$ is a \emph{martingale} (resp. \emph{supermartingale}, \emph{submartingale})
if for all $n \geq 0$, $\expv(|X_n|)<\infty$ and it holds almost surely (i.e.,~with probability $1$) that
$\condexpv{X_{n+1}}{\mathcal{F}_n}=X_n$ (\mbox{resp. } $\condexpv{X_{n+1}}{\mathcal{F}_n}\le X_n$, $\condexpv{X_{n+1}}{\mathcal{F}_n}\ge X_n$).
See~\cite{williams1991probability} for details.
Applying martingales to qualitative and quantitative analysis of probabilistic programs is a well-studied technique~\cite{SriramCAV,ChatterjeeFG16,ChatterjeeNZ2017}.

\subsection{Details for WPTS Semantics}\label{app:wpts-semantics}
We denote by $\Lambda$ the set of all states, by $\Delta$ the set of all weighted states, and by $\Sigma_\Delta$ the product $\sigma$-algebra (on $\Delta$) among the discrete $\sigma$-algebra $(L, 2^{L})$ for locations, the $\sigma$-algebra $\Sigma_{\mathbb{R}^{|\pvars|}}$ for program valuations, and the $\sigma$-algebra $\Sigma_{\mathbb{R}}$ for the multiplicative likelihood weight. 
We define $\Sigma^n_\Delta$ (for $n\ge 1$) as the set $\{A_1\times \dots \times A_n\mid \forall 1\le i\le n. (A_i\in \Sigma_{\Delta})\}$, and $\Delta^\infty$ as the set of all infinite sequences of weighted states.

The probability space for the WPTS $\Pi$ is defined such that its sample space is the set of all program runs, its $\sigma$-algebra is generated by the countable union $\bigcup_{n\ge 1} \{B\times \Delta^\infty\mid B\in \Sigma^n_\Delta\}$), and its probability measure $\probm$ is the unique one such that (i) $\probm(A\times \Lambda^\infty)=\mu_\mathrm{init}(\{\pv\mid (\lin, \pv, 1)\in A\})$ for all $A\in\Sigma_\Delta$, and (ii) $\probm(A\times B\times \Delta^\infty)$ (for every $A\in\Sigma_\Delta, B\in \Sigma^n_\Delta$ ($n\ge 1$)) equals the probability w.r.t the sampling of $\mu_\mathrm{init}$ (for the initial program valuation) and $\rdvarjdis$ (for a sampling valuation in each step until the $(n+1)$-th step) that a program run $\{\Theta_n\}_{n\ge 0}$ is subject to $\Theta_0\in A$ and $(\Theta_1,\dots,\Theta_{n+1})\in B$. 
For each program valuation $\pv$, we denote by $\mathbb{P}_\pv$ the probability measure of $\Pi$ when the initial distribution is changed to the Dirac distribution at $\pv$.

\subsection{Sampling-based Semantics}\label{app:sampling-semantcis}
We recall one prominent semantics in the literature, i.e., the sampling-based semantics~\cite{borgstrom2016lambda,DBLP:conf/lics/StatonYWHK16}. In the next section, we will show that the transition-based semantics in our work is equivalent to the widely-used sampling-based semantics in Bayesian statistical probabilistic programming. 

The sampling-based semantics by Borgstr{\"{o}}m et al.~\cite{borgstrom2016lambda,Beutner2022b} interprets a probabilistic program as a deterministic program parameterized by a sequence of random draws sampled during the execution of the program. 
A \emph{sampling trace} is a finite sequence $\tr= \langle r_1, \dots, r_n \rangle$ of real numbers, and we define $\mathcal{T}:=\bigcup_{n\in\Nset} \Rset^n$ as the set of all sampling traces. 
Given a probabilistic program $P$, a \emph{configuration} $\sigma$ under the semantics is a tuple $\langle \pv,S,w,\tr \rangle$ where $\pv\in \val{\mathrm{p}}$ 
, $S$ is the statement to be executed, $w\in [0,\infty)$ is the multiplicative likelihood weight variable whose value expresses how well the current computation matches the observations, and $\tr$ is a sampling trace. We denote by $\Sigma$ the set of all configurations. 

The semantics operates on the configurations, where an execution of the program is initialized with $\sigma_0=\tuple{\pv_0,P, 1,\tr}$, and the termination configurations have the form of $\tuple{\_,\textbf{skip}, \_, []}$, for which $\_$ is a ``wildcard'' character that matches everything and $[]$ represents an empty set. 
\cref{fig:OperationalSemantics} shows the corresponding one-step reduction relation $\to$ (note that $\Downarrow$ is the usual big-step semantics for deterministic Boolean and arithmetic expressions, so we omit it here). We do not give the reduction rule for the probabilistic branching statement in our PPL (see~\cref{fig:syntax}) as it can be represented by a sequential composition of a boolean statement and a conditional statement.

Let $\to^*$ be the reflexive transitive closure of the one-step reduction $\to$ in \cref{fig:OperationalSemantics}. Given a probabilistic program $P$, we call a sampling trace $\tr$ \emph{terminating} if $\tuple{\pv,P,1,\tr}\to^* \tuple{\pv',\text{skip},w,[]}$ for some valuations $\pv, \pv'\in \val{\mathrm{p}}$ and weight $w\in \Rset_{\ge 0}$, i.e.,~the program $P$ terminates under the samples drawn as in $\tr$.

\begin{figure}
 	
 	\begin{minipage}{0.5\columnwidth}
 		\begin{scprooftree}{1}
 			\AxiomC{}
 			\UnaryInfC{$\tuple{\pv,\textbf{skip}, w,\tr} \to  \tuple{
 					\pv,\textbf{skip},w,\tr
 				} $}
 		\end{scprooftree}
 	\end{minipage}%
 	\begin{minipage}{0.5\columnwidth}
 		\begin{scprooftree}{1}
 			\AxiomC{$\pv\vdash E \Downarrow r$}
 			\UnaryInfC{$ \tuple{\pv,x:=E,w,\tr}  \to  \tuple{\pv[x\mapsto r],\textbf{skip},w,\tr}$}
 		\end{scprooftree}
 	\end{minipage}

 	\begin{minipage}{1\columnwidth}
 		\begin{scprooftree}{1}
 			\AxiomC{$\pv\vdash B\Downarrow b$, $b=\textbf{true}$}
 			\UnaryInfC{$\tuple{\pv, \textbf{if}\, B\, \textbf{then}\, S_1\ \textbf{else}\ S_2\ \textbf{fi},w,\tr } \to  	\tuple{
 					\pv,S_1,w,\tr
 				} $}
 		\end{scprooftree}
 	\end{minipage}%
  
 	\begin{minipage}{1\columnwidth}
 		\begin{scprooftree}{1}
 			\AxiomC{$\pv\vdash B\Downarrow b$, $b=\textbf{false}$}
 			\UnaryInfC{$\tuple{\pv, \textbf{if}\, B\, \textbf{then}\, S_1\ \textbf{else}\ S_2\ \textbf{fi},w  ,\tr} \to  	\tuple{
 					\pv,S_2,w,\tr
 				} $}
 		\end{scprooftree}
 	\end{minipage}
 	
 		\begin{prooftree}
 			\AxiomC{$w'=\text{pdf}_D(r)\ge 0$}
 			\UnaryInfC{$\tuple{\pv,x:=\textbf{sample}\ D,w,r::\tr}
 				\rightarrow
 				\tuple{\pv[x\mapsto r], \textbf{skip}, w\cdot w',\tr} $}
 		\end{prooftree}
 	
 	\begin{prooftree}
 		\AxiomC{$\pv\vdash B\Downarrow b$, $b=\textbf{true}$}
 		\UnaryInfC{$\tuple{
 				\pv,\textbf{while}\ B\ \textbf{do}\ S\ \textbf{od},w,\tr
 			}
 			\rightarrow
 			\tuple{
 				\pv,S;\textbf{while}\ B\ \textbf{do}\ S\ \textbf{od},w,\tr
 			} $}
 	\end{prooftree}

 	\begin{minipage}{1\columnwidth}
 		\begin{scprooftree}{1}
 			\AxiomC{$\pv\vdash B\Downarrow b$, $b=\textbf{false}$}
 			\UnaryInfC{$\tuple{
 					\pv,\textbf{while}\ B\ \textbf{do}\ S\ \textbf{od},w,\tr
 				}
 				\rightarrow
 				\tuple{
 					\pv,\textbf{skip},w,\tr
 				} $}
 		\end{scprooftree}
 	\end{minipage}%

   	\begin{minipage}{1\columnwidth}
  \centering
 		\begin{scprooftree}{1}
 			\AxiomC{$w' \geq 0$}
 			\UnaryInfC{$\tuple{
 					\pv,\textbf{score}(w'),w,\tr
 				}
 				\rightarrow
 				\tuple{
 					\pv,
 					\textbf{skip},w\cdot w',\tr
 				} $}
 		\end{scprooftree}
 	\end{minipage}
 	
 	\begin{minipage}{0.5\columnwidth}
 		\begin{scprooftree}{1}
 			\AxiomC{$\tuple{\pv,S_1,w,\tr}\to\tuple{\pv',S'_1,w',\tr'}$, $S'_1\neq \textbf{skip}$}
 			\UnaryInfC{$	\tuple{
 					\pv,S_1;S_2,w,\tr
 				}
 				\rightarrow
 				\tuple{
 					\pv',S'_1;S_2, w',\tr'
 				} $}
 		\end{scprooftree}
 	\end{minipage}%
 	\begin{minipage}{0.5\columnwidth}
 		\begin{scprooftree}{1}
 			\AxiomC{$\tuple{\pv,S_1,w,\tr}\to\tuple{\pv',S'_1,w',\tr'}$, $S'_1=\textbf{skip}$}
 			\UnaryInfC{$	\tuple{
 					\pv,S_1;S_2,w,\tr
 				}
 				\rightarrow
 				\tuple{
 					\pv',S_2, w',\tr'
 				} $}
 			
 		\end{scprooftree}
 	\end{minipage}

 	\begin{prooftree}
 		
 		\AxiomC{}
 		\UnaryInfC{$\tuple{
 				\pv,\textbf{return}\ x,w,\tr
 			}
 			\rightarrow
 			\tuple{
 				\pv,\textbf{skip},w,t
 			} $}
 	\end{prooftree}

 	\caption{One-step reduction for probabilistic programs.  \vspace{-0.3cm}}
 	\label{fig:OperationalSemantics}
\end{figure}

Then the notion of normalised posterior distributions by the sampling-based semantics is given as follows. 
From the one-step reduction rules (in \cref{fig:OperationalSemantics}), we can reason about the global behavior of probabilistic programs in terms of the sampling traces they produce. 
That is, given a probabilistic program $P$, and a terminating trace $\tr$ such that $\tuple{\pv,P,1,\tr}\to^* \tuple{\pv',\text{skip},w,[]}$ for valuations $\pv, \pv'\in\val{\mathrm{p}}$ and weight $w\in \Rset_{\ge 0}$, we define the \emph{value function} $\valueSem P$
and the \emph{weight function} $\weightSem P$ as follows: 
\begin{align}
 \valueSem P(\pv,\tr) &= \begin{cases}
 \pv' &\text{if }\tuple{\pv,P,1,\tr}\to^* \tuple{\pv',\text{skip},w,[]}\\
 \text{unspecified} &\text{otherwise},
 \end{cases}\\
 \weightSem P(\pv,\tr) &= \begin{cases}
 w &\text{if }\tuple{\pv,P,1,\tr}\to^* \tuple{\pv',\text{skip},w,[]}\\
 0 &\text{otherwise}.
 \end{cases}
 \end{align}
 Moreover, we denote the return variable by $\valueSem{P,\mathit{ret}}(\pv,\tr)$, i.e.,~ $\valueSem{P,\mathit{ret}}(\pv,\tr):=\pv'[\mathit{ret}]$.
We also consider the measure space $(\mathcal{T},\Sigma_{\mathcal{T}},\mu_{\mathcal{T}})$ where $\mathcal{T}=\bigcup_{n\in\Nset} \Rset^n$ (as mentioned previously), $\Sigma_{\trans}:=\{\bigcup_{n\in\Nset} U_n\mid U_n\in \Sigma_{\Rset^n}\}$  and $\mu_{\trans}(U):=\sum_{n\in\Nset}\lambda_n(U\cap \Rset^n)$. 
By definition, the measure space $(\mathcal{T},\Sigma_{\mathcal{T}},\mu_{\mathcal{T}})$ specifies the probability values for sets of sampling traces.
 
\paragraph{Posterior Distributions.} Given a probabilistic program $P$, an initial program valuation $\pv\in \val{\mathrm{p}}$ and a measurable set $U\in\Sigma_{\Rset^{|\pvars|}}$, we define the set of terminating traces where the value of the return variable falls into $U$ as
 
 \[
 \mathcal{T}_{P, \pv, U} := \{\tr\in\mathcal{T} \mid \tuple{\pv,P,1,\tr}\to^*\tuple{\pv',\text{skip},w,[]}, \pv'\in U\}
 \]
 and the set of all terminating traces as
 \[
 \mathcal{T}_{P, \pv} := \{\tr\in\mathcal{T} \mid \tuple{\pv,P,1,\tr}\to^*\tuple{\pv',\text{skip},w,[]}\}.
 \]
Note that $\mathcal{T}_{P, \pv} = \mathcal{T}_{P, \pv, \Rset^{|\pvars|}}$. Therefore, we can define the \emph{unnormalised density} w.r.t $P,\pv,U$ as

\begin{align}\label{eq:unnormalised-posterior}
\measureSem{P}_{\pv}(U) := \int_{\mathcal{T}_{P, \pv, U}} \weightSem P(\pv,\tr)  \,\mu_\trans(\mathrm{d} \tr).
\end{align}
That is, the integral takes all traces $\tr$ on which $P$ starts from $\pv$ and evaluates to a valuation in $U$, weighting each $\tr$ with the weight $\weightSem P(\pv,\tr)$ of the corresponding execution. The \emph{normalising constant} is thus defined by 
\begin{align}\label{eq:normalising-constant}
 Z_{P,\pv} := \int_{\mathcal{T}_{P, \pv}} \weightSem P(\pv,\tr)  \,\mu_\trans(\mathrm{d} \tr).
\end{align}
 Therefore, the  \emph{normalised posterior distribution} is defined as $\posterior_P (\pv,U):=\frac{\measureSem{P}_\pv(U)}{Z_{P,\pv}}$.

We call a program $P$ \emph{integrable} if its normalised constant is finite, i.e,~$0<Z_{P,\pv}<\infty$ for any $\pv\in\val{\mathrm{p}}$. 
Given an integrable program, we are interested in deriving lower and upper bounds on the posterior distribution.

\subsection{Equivalence between The Two Semantics}
In this subsection, we prove the equivalence of the sample-based semantic and our WPTS transition-based semantics over posterior distributions. Fix a probabilistic program $P$ and its WPTS $\Pi$.

\begin{lemma}\label{lem:density-composition}
	For all non-negative bounded measurable function $g:\Rset^{|\pvars|}\to \Rset$, and a program state $\Xi=(\loc,\pv)$, and a statement $P:=S_1;S_2$, we have that	
	\begin{align*}
	&\int_{\mathcal{T}_{P, \pv}} \weightSem P(\pv,\tr) \cdot g(\valueSem P(\pv,\tr)) \,\mu_\trans(\mathrm{d} \tr) \\
	& {}=\int_{\mathcal{T}_{S_1, \pv}} \weightSem{S_1}(\pv,\tr)  \,\mu_\trans(\mathrm{d} \tr)   \Big([\valueSem{S_1}(\pv,\tr)=\pv']\\ 
	& \quad {}\cdot \int_{\mathcal{T}_{S_2, \pv'}} \weightSem{S_2}(\pv',\tr') \cdot g(\valueSem{S_2}(\pv',\tr')) \,\mu_\trans(\mathrm{d} \tr') \Big)\\
	\end{align*}
\end{lemma}

The proof is straightforward and resembles that of \cite{LeeYRY20}[Lemma 4.8].

\begin{proposition}\label{prop:posterior-weight}
	For all non-negative bounded measurable function $g:\Rset^{|\pvars|}\to \Rset$, a probabilistic program $P$ and an initial program state $\Xi = (\loc,\pv)\in\Lambda$, we have that
	\begin{align*}
		\expectdist{\pv}{\widehat{w}_T\cdot g(\widehat{\pv}_T)}= \int_{\mathcal{T}_{P, \pv}} \weightSem P(\pv,\tr) \cdot g(\valueSem P(\pv,\tr)) \,\mu_\trans(\mathrm{d} \tr).
	\end{align*}
\end{proposition}

\begin{proof}
	We prove by induction on the structure of statements.
	
	\begin{itemize}
		\item Case $P\equiv ``\textbf{skip}"$.
		\begin{eqnarray*}
			\expectdist{\pv}{\widehat{w}_T\cdot g(\widehat{\pv}_T)} &=& \int \widehat{w}_T(\omega)\cdot g(\widehat{\pv}_T(\omega)) \probm_{\Xi}(\mathrm{d}\omega)  \\
			&=&   g(\pv)  \\
			&=&  \int_{\mathcal{T}_{P, \pv}} [\tr=[]] \cdot g(\pv) \,\mu_\trans(\mathrm{d} \tr) \\
			&=&   \int_{\mathcal{T}_{P, \pv}} \weightSem P(\pv,\tr) \cdot g(\valueSem P(\pv,\tr)) \,\mu_\trans(\mathrm{d} \tr) \\
		\end{eqnarray*}
		\item Case $P\equiv``x:=E"$.
		\begin{eqnarray*}
			\expectdist{\pv}{\widehat{w}_T\cdot g(\widehat{\pv}_T)} &=& \int \widehat{w}_T(\omega)\cdot g(\widehat{\pv}_T(\omega)) \probm_{\Xi}(\mathrm{d}\omega)  \\
			&=&   g(\pv[x\mapsto \llbracket E\rrbracket(\pv)])  \\
			&=&  \int_{\mathcal{T}_{P, \pv}} [\tr=[]] \cdot g(\pv[x\mapsto \llbracket E\rrbracket(\pv)]) \,\mu_\trans(\mathrm{d} \tr)  \\
			&=&   \int_{\mathcal{T}_{P, \pv}} \weightSem P(\pv,\tr) \cdot g(\valueSem P(\pv,\tr)) \,\mu_\trans(\mathrm{d} \tr) \\
		\end{eqnarray*}
		\item Case $P\equiv``x:=\textbf{sample}\ D"$.
		\begin{eqnarray*}
			\expectdist{\pv}{\widehat{w}_T\cdot g(\widehat{\pv}_T)} &=& \int \widehat{w}_T(\omega)\cdot g(\widehat{\pv}_T(\omega)) \probm_{\Xi}(\mathrm{d}\omega)  \\
			&=&  \int  g(\pv[x\mapsto r]) \mu_D (\mathrm{d}r) \\
			%			&=&  \int_{\mathcal{T}_{P, \pv}} [\tr=[]] \cdot g(\pv[x\mapsto r]) \,\mu_\trans(\mathrm{d} \tr)  \\
			&=&   \int_{\mathcal{T}_{P, \pv}} \weightSem P(\pv,\tr) \cdot g(\valueSem P(\pv,\tr)) \,\mu_\trans(\mathrm{d} \tr) \\
		\end{eqnarray*}
		
		\item Case $P\equiv``\textbf{score}(pdf(D,x))"$.
		\begin{eqnarray*}
			\expectdist{\pv}{\widehat{w}_T\cdot g(\widehat{\pv}_T)} &=& \int \widehat{w}_T(\omega)\cdot g(\widehat{\pv}_T(\omega)) \probm_{\Xi}(\mathrm{d}\omega)  \\
			&=& W_\loc(\pv) \cdot g(\pv) \\
			&=&   \int_{\mathcal{T}_{P, \pv}} \weightSem P(\pv,\tr) \cdot g(\valueSem P(\pv,\tr)) \,\mu_\trans(\mathrm{d} \tr) \\
		\end{eqnarray*}
		\item Case $P\equiv``\textbf{return}\ x"$.
		\begin{eqnarray*}
			\expectdist{\pv}{\widehat{w}_T\cdot g(\widehat{\pv}_T)} &=& \int \widehat{w}_T(\omega)\cdot g(\widehat{\pv}_T(\omega)) \probm_{\Xi}(\mathrm{d}\omega)  \\
			&=&    g(\pv)  \\
			&=&  \int_{\mathcal{T}_{P, \pv}} [\tr=[]] \cdot g(\pv) \,\mu_\trans(\mathrm{d} \tr)  \\
			&=&   \int_{\mathcal{T}_{P, \pv}} \weightSem P(\pv,\tr) \cdot g(\valueSem P(\pv,\tr)) \,\mu_\trans(\mathrm{d} \tr) \\
		\end{eqnarray*}
		
		\item Case $P\equiv``\textbf{if}\ B\ \textbf{then}\ S_1\ \textbf{else}\ S_2\ \textbf{fi}"$. Assume the next state corresponding to the \textbf{then}-branch (resp. \textbf{else}-branch) is $\Xi_1=(\loc_1,\pv)$ (resp. $\Xi_2=(\loc_2,\pv)$). Then we obtain that
		\begin{eqnarray*}
			\expectdist{\pv}{\widehat{w}_T\cdot g(\widehat{\pv}_T)} 
			&=&  [\llbracket B \rrbracket(\pv)=true] \cdot \expectdist{\pv}{\widehat{w}_T\cdot g(\widehat{\pv}_T)} +  [\llbracket B \rrbracket(\pv)=false]\cdot \expectdist{\pv}{\widehat{w}_T\cdot g(\widehat{\pv}_T)} \\
			&=&  [\llbracket B \rrbracket(\pv)=true] \cdot \int_{\mathcal{T}_{S_1, \pv}} \weightSem{S_1}(\pv,\tr) \cdot g(\valueSem{S_1}(\pv,\tr)) \,\mu_\trans(\mathrm{d} \tr) \\
			& &	+[\llbracket B \rrbracket(\pv)=false] \cdot \int_{\mathcal{T}_{S_2, \pv}} \weightSem{S_2}(\pv,\tr) \cdot g(\valueSem{S_2}(\pv,\tr)) \,\mu_\trans(\mathrm{d} \tr)\\
			&=&   \int_{\mathcal{T}_{P, \pv}} \weightSem P(\pv,\tr) \cdot g(\valueSem P(\pv,\tr)) \,\mu_\trans(\mathrm{d} \tr) \\
		\end{eqnarray*}

		\item Case $P\equiv``S_1;S_2"$.
		\begin{eqnarray*}
			\expectdist{\pv}{\widehat{w}_T\cdot g(\widehat{\pv}_T)}
			&=& \int \widehat{w}_T(\omega) \probm_{\Xi}(\mathrm{d}\omega) \int [\omega_T=\Xi'] \cdot \widehat{w}_T(\omega')\cdot g(\widehat{\pv}_T(\omega')) \probm_{\Xi'}(\mathrm{d}\omega') \\
			&=& \int \widehat{w}_T(\omega) \probm_{\Xi}(\mathrm{d}\omega) \left([\omega_T=\Xi']\cdot \int  \cdot \widehat{w}_T(\omega')\cdot g(\widehat{\pv}_T(\omega')) \probm_{\Xi'}(\mathrm{d}\omega')\right) \\
			&=&  \int \widehat{w}_T(\omega) \probm_{\Xi}(\mathrm{d}\omega) \left([\omega_T=\Xi']\cdot \int_{\mathcal{T}_{S_2, \pv'}} \weightSem{S_2}(\pv',\tr') \cdot g(\valueSem{S_2}(\pv',\tr')) \,\mu_\trans(\mathrm{d} \tr') \right)\\
			&=&  \int_{\mathcal{T}_{S_1, \pv}} \weightSem{S_1}(\pv,\tr)  \,\mu_\trans(\mathrm{d} \tr)     \Big([\valueSem{S_1}(\pv,\tr)=\pv']\\ 
			&& \cdot \int_{\mathcal{T}_{S_2, \pv'}} \weightSem{S_2}(\pv',\tr') \cdot g(\valueSem{S_2}(\pv',\tr')) \,\mu_\trans(\mathrm{d} \tr') \Big)\\
			&=&   \int_{\mathcal{T}_{P, \pv}} \weightSem P(\pv,\tr) \cdot g(\valueSem P(\pv,\tr)) \,\mu_\trans(\mathrm{d} \tr) \\
		\end{eqnarray*}
		Here $\omega_i$ is the $i$-th element of the sequence $\omega=\{\Xi_n\}_{n\in\Nset}$, i.e., $\omega_i:=\Xi_i$. $\omega_T$ is the last element of $\omega$, and $\Xi'=(\loc',\pv')$. The third and fourth equalities follow from the induction hypothesis, and the last equality from \cref{lem:density-composition}.
		\item Case $P\equiv``\textbf{while}\ B\ \textbf{do}\ S\ \textbf{od}"$. Assume the next state corresponding to the entry of the loop (resp. the exit of the loop) is $\Xi_1=(\loc_1,\pv)$ (resp. $\Xi_2=(\loc_2,\pv)$). Then we obtain that
		\begin{eqnarray*}
			\expectdist{\pv}{\widehat{w}_T\cdot g(\widehat{\pv}_T)} 
			&=&  [\llbracket B \rrbracket(\pv)=true] \cdot \expectdist{\Xi_1}{\widehat{w}_T\cdot g(\widehat{\pv}_T)} +  [\llbracket B \rrbracket(\pv)=false]\cdot \expectdist{\Xi_2}{\widehat{w}_T\cdot g(\widehat{\pv}_T)} \\
			&=&  [\llbracket B \rrbracket(\pv)=true] \cdot \int \widehat{w}_T(\omega)\cdot g(\widehat{\pv}_T)\probm_{\Xi_1}(\mathrm{d}\omega) +[\llbracket B \rrbracket(\pv)=false] \cdot g(\pv)  \\
			&=&  [\llbracket B \rrbracket(\pv)=true] \cdot \int_{\mathcal{T}_{S;P, \pv}} \weightSem{S;P}(\pv,\tr) \cdot g(\valueSem{S;P}(\pv,\tr)) \,\mu_\trans(\mathrm{d} \tr) \\
			& &	+[\llbracket B \rrbracket(\pv)=false] \cdot  \int_{\mathcal{T}_{\textbf{skip}, \pv}} [\tr=[]] \cdot g(\valueSem{\textbf{skip}}(\pv,\tr)) \,\mu_\trans(\mathrm{d} \tr) \\
			&=&   \int_{\mathcal{T}_{P, \pv}} \weightSem P(\pv,\tr) \cdot g(\valueSem P(\pv,\tr)) \,\mu_\trans(\mathrm{d} \tr) \\
		\end{eqnarray*}
	\end{itemize}
\end{proof}

\begin{theorem}\label{thm:posterior-weight}
	Given a probabilistic program $P$ and its WPTS $\Pi$, an initial program state $\Xi=(\loc,\pv)\in\Lambda$ and a measurable set $U\in\Sigma_{\Rset}$, it holds that $\expectdist{\pv}{\widehat{w}_T\cdot [\widehat{\pv}_T\in \calU]}=\llbracket \Pi \rrbracket_{\pv}(\calU) = \measureSem{P}_{\pv}(\calU)$. Moreover, the expected weight $\expectdist{\pv}{\widehat{w}_T}$ is equivalent to the normalising constant $Z_{P,\pv}$.
\end{theorem}

\begin{proof}
	We instantiate \cref{prop:posterior-weight} with $g(x)=[x\in \calU]$. For any initial program state $\Xi=(\loc,\pv)$, we have that
	\begin{eqnarray*}
		\expectdist{\pv}{\widehat{w}_T\cdot [\widehat{\pv}_T\in \calU]} 
		&=&  \int_{\mathcal{T}_{P, \pv}} \weightSem P(\pv,\tr) \cdot [\valueSem P(\pv,\tr)\in \calU] \,\mu_\trans(\mathrm{d} \tr)  \\
		&=&  \int_{\mathcal{T}_{P, \pv, U}} \weightSem P(\pv,\tr)  \,\mu_\trans(\mathrm{d} \tr) \\
		&=& \measureSem{P}_{\pv}(\calU)
	\end{eqnarray*}
	We prove $\expectdist{\pv}{\widehat{w}_T}=Z_{P,\pv}$ by setting $\calU=\Rset^{|\pvars|}$.
\end{proof}

By \cref{thm:posterior-weight}, we show the equivalence of sampling-based semantics and transition-based semantics over posterior distributions. 

\subsection{Supplementary Details for $\llbracket \Pi \rrbracket(\calU)$ }\label{app:sec2-prop}
\begin{proposition}\label{prop:unnorm-norm}
   Given a WPTS $\Pi$ in the form of \eqref{eq:wpts}, a measurable set $\calU\in\Sigma_{\Rset^{|\pvars|}}$ and the WPTS $\Pi_\calU$ constructed as above, we have that $\llbracket \Pi \rrbracket_{\pv}(\calU)=\llbracket \Pi_\calU\rrbracket_{\pv}(\Rset^{|\pvars|})$ for any $\pv\in\calV=\supp{\mu_{\mathrm{init}}}$. Furthermore,
   if there exist intervals $[l_1,u_1],[l_2,u_2]\subseteq [0,\infty)$ such that $\llbracket \Pi_\calU\rrbracket_{\pv}(\Rset^{|\pvars|})\in [l_1,u_1]$ and $\llbracket \Pi\rrbracket_{\pv}(\Rset^{|\pvars|})\in [l_2,u_2 ]$ for any $\pv\in\calV$, then we have two intervals $[l_\calU,u_\calU],[l_Z,u_Z]\subseteq [0,\infty)$ such that the unnormalised posterior distribution $\llbracket \Pi\rrbracket (\calU)\in [l_\calU,u_\calU]$ and the normalising constant $Z_\Pi\in [l_Z,u_Z]$. Moreover, if $\Pi$ is integrable, i.e., $[l_Z,u_Z]\subseteq (0,\infty)$, then we can obtain the NPD $\posterior_{\Pi}(\calU)\in [\frac{l_\calU}{u_Z},\frac{u_\calU}{l_Z}]$.
\end{proposition}

We prove the proposition by dividing it to the following sub-propositions and complete the proofs.

\begin{proposition}\label{prop:norm}
  Given a WPTS in the form of \eqref{eq:wpts}, the interval-bound analysis of $Z_\Pi$ can be reduced to the interval-bound analysis of the expected weight $\llbracket \Pi \rrbracket_{\pv}(\Rset^{|\pvars|})$ for all $\pv\in \calV$.
\end{proposition}
\begin{proof}
Fix a WPTS $\Pi$ in the form of \eqref{eq:wpts}.
By \cref{def:npd}, the normalising constant $Z_\Pi=\measureSem{\Pi}(\Rset^{|\pvars|})=\int_{\calV} \measureSem{\Pi}_{\pv}(\Rset^{\pvars})\cdot \mu_{\mathrm{init}}(\mathrm{d} \pv)$ is the integral of the expected weight function $\llbracket \Pi \rrbracket_{\pv}(\Rset^{|\pvars|})$ over the finite set $\calV=\supp{\mu_{\mathrm{init}}}$. Assume there exist two integrable bound functions $Up,Lw$ such that $Lw(\pv)\le \llbracket \Pi \rrbracket_{\pv}(\Rset^{|\pvars|})\le Up(\pv)$ for all $\pv\in\calV$. 
Then the interval bounds $l_Z,u_Z$ for $Z_\Pi$ can be derived by the integrals of $Up,Lw$ over $\calV$, that is, 
\begin{align*}
l_Z:= \int_{\calV} Lw(\pv) \cdot \mu_{\mathrm{init}}(\mathrm{d} \pv) \le \measureSem{\Pi}(\Rset^{|\pvars|})=\int_{\calV} \measureSem{\Pi}_{\pv}(\Rset^{|\pvars|})\cdot \mu_{\mathrm{init}}(\mathrm{d} \pv)\le \int_{\calV} Up(\pv) \cdot \mu_{\mathrm{init}}(\mathrm{d} \pv)=:u_Z .
\end{align*}
\end{proof}

\begin{proposition}\label{prop:unnorm}
  Given a WPTS in the form of \eqref{eq:wpts} and any measurable set $\calU\subsetneqq \Rset^{|\pvars|}$, the interval-bound analysis of $\llbracket \Pi\rrbracket (\calU)$ can be reduced to the interval-bound analysis of the expected weight $\llbracket \Pi \rrbracket_{\pv}(\Rset^{|\pvars|})$ for all $\pv\in \calV'\subseteq \calV$. 
\end{proposition}
\begin{proof}
Fix a WPTS $\Pi$ in the form of \eqref{eq:wpts} and a measurable set $\calU\subsetneqq \Rset^{|\pvars|}$. By~\cref{def:npd}, $\llbracket \Pi\rrbracket (\calU)=\int_{\calV} \llbracket \Pi\rrbracket_{\pv}(\calU)\cdot\mu_{\mathrm{init}}(\mathrm{d}\pv)$. By the definition of expected weights in \cref{sec2:NPD}, $\llbracket \Pi\rrbracket_{\pv}(\calU)=\expectdist{\pv}{[\widehat{\pv}_T\in \calU]\cdot\widehat{w}_T}$ for any initial program valuation $\pv$. As $\widehat{\pv}_T$ is the random (vector) variable of the program valuation at termination and $\calU\neq \Rset^{|\pvars|}$, it is not possible to track the valuation of $\widehat{\pv}_T$ for each program run starting from $\pv$, which makes it challenging to compute interval bounds for $\llbracket \Pi\rrbracket_{\pv}(\calU)$ directly. However, in practice, we only care about the posterior distribution that whether the return variable falls into some target sets of interest.\footnote{ For the sake of simplicity, we consider that the WPTS has only one return variable $ret\in\pvars$ of interest, but it can be straightforwardly extended to general cases.} That is, the set $\calU$ is defined such that $[\widehat{\pv}_T\in \calU]=1$ iff $\widehat{\pv}_T[ret]\in\calU(ret)$ where $\calU(ret)$ is the element in $\calU$ that corresponds to the set w.r.t. the return variable. Below we distinguish the computation into two cases.
\begin{itemize}
	\item[Case 1.] If the value of $ret$ is determined once and unaffected by loop iterations, then $[\widehat{\pv}_T[ret]\in\calU(ret)]$ is equivalent to $[\valin[ret]\in\calU(ret)]$. We can conclude that $\llbracket \Pi\rrbracket_{\valin}(\calU)=\\ \expectdist{\valin}{[\widehat{\pv}_T\in \calU]\cdot\widehat{w}_T}= \expectdist{\valin}{\widehat{w}_T}=\llbracket \Pi\rrbracket_{\valin}(\Rset^{|\pvars|})$ if $\valin[ret]\in\calU(ret)$, and $0$ otherwise.
	\item[Case 2.] If the value of $ret$ is affected during loop iterations, then we can construct a WPTS $\Pi_\calU$ by adding a conditional branch of the form ``\textbf{if} $\pv_T\notin\calU$ \textbf{then} \textbf{score}($0$) \textbf{fi}'' immediately before the termination of $\Pi$. And we can prove that $\llbracket \Pi\rrbracket_{\valin}(\calU)$ is equivalent to $\llbracket \Pi_\calU\rrbracket_{\valin}(\Rset^{|\pvars|})$ for all $\valin\in\calV$.
\end{itemize}
The correctness of Case 1 above is straightforward to be proved, while the correctness of Case 2 follows from the proposition below.
\end{proof}

\begin{proposition}
	Given a WPTS $\Pi$ in the form of \eqref{eq:wpts} and a measurable set $\calU\subsetneqq \Rset^{|\pvars|}$, for any initial program valuation $\valin\in\calV$, the expected weight $\measureSem{\Pi}_{\valin}(\calU)$ is equivalent to $\measureSem{\Pi_\calU}_{\valin}(\Rset^{|\pvars|})$ where $\Pi_\calU$ is constructed by the method described in Case 2 above.
\end{proposition}

\begin{proof}
Fix a WPTS $\Pi$ in the form of \eqref{eq:wpts}, a measurable set $\calU\subsetneqq \Rset^{|\pvars|}$ and a new WPTS $\Pi_\calU$ constructed by the method described in Case 2. Let the termination location of $\Pi$ be $\lout$ and the termination location of $\Pi_\calU$ be $\lout'$. Define a predicate $\Phi_\calU$ over $\pvars$  such that $\calU=\{\pv\mid \pv\models\Phi_\calU  \}$. Then $\Pi_\calU$ has two additional transitions than $\Pi$, i.e., $\langle \lout, \Phi_\calU, F_{\top}\rangle$ with $F_\top:=\langle \lout',1, \mbox{\sl id},\overline{0} \rangle$ where $\mbox{\sl id}$ is the identity function and $\overline{0}$ is the constant function that always takes the value $0$, and $\langle \lout, \neg\Phi_\calU, F_{\bot}\rangle$ with $F_\bot:=\langle \lout',1, \mbox{\sl id},\overline{1} \rangle$ where $\overline{1}$ is the constant function that always takes the value $1$. Let the termination time of $\Pi$ be $T$ and the termination time of $\Pi_\calU$ be $T'$.
Therefore, for the random variable of the multiplicative likelihood weight at termination, we have that $\widehat{w}_{T'}=[\widehat{\pv}_{T'}\in \calU]\cdot \widehat{w}_{T}+[\widehat{\pv}_{T'}\notin \calU]\cdot 0$.
By the definition of expected weights,
	\begin{align*}
	    \llbracket \Pi\rrbracket_{\valin}(\calU)&=\expectdist{\valin}{[\widehat{\pv}_T\in \calU]\cdot\widehat{w}_T} \\
     &= \expectdist{\valin}{[\widehat{\pv}_T\in \calU]\cdot\widehat{w}_T+[\widehat{\pv}_T\notin \calU]\cdot 0}  \\
     &=\expectdist{\valin}{[\widehat{\pv}_{T'}\in \calU]\cdot \widehat{w}_{T}+[\widehat{\pv}_{T'}\notin \calU]\cdot 0} \\
     &=\expectdist{\valin}{\widehat{w}_{T'}} \\
     &=\expectdist{\valin}{[\widehat{\pv}_{T'}\in\Rset^{|\pvars|}]\cdot\widehat{w}_{T'}}\\
     &= \llbracket \Pi_\calU \rrbracket_{\valin}(\Rset^{|\pvars|})
	\end{align*}
where the third equality is derived from the fact that $\widehat{\pv}_{T'}=\widehat{\pv}_{T}$ and the fourth equality is obtained from the fact that $[\widehat{\pv}_{T'}\in\Rset^{|\pvars|}]\equiv 1$.
\end{proof}

\section{Supplementary Material for Section~\ref{sec:math}}\label{app:sec4}

\subsection{Basics of Fixed Point Theory}\label{app:fixed-point-materials}

We complement some basic concepts of lattice theory here. Given a partial order $\sle$ on a set $K$ and a subset $K' \subseteq K,$ an \emph{upper bound} of $K'$ is an element $u \in K$ that is no smaller than every element of $K'$, i.e.,~$\forall k' \in K'.~k' \sle u.$ Similarly, a \emph{lower bound} for $K'$ is an element $l$ that is no greater than every element of $K',$ i.e.~$\forall k' \in K'.~l \sle k'.$ The \emph{supremum} of $K',$ denoted by $\bigsqcup K'$, is an element $u^* \in K$ such that $u^*$ is an upper-bound of $K'$ and for every upper bound $u$ of $K',$ we have $u^* \sle u.$ Similarly, the \emph{infimum} $\bigsqcap K'$ is a lower bound $l^*$ of $K'$ such that for every lower-bound $l$ of $K',$ we have $l \sle l^*.$ We define $\bot\!:=\!\bigsqcap K$ and $\top\!:=\!\bigsqcup K.$ In general, suprema and infima may not exist.

A partially-ordered set $(K, \sle)$ is called a \emph{complete lattice} if every subset $K'\subseteq K$ has a supremum and an infimum.
Given a partial order $(K, \sle)$, a function $f: K \to K$ is called \textit{monotone} if for every $k_1 \sle k_2$ in $K$, we have $f(k_1) \sle f(k_2).$

Given a complete lattice $(K, \sle),$ a function $f: K \to K$ is called \emph{continuous} if for every increasing chain $k_0 \sle k_1 \sle \ldots$ in $K,$ we have $f(\bigsqcup \{k_n\}_{n=0}^\infty) = \bigsqcup \{f(k_n)\}_{n=0}^\infty,$ and \emph{cocontinuous} if for every decreasing chain $k_0 \sge k_1 \sge \ldots$ of elements of $K,$ we have $f(\bigsqcap \{k_n\}_{n=0}^\infty) = \bigsqcap \{f(k_n)\}_{n=0}^\infty.$ An element $k \in K$ is called a \emph{fixed-point} if $f(k) = k.$ Moreover, $k$ is a \emph{pre fixed-point} if $f(k) \sle k$ and a \emph{post fixed-point} if $k\sle f(k)$. The \emph{least fixed-point} of $f$, denoted by $\lfp f,$ is the fixed-point that is no greater than every fixed-point under $\sle.$ Analogously, the \emph{greatest fixed-point} of $f$, denoted by $\gfp f$, is the fixed-point that is no smaller than all fixed-points.

\begin{theorem}[\textit{Kleene}~\cite{Sangiorgibook}]
\label{thm:kleene}
Let $(K, \sle)$ be a complete lattice and $f: K \to K$ be a continuous function. Then, we have
	$$\textstyle \lfp\ f = {\textstyle \mathop{\bigsqcup}_{i \ge 0}} \left\{f^{(i)}(\bot)\right\}.$$
Analogously, if $f$ is cocontinuous, then we have
	$$\textstyle \gfp\ f = {\textstyle \mathop{\bigsqcap}_{i \ge 0}} \left\{f^{(i)}(\top)\right\}.$$
\end{theorem}

\subsection{Proofs for Our Fixed-Point Approach}\label{app:fixedpoint}

Recall that $\mbox{\sl ew}_\Pi(\lin,\valin)=\llbracket \Pi \rrbracket_{\valin}(\Rset^{|\pvars|})$ (see the definition of $\mbox{\sl ew}_\Pi$ in~\cref{sec:fixed-point}).
We first show that the expected weight function $\mbox{\sl ew}_\Pi$ is a fixed point of $\ewt_\Pi$ in the complete lattice $(\mathcal{K}_M,\le)$ when $M>1$ and the WPTS $\Pi$ is score bounded by $M$. 

\begin{proposition}\label{prop:fixedpoint}
    The expected-weight function $\mbox{\sl ew}$ is a fixed point of the expected-weight transformer $\ewt$. 
\end{proposition}

\begin{proof}
Define the step-bounded weight random variable $\widehat{w}^n_{(\loc, \pv)}$ starting from any program state $\Xi=(\loc,\pv)$  for a step bound $n\in\mathbb{N}$ by 

\[
\widehat{w}^n_{(\loc, \pv)}(\omega)=\begin{cases} \widehat{w}_{(\loc, \pv)}(\omega) & \mbox{if $T(\omega)\le n$} \\ 0 & \mbox{otherwise}\end{cases}\enskip.
\]
Since we always assume that the underlying WPTS is almost-surely terminating, 
it follows that the sequence of random variables $\{\widehat{w}^n\}_{n\in\Nset}$ converges non-decreasingly to $W$. 

Given any program state $\Xi=(\loc,\pv)$ with a unique transition $\tau = \langle \loc, \phi_\tau, f_\tau \rangle$ satisfying $\pv\models\phi_{\tau}$, define the step-bounded expected-weight function $\mbox{\sl ew}^n$ by $\mbox{\sl ew}^n(\loc,\pv)=\expectdist{(\loc,\pv)}{\widehat{w}^n}$. 
Without loss of generality, we assume there is only one fork $f_\tau$ in this transition. Assume the next sampling valuation from $\Xi$ is $\rv_0$ and the next program state is $\Xi'=(\loc',\pv')$, i.e., $\pv' = f_\tau(\pv,\rv_0)$. Following the symbols in \cref{sec:prelim}, we denote the probability space of the WPTS $\Pi$ starting from $(\loc, \pv)$, i.e., the program runs starting from $\Xi=(\loc, \pv)$ as $\pspace_{\Xi}$. By Tonelli-Fubini Theorem, we have that for all $n\ge 0$,
\begin{eqnarray*}
\mbox{\sl ew}^{n+1}(\loc,\pv) & = & \int \widehat{w}^{n+1}_{(\loc, \pv)} \,\mathrm{d}\probm_{\Xi} \\
&=& \int \widehat{w}^{n+1}_{(\loc, \pv)} \,\mathrm{d}(\mathcal{D}_{\rv_0} \times \probm_{\Xi'}) \\
&=& \int \wet \cdot W^n_{(\loc', \pv')}(\omega)\,\mathrm{d}(\mathcal{D}_{\rv_0} \times \probm_{\Xi'}) \\
&=& \int_{\rv_0}\int_{\omega_{\Xi'}} \wet\cdot W^n_{(\loc', \pv')}(\omega)\,\mathrm{d}\probm_{\Xi'}\,\mathrm{d}\mathcal{D}_{\rv_0}\\
&=& \int_{\rv_0}\wet\cdot\left(\int_{\omega_{\Xi'}}\widehat{w}^n_{(\loc', \pv')}(\omega)\,\mathrm{d}\probm_{\Xi'}\right)\,\mathrm{d}\mathcal{D}_{\rv_0}\\
&=& \int_{\rv_0}\wet\cdot \mbox{\sl ew}^n(\loc', \pv') \,\mathrm{d}\mathcal{D}_{\rv_0}\\
&=& \expectdist{\rv_0}{ \wet\cdot \mbox{\sl ew}^n(\loc', \pv')} \\
&=& \ewt(\mbox{\sl ew}^n)(\loc,\pv) 
\end{eqnarray*}
By applying MCT to the both sides of the equality above, we have that 
\[
\mbox{\sl ew}(\loc,\pv)=\ewt(\mbox{\sl ew})(\loc,\pv).
\]
\end{proof}

In order to further show the uniqueness of the fixed point, we prove that $\ewt$ is both continuous and cocontinuous. 

\begin{proposition}\label{prop:continuity} If $M\in [0,\infty)$, then the expected-weight transformer $\ewt:\mathcal{K}_M\to \mathcal{K}_M$ is both continuous and cocontinous.
\end{proposition}

\begin{proof}
	We first prove that $\ewt$ is well-defined. Given an arbitrary $h\in \mathcal{K}_M$, for any $\Xi=(\loc,\pv)\in\Lambda$,
	\begin{itemize}
		\item When $\loc=\lout$, $\ewt(h)(\loc,\pv)=1$.
		\item When $\loc\neq \lout$, for a unique transition $\tau = \langle \loc, \phi_{\tau}, f_{\tau}, \loc' \rangle$ such that $\pv\models\phi_{\tau}$,
		\begin{align*}
		\ewt(h)(\loc,\pv)&=\expectdist{\rv}{h(\loc',f_{\tau}(\pv,\rv))\cdot W(\loc,\pv)} \\
		&\le M\cdot maxscore \\
		&< \infty  \\
		\end{align*}
	\end{itemize}
where $maxscore$ is the maximum of $W$ given any state $\Xi$.
As $W$ is a non-negative function, we can prove that $\ewt(h)(\loc,\pv)\ge 0$. Thus, $\ewt$ is well defined.
Next, we prove that $\ewt$ is monotone. Given any two functions $h_1,h_2\in \mathcal{K}_M$ such that $h_1\le h_2$, by case analysis on $(\loc,\pv)$,
\begin{itemize}
	\item If $\loc=\lout$, $\ewt(h_1)(\loc,\pv)=1=\ewt(h_2)(\loc,\pv)$.
	\item If $\loc\neq\lout$, given a unique transition $\tau = \langle \loc, \phi_{\tau}, f_{\tau}, \loc' \rangle$ such that $\pv\models\phi_{\tau}$,
	\begin{align*}
	\ewt(h_1)(\loc,\pv)&= \expectdist{\rv}{h_1(\loc',f_{\tau}(\pv,\rv))\cdot W(\loc,\pv)} \\&\le \expectdist{\rv}{h_2(\loc',f_{\tau}(\pv,\rv))\cdot W(\loc,\pv)} \\
	&= \ewt(h_2)(\loc,\pv)
	\end{align*}
\end{itemize}
Therefore, $\ewt(h_1)\le \ewt(h_2)$, hence it is monotone.
Then we prove upper continuity of $\ewt$. Choose any increasing chain $h_0\sle h_1\sle h_2\sle \cdots$ and do another case analysis on $(\loc,\pv)$:
\begin{itemize}
	\item If $\loc = \lout$, then 
	\[
	\ewt(\mathop{\bigsqcup}\limits_{n\ge 0}\left\{h_n\right\})(\loc,\pv)=1=\mathop{\bigsqcup}\limits_{n\ge 0}\left\{\ewt(h_n)\right\}(\loc,\pv).
	\]
	\item Otherwise, for a unique transition $\tau = \langle \loc, \phi_{\tau}, f_{\tau} \rangle$ such that $\pv\models\phi_{\tau}$:
	\begin{align*}
	&\ewt(\mathop{\bigsqcup}\limits_{n\ge 0}\left\{h_n\right\})(\loc,\pv) \\
	=& \expectdist{\rv}{ \wet_j(\pv,\rv) \cdot (\mathop{\bigsqcup}\limits_{n\ge 0}\left\{h_n\right\})(\loc',f_{\tau}(\pv,\rv))} \\
	=& \expectdist{\rv}{\mathop{\text{sup}}\limits_{n\ge 0}\left\{h_n (\loc',f_{\tau}(\pv,\rv))\right\}} \\
	=& \expectdist{\rv}{\mathop{\text{lim}}\limits_{n\to \infty}\left\{h_n (\loc',f_{\tau}(\pv,\rv))\right\}} \\
	\overset{\mathrm{MCT}}{=} &\mathop{\text{lim}}\limits_{n\to \infty} \expectdist{\rv}{h_n (\loc',f_{\tau}(\pv,\rv))} \\
	=& \mathop{\text{lim}}\limits_{n\to \infty} \ewt(h_n)(\loc,\pv)\\
	= &\sup_{n\ge 0}\left\{\ewt(h_n)(\loc,\pv)\right\}\\
	= &\mathop{\bigsqcup}\limits_{n\ge 0}\left\{\ewt(h_n)\right\}(\loc,\pv)
	\end{align*}
\end{itemize}
The ``MCT'' above denotes the monotone convergence theorem. A similar argument establishes cocontinuity for integrable $h_0$ and decreasing chains.
\end{proof}

Now the uniqueness follows from Theorem~\ref{thm:kleene}. 

\begin{theorem}\label{thm:lfp}
Let $\Pi$ be a score-at-end WPTS that is score-bounded by a positive real $M>1$.Then the expected-weight function $\mbox{\sl ew}$ is the unique fixed-point of the higher-order function $\ewt$ on the complete lattice $(\mathcal{K}_{M},\le)$.  

\end{theorem}

\begin{proof}
The proof follows similar arguments in \cite[Theorem 4.4]{DBLP:conf/pldi/WangS0CG21}. By Proposition~\ref{prop:continuity}, we have that for every state $\Xi=(\loc,\pv)$, 

\begin{itemize}
\item $\lfp\ \ewt(\loc, \pv) = \lim\limits_{n\rightarrow\infty} \ewt^n(\bot)(\loc, \pv)$, and 
\item $\gfp\ \ewt(\loc, \pv) = \lim\limits_{n\rightarrow\infty} \ewt^n(\top)(\loc, \pv)$. 
\end{itemize}

By the definition of $\ewt_M^n$ and Proposition~\ref{prop:continuity}, we have that 

\begin{itemize}
\item $\ewt^n(\bot)(\loc, \pv) = \expectdist{\Xi}{W\cdot [T\le n]}-M\cdot\mathbb{P}(T>n)$, and 
\item $\ewt^n(\top)(\loc, \pv) = \expectdist{\Xi}{W\cdot [T\le n]}+M\cdot \mathbb{P}(T>n)$. 
\end{itemize}

Recall that we assume the underlying PTS to be almost-surely terminating. Hence, $\lim\limits_{n\rightarrow\infty} \mathbb{P}(T >n) = \mathbb{P}(T=\infty)= 0$. 
It follows that $\lfp\ \ewt(\loc, \pv) = \gfp\ \ewt(\loc, \pv)$, i.e., the fixed point is unique. 
\end{proof}

Theorem~\ref{thm:fix-point-bounds} (Fixed-Point Approach)

$\llbracket \Pi\rrbracket_{\valin} (\Rset^{|\pvars|})\le h(\lin,\valin)$ (resp. $\llbracket \Pi\rrbracket_{\valin} (\Rset^{|\pvars|})\ge h(\lin,\valin)$) for any bounded PUWF (resp. PLWF) $h$ over $\Pi$ and initial state $(\lin,\valin)$.

\begin{proof}
Let $M>1$ be a bound for a PUWF $h$. By \cref{thm:lfp}, the expected weight function $\mbox{\sl ew}_\Pi$ is the unique fixed point of $\ewt_\Pi$. Note that $\mbox{\sl ew}_\Pi(\lin,\valin)=\llbracket \Pi \rrbracket_{\valin}(\Rset^{|\pvars|})$. Then by applying Tarski's Fixed Point Theorem (\cref{thm:tarski}) and the definition of potential weight functions (\cref{def:puwf}), the PUWF satisfying the prefixed-point conditions can serve as the upper bound for the expected weight $\llbracket \Pi \rrbracket_{\valin}(\Rset^{|\pvars|})$ for any initial state $(\lin,\valin)$. The proof for a PLWF is completely dual. 
\end{proof}

\subsection{Classical OST}~\label{classical_OST}
Optional Stopping Theorem (OST) is a classical theorem in martingale theory that characterizes the relationship between the expected values initially and at a stopping time in a supermartingale. Below we present the classical form of OST.

\begin{theorem}[Optional Stopping Theorem (OST) \cite{williams1991probability}]
Let $\{X_n\}_{n=0}^\infty$ be a supermartingale adapted to a filtration $\mathcal{F}=\{\mathcal{F}_n\}_{n=0}^\infty$, and $\kappa$ be a stopping time w.r.t. the filtration $\mathcal{F}$. 
Then the following condition is sufficient to ensure that $\expv\left(|X_\kappa|\right)<\infty$ and  
$\expv\left(X_\kappa\right)\le\expv(X_0)$:
\begin{itemize}
\item (\emph{almost-sure termination}) $\expv(\kappa)<\infty$, and
\item (\emph{bounded difference}) there exists a constant $C>0$ such that for all $n\ge 0$, $|X_{n+1}-X_n|\le C$ holds almost surely.
\end{itemize}	
	
\end{theorem}

\subsection{Proof for the OST Variant}\label{app:ost-variant-proof}

Theorem~\ref{thm:ost-variant} (OST Variant)
Let $\{X_n\}_{n=0}^\infty$ be a supermartingale adapted to a filtration $\mathcal{F}=\{\mathcal{F}_n\}_{n=0}^\infty$, and $\kappa$ be a stopping time w.r.t. the filtration $\mathcal{F}$. 
Suppose that there exist positive real numbers $b_1,b_2,c_1,c_2,c_3$ such that $c_2>c_3$ and
\begin{itemize}
\item[(A1)] $\probm(\kappa>n) \leq c_1 \cdot e^{-c_2 \cdot n}$ for sufficiently large $n \in \Nset$, and
\item[(A2)] for all $n \in \Nset$, $\left\vert X_{n+1}-X_n \right\vert \le b_1\cdot n^{b_2}\cdot e^{c_3\cdot n}$ holds almost surely. 
\end{itemize}
Then we have that $\expv\left(|X_\kappa|\right)<\infty$ and $\expv\left(X_\kappa\right)\le\expv(X_0)$.

\begin{proof}
 For every $n\in\Nset_0$,
	\begin{eqnarray*}
		\left|X_{\kappa\wedge n}\right|&=& \left|X_0+\sum_{k=0}^{\kappa\wedge n-1} \left(X_{k+1}-X_k\right)\right| \\
		&=& \left|X_0+\sum_{k=0}^\infty \left(X_{k+1}-X_k\right)\cdot \mathbf{1}_{\kappa>k\wedge n>k}\right|\\
		&\le& \left|X_0\right|+\sum_{k=0}^\infty \left|\left(X_{k+1}-X_k\right)\cdot \mathbf{1}_{\kappa>k\wedge n>k}\right| \\
		&\le& \left|X_0\right|+\sum_{k=0}^\infty \left|\left(X_{k+1}-X_k\right)\cdot \mathbf{1}_{\kappa>k}\right|\enskip. \\
	\end{eqnarray*}
	
	Then
	\begin{eqnarray*}
		& &   \expv\left(\left|X_0\right|+\sum_{k=0}^\infty \left|\left(X_{k+1}-X_k\right)\cdot \mathbf{1}_{\kappa>k}\right|\right) \\
		&=& \mbox{(By Monotone Convergence Theorem)} \\
		& & \expv\left(\left|X_0\right|\right)+\sum_{k=0}^\infty \expv\left(\left|\left(X_{k+1}-X_k\right)\cdot \mathbf{1}_{\kappa>k}\right|\right) \\
		&=& \expv\left(\left|X_0\right|\right)+\sum_{k=0}^\infty \expv\left(\left|X_{k+1}-X_k\right|\cdot \mathbf{1}_{\kappa>k}\right) \\
		&\le & \expv\left(\left|X_0\right|\right)+\sum_{k=0}^\infty \expv\left( b_1\cdot k^{b_2}\cdot e^{c_3\cdot k}\cdot \mathbf{1}_{\kappa>k}\right) \\
		&=& \expv\left(\left|X_0\right|\right)+\sum_{k=0}^\infty b_1 \cdot k^{b_2} \cdot e^{c_3\cdot k}\cdot \probm\left(\kappa>k\right) \\
		&\le & \expv\left(\left|X_0\right|\right)+\sum_{k=0}^\infty b_1 \cdot k^{b_2}\cdot e^{c_3\cdot k} \cdot c_{1}\cdot e^{-c_{2}\cdot k} \\
		&=& \expv\left(\left|X_0\right|\right)+ b_1 \cdot c_1 \cdot \sum_{k=0}^\infty   k^{b_2}\cdot e^{-(c_2-c_3)\cdot k} \\
		&<& \infty\enskip.
	\end{eqnarray*}
	where the first inequality is obtained by Condition (A2), the second inequality is derived from Condition (A1), and the third inequality stems from the fact that $c_2>c_3$ in the definition.

	Therefore, by Dominated Convergence Theorem
	and the fact that $X_\kappa=\lim\limits_{n\rightarrow\infty} X_{\kappa\wedge n}$ a.s.,
	\[
	\expv\left(X_\kappa\right)=\expv\left(\lim\limits_{n\rightarrow\infty} X_{\kappa\wedge n}\right)=\lim\limits_{n\rightarrow\infty}\expv\left(X_{\kappa\wedge n}\right)\enskip.
	\]
	Finally, the result follows from properties for the stopped process $\{X_{\kappa\wedge n}\}_{n\in\Nset_0}$ that 
	\[
	\expv\left(X_\kappa\right) \le \expv\left(X_0\right)\enskip.
	\]
\end{proof}

\subsection{Proofs for Our OST-Based Approach}\label{app:ost}

\noindent\textbf{Theorem}~\ref{thm:puwf-normalizing}. (OST Approach)
Let $\Pi$ be a bounded-update score-recursive WPTS. 
Suppose that there exist real numbers $c_1>0$ and $c_2>c_3>0$ such that 
\begin{itemize}
\item[(E1)] $\probm(T>n) \leq c_1 \cdot e^{-c_2 \cdot n}$ for sufficiently large $n\in\Nset$, and 
\item[(E2)] for each score function $\wet$ in $\Pi$, we have $|\wet|\le e^{c_3}$. 
\end{itemize}
 Then for any polynomial PUWF (resp. PLWF) $h$ over $\Pi$, we have that $\llbracket \Pi\rrbracket_{\valin} (\Rset^{|\pvars|})\le h(\lin,\valin)$ (resp. $\llbracket \Pi\rrbracket_{\valin} (\Rset^{|\pvars|})\ge h(\lin,\valin)$) for any initial state $(\lin,\valin)$, respectively.  

\begin{proof}
\textbf{Part 1. Upper bounds.}  Consider the WPTS $\Pi$ has a $d$-degree polynomial PUWF $h$ (see~\cref{def:puwf}).
	Define the stochastic process $\{X_n\}_{n=0}^\infty$ as $X_n:=h(\loc_n,\pv_n)$ where $(\loc_n,\pv_n)$ is the program state at the $n-$th step of a program run. Then construct a stochastic process $\{Y_n\}_{n=0}^\infty$ such that $Y_n:=X_n\cdot \prod_{i=0}^{n-1} W_i$ where $W_i$ is the weight at the $i-$th step of the program run. According to Condition (C1), we have that $\expect{X_{n+1}\cdot W_n|\mathcal{F}_n}\le X_n$. Therefore, by the “take out what is known” property of conditional
	expectation (see \cite{williams1991probability}), it follows that
	
	\begin{eqnarray*}
		& &\expect{X_{n+1}\cdot \prod_{i=0}^{n}W_i|\mathcal{F}_n} \le X_n\cdot \prod_{i=0}^{n-1} W_i \\
		&\Leftrightarrow& \expect{Y_{n+1}|\mathcal{F}_n}\le \expect{Y_n}, 
	\end{eqnarray*}
	which means that $\expect{Y_{n+1}}\le \expect{Y_n}$ from the basic property of conditional expectation. By an easy induction on $n$, we have that $\expect{Y_n}\le \expect{Y_0}<\infty$ for all $n\ge 0$, thus the conditional expectation is also taken in the normal sense as each $Y_n$ is indeed integrable. Hence, $\{Y_n\}_{n=0}^\infty$ is a supermartingale. Moreover, we have from the bounded-update property that $|X_{n+1}|\le \zeta \cdot (n+1)^d$ for a real number $\zeta>0$. By definition, we obtain that for sufficiently large $n$,
	\begin{eqnarray*}
		\left|Y_{n+1}-Y_n\right|&=&  \left|X_{n+1}\cdot \prod_{i=0}^{n} W_i-X_n\cdot \prod_{i=0}^{n-1} W_i \right|  \\
		&\le& \left\vert X_{n+1}\cdot \prod_{i=0}^{n} W_i\right\vert+ \left|X_n\cdot \prod_{i=0}^{n-1} W_i\right|\\
		&<& e^{c_3\cdot n}\cdot (|X_{n+1}|+|X_n|) \\
		&\le& e^{c_3\cdot n}\cdot [\zeta \cdot (n+1)^d+\zeta \cdot n^d] \\
		&\le& \lambda \cdot n^d\cdot e^{c_3\cdot n}
	\end{eqnarray*}
	where the first inequality is induced by the triangle inequality, and the second inequality is derived from the bounded stepwise weight condition such that each $W_i\in [0,e^{c_3}]$ and the fact $W_0=1$.
	By applying the OST variant (Theorem \ref{thm:ost-variant}), we obtain that $\expect{Y_T}\le \expect{Y_0}$. By definition and Condition (C2) in \cref{def:puwf},
	\begin{eqnarray*}
		Y_T&=&h(\loc_T,\pv_T)\cdot \prod_{i=0}^{T-1} W_i \\
		&=& h(\lout,\pv_T)\cdot  \prod_{i=0}^{T-1} W_i \\
		&=& \prod_{i=0}^{T-1} W_i \\
	\end{eqnarray*}  
	Finally, we have that $\llbracket \Pi \rrbracket_{\valin}(\Rset^{|\pvars|}) =\expect{\prod_{i=0}^{T-1} W_i}\le \expect{Y_0}=h(\lin,\valin)$.

\textbf{Part 2. Lower bounds.}  Consider the WPTS $\Pi$ has a $d$-degree polynomial PLWF $h$ (see~\cref{def:puwf}).
	Define the stochastic process $\{X_n\}_{n=0}^\infty$ as $X_n:=h(\loc_n,\pv_n)$ where $(\loc_n,\pv_n)$ is the program state at the $n-$th step of a program run. Then construct a stochastic process $\{-Y_n\}_{n=0}^\infty$ such that $-Y_n:=-X_n\cdot \prod_{i=0}^{n-1} W_i$ where $W_i$ is the weight at the $i-$th step of the program run. According to Condition (C1'), we have that $\expect{-X_{n+1}\cdot W_n|\mathcal{F}_n}\le -X_n$. Therefore, by the “take out what is known” property of conditional
	expectation (see \cite{williams1991probability}), it follows that
	
	\begin{eqnarray*}
		& &\expect{-X_{n+1}\cdot \prod_{i=0}^{n}W_i|\mathcal{F}_n} \le -X_n\cdot \prod_{i=0}^{n-1} W_i \\
		&\Leftrightarrow& \expect{-Y_{n+1}|\mathcal{F}_n}\le \expect{-Y_n}, 
	\end{eqnarray*}
	which means that $\expect{-Y_{n+1}}\le \expect{-Y_n}$ from the basic property of conditional expectation. By an easy induction on $n$, we have that $\expect{-Y_n}\le \expect{-Y_0}<\infty$ for all $n\ge 0$, thus the conditional expectation is also taken in the normal sense as each $Y_n$ is indeed integrable. Hence, $\{-Y_n\}_{n=0}^\infty$ is a supermartingale. Moreover, we have from the bounded-update property that $|X_{n+1}|\le \zeta \cdot (n+1)^d$ for a real number $\zeta>0$. By definition, we obtain that for sufficiently large $n$,
	\begin{eqnarray*}
		\left|-Y_{n+1}-(-Y_n)\right| &=& \left|X_{n+1}\cdot \prod_{i=0}^{n} W_i-X_n\cdot \prod_{i=0}^{n-1} W_i \right|  \\
		&\le& \left\vert X_{n+1}\cdot \prod_{i=0}^{n} W_i\right\vert+ \left|X_n\cdot \prod_{i=0}^{n-1} W_i\right|\\
		&<& e^{c_3\cdot n}\cdot (|X_{n+1}|+|X_n|) \\
		&\le& e^{c_3\cdot n}\cdot [\zeta \cdot (n+1)^d+\zeta \cdot n^d] \\
		&\le& \lambda \cdot n^d\cdot e^{c_3\cdot n}
	\end{eqnarray*}
	where the first inequality is induced by the triangle inequality, and the second inequality is derived from the bounded stepwise weight condition such that each $W_i\in [0,e^{c_3}]$ and the fact $W_0=1$.
	By applying the variant of Optional Stopping Theorem (Theorem \ref{thm:ost-variant}), we obtain that $\expect{-Y_T}\le \expect{-Y_0}$, so $\expect{Y_T}\ge \expect{Y_0}$. By definition and Condition (C2') in \cref{def:puwf},
	\begin{eqnarray*}
		-Y_T&=&-h(\loc_T,\pv_T)\cdot \prod_{i=0}^{T-1} W_i \\
		&=& -h(\lout,\pv_T)\cdot  \prod_{i=0}^{T-1} W_i \\
		&=& -\prod_{i=0}^{T-1} W_i \\
	\end{eqnarray*}  
	Finally, we have that $\llbracket \Pi \rrbracket_{\valin}(\Rset^{|\pvars|})=\expect{\prod_{i=0}^{T-1} W_i}\ge \expect{Y_0}=h(\lin,\valin)$.
\end{proof}

\section{Supplementary Material for Section~\ref{sec:algorithm}}\label{app:sec5}

\subsection{Error Analysis for Polynomial Approximation}\label{app:error-analysis}

\begin{theorem}\label{app:score-error}
Let $\Pi$ be a score-at-end WPTS with score functions $g_1,\dots,g_k$ on the transitions to the termination location $\lout$. Suppose we have a non-negative real number $\epsilon$ and polynomials $g'_1,\dots,g'_k$ such that for all $x\in \mbox{\sl exit}(\Pi)$ and $1\le j\le k$, $|g'_j(x)-g_j(x)|\le \epsilon$. Then we have that $| \llbracket \Pi \rrbracket_{\pv}(\Rset^{|\pvars|}) - \llbracket \Pi' \rrbracket_{\pv}(\Rset^{|\pvars|}) | \le \epsilon$ for all initial program valuation $\pv$, where $\Pi'$ is obtained from $\Pi$ by replacing each $g_j$ ($1\le j\le k$) with $g'_j$. 
\end{theorem}

\begin{proof}
By \cref{thm:lfp}, we have that 
$\lim\limits_{n\rightarrow \infty} \mbox{\sl ewt}^n_{\Pi}(\bot)=\mbox{\sl ew}_{\Pi}$ and 
$\lim\limits_{n\rightarrow \infty} \mbox{\sl ewt}^{n}_{\Pi'}=\mbox{\sl ew}_{\Pi'}$. 
Since $|g'_j-g_j|\le \epsilon$ for every $j$, one can perform a straightforward induction on $n$ to prove that for all $n\ge 0$, it holds that that $|\mbox{\sl ewt}^n_{\Pi}(\bot) - \mbox{\sl ewt}^n_{\Pi'}(\bot)|\le  \epsilon$. By the definition of $\llbracket \Pi \rrbracket_{\pv}(\Rset^{|\pvars|})$, we have that $| \llbracket \Pi \rrbracket_{\pv}(\Rset^{|\pvars|}) - \llbracket \Pi' \rrbracket_{\pv}(\Rset^{|\pvars|}) | \le \epsilon$. 
\end{proof}

\subsection{Truncation over WPTS's}\label{app:truncaion}
In the following, We formally define the truncation operation and supplement some descriptions and proofs for \textbf{Stage 3}.

A truncation operation for a WPTS is to restrict the value of every program variable in the WPTS to a prescribed bounded range. We consider that a bounded range for a program variable could be either $[-R,R]$ ($R> 0$), or $[0,R], [-R,0]$ if the value of the program variable is guaranteed to be non-negative or non-positive. 

To present our truncation operation, we define the technical notions of truncation function and truncation approximations.   
A \emph{truncation function} 
$\trunc$ is a function that maps every program variable $x\in\pvars$ to a bounded interval $\trunc(x)$ in $\Rset$ that specifies the bounded range of the variable $x$. We denote by $\Phi_\trunc$ the formula $\bigwedge_{x\in\pvars} x\in \trunc(x)$ for a truncation function $\trunc$. 
A \emph{truncation approximation} is a function $\calM:\mathbb{R}^{|\pvars|}\to [0,\infty)$ such that each $\calM(\pv)$ ($\pv\in \mathbb{R}^{|\pvars|}$) is intended to be an over- or under-approximation of the expected weight $\llbracket \Pi\rrbracket_{\pv} (\Rset^{|\pvars|})$ outside the bounded range specified by $\Phi_\trunc$. The truncation operation is given by the following definition. 

Our main theorem shows that by choosing an appropriate truncation approximation $\calM$ in the truncation, one can obtain upper/lower approximation of the original WPTS. 

\begin{theorem}\label{thm:upperlower}
Let $\Pi$ be a WPTS in the form of \eqref{eq:wpts}, $\trunc$ a truncation function and $\calM$ a bounded truncation approximation.
Suppose that the following condition ($\ast$) holds:
\begin{itemize}
\item[($\ast$)] for each fork $F^{M,\sharp}=\langle \sharp, p, \mbox{\sl upd}, \calM\rangle$ in the truncated WPTS $\Pi_{\trunc,\calM}$ that is derived from
some fork $F=\langle \loc', p, \mbox{\sl upd}, \wet\rangle$ with the source location $\loc$ in the original WPTS (see sentence (b) in Definition~\ref{def:truncation}),  we have that $\llbracket \Pi\rrbracket_{\pv}(\Rset^{|\pvars|})\le \calM(\pv)$ for all $\pv$ such that the state $(\loc,\pv)$ is reachable and $\pv\not\models\Phi_\trunc$. 
\end{itemize}
Then $\llbracket \Pi\rrbracket_{\valin} (\Rset^{|\pvars|})\le \llbracket \Pi_{\trunc,\calM}\rrbracket_{\valin}(\Rset^{|\pvars|})$ for all initial program valuations $\valin$. 
Analogously, if it holds the condition ($\star$) which is almost the same as ($\ast$) except for that ``$\llbracket \Pi\rrbracket_{\pv}(\Rset^{|\pvars|})\le \calM(\pv)$'' is replaced with ``$\llbracket \Pi\rrbracket_{\pv}(\Rset^{|\pvars|})\ge \calM(\pv)$'', then we have $\llbracket \Pi\rrbracket_{\valin} (\Rset^{|\pvars|})\ge \llbracket \Pi_{\trunc,\calM}\rrbracket_{\valin}(\Rset^{|\pvars|})$ for all initial program valuations $\valin$.
\end{theorem}

\begin{proof}
We first prove that when every score function $\mathcal{M}$ in a $F^{\mathcal{M},\sharp}$ derived from a transition with source location $\loc$ is equal to the function $\ewt(\loc,-)$, we have that $\Pi_{\trunc,\mathcal{M}}$ is equal to $\Pi$. 
By \cref{thm:lfp}, the expected weight functions $\mbox{\sl ew}_\Pi$, $\mbox{\sl ew}_{\Pi_{\trunc,\mathcal{M}}}$ are the least fixed point of the higher-order operator $\ewt$ defined in \cref{def:ewt}. 
We prove that both $\mbox{\sl ew}_\Pi\le \mbox{\sl ew}_{\Pi_{\trunc,\mathcal{M}}}$ and  $\mbox{\sl ew}_\Pi\ge \mbox{\sl ew}_{\Pi_{\trunc,\mathcal{M}}}$ holds. Note that since we choose the scoring function to be the exact expected weight function of $\Pi$, it holds that 
$\mbox{\sl ew}_\Pi(-,\pv)=\mbox{\sl ew}_{\Pi_{\trunc,\mathcal{M}}}(-,\pv)$ for all program valuations outside $B$. Thus, the nontrivial part is to consider program valuations inside the truncated range. 

\begin{itemize}
	\item $\mbox{\sl ew}_\Pi\le \mbox{\sl ew}_{\Pi_{\trunc,\mathcal{M}}}$: To show that $\mbox{\sl ew}_\Pi\le \mbox{\sl ew}_{\Pi_{\trunc,\mathcal{M}}}$, it suffices to observe that $\mbox{\sl ew}_{\Pi_{\trunc,\mathcal{M}}}$ satisfies $\ewt_{\Pi}(\mbox{\sl ew}_{\Pi_{\trunc,\mathcal{M}}})=\mbox{\sl ew}_{\Pi_{\trunc,\mathcal{M}}}$. 
    Since $\mbox{\sl ew}_{\Pi}$ is the least fixed point of the higher order equation, we directly obtain that $\mbox{\sl ew}_\Pi\le \mbox{\sl ew}_{\Pi_{\trunc,\mathcal{M}}}$. 
	\item $\mbox{\sl ew}_\Pi\ge \mbox{\sl ew}_{\Pi_{\trunc,\mathcal{M}}}$: To show that $\mbox{\sl ew}_\Pi\ge \mbox{\sl ew}_{\Pi_{\trunc,\mathcal{M}}}$, it suffices to observe that 
	$\mbox{\sl ew}_{\Pi}$ (extended with the $\sharp$ location whose score function is $1$) satisfies the higher-order equation of $\Pi_{\trunc,\mathcal{M}}$. Thus, we directly have that $\mbox{\sl ew}_\Pi\ge \mbox{\sl ew}_{\Pi_{\trunc,\mathcal{M}}}$. 
\end{itemize}

Then we prove the theorem. We only prove the upper-bound case, since the lower-bound case can be proved similarly. The proof follows from \cref{thm:lfp}. Denote $\bot$ as the bottom element of the complete lattice 
$(\mathcal{K}_M, \le)$. Then by \cref{thm:lfp}, we have that $\lim\limits_{n\rightarrow \infty} \ewt_\Pi^n(\bot)=\mbox{\sl ew}_{\Pi}$ and $\lim\limits_{n\rightarrow \infty} \ewt_{\Pi_{\trunc,\mathcal{M}}}^{n}(\bot)=\mbox{\sl ew}_{\Pi_{\trunc,\mathcal{M}}}$. 
Since $\mbox{\sl ew}_{\Pi}(\pv)\le \mathcal{M}(\pv)$ for all $\pv\in \mbox{\sl exit}(\Pi)$ , one can perform a straightforward induction on $n$ that $\mbox{\sl ew}^{n,\mbox{\sl ew}_{\Pi}}_{\Pi_B}\le  \mbox{\sl ew}^{n,f}_{\Pi_B}$ for all $n$. 

\end{proof}

\vspace{2em}

In the case of OST, the proof is as follows.

\begin{theorem}\label{thm:ostupperlower}
Let 
$\Pi$ be a WPTS in the form of \eqref{eq:wpts}, $\trunc$ a truncation function and $\calM$ a polynomial truncation approximation.
Suppose that condition ($\ast$) as in Theorem~\ref{thm:upperlower} holds. 
Furthermore, suppose that the expected weight functions of the original and the truncated WPTS are both bounded by a polynomial in the program variables, and the WPTS has exponentially-decreasing concentration property for its termination time.
Then $\llbracket \Pi\rrbracket_{\valin} (\Rset^{|\pvars|})\le \llbracket \Pi_{\trunc,\calM}\rrbracket_{\valin}(\Rset^{|\pvars|})$ for all initial program valuations $\valin$. 
Analogously, if it holds the condition ($\star$) which is almost the same as ($\ast$) except for that ``$\llbracket \Pi\rrbracket_{\pv}(\Rset^{|\pvars|})\le \calM(\pv)$'' is replaced with ``$\llbracket \Pi\rrbracket_{\pv}(\Rset^{|\pvars|})\ge \calM(\pv)$'', then we have $\llbracket \Pi\rrbracket_{\valin} (\Rset^{|\pvars|})\ge \llbracket \Pi_{\trunc,\calM}\rrbracket_{\valin}(\Rset^{|\pvars|})$ for all initial program valuations $\valin$.
\end{theorem}

\begin{proof}
We first prove that when every score function $\mathcal{M}$ in a $F^{\mathcal{M},\sharp}$ derived from a transition with source location $\loc$ is equal to the function $\ewt(\loc,-)$, we have that $\Pi_{\trunc,\mathcal{M}}$ is equal to $\Pi$. 
By \cref{prop:fixedpoint}, the expected weight functions $\mbox{\sl ew}_\Pi$, $\mbox{\sl ew}_{\Pi_{\trunc,\mathcal{M}}}$
is a fixed point of the higher-order operator $\ewt$ defined in \cref{def:ewt}. 
We prove that both $\mbox{\sl ew}_\Pi\le \mbox{\sl ew}_{\Pi_{\trunc,\mathcal{M}}}$ and  $\mbox{\sl ew}_\Pi\ge \mbox{\sl ew}_{\Pi_{\trunc,\mathcal{M}}}$ holds. Note that since we choose the scoring function to be the exact expected weight function of $\Pi$, it holds that 
$\mbox{\sl ew}_\Pi(-,\pv)=\mbox{\sl ew}_{\Pi_{\trunc,\mathcal{M}}}(-,\pv)$ for all program valuations outside $B$. Thus, the nontrivial part is to consider program valuations inside the truncated range. 

\begin{itemize}
	\item $\mbox{\sl ew}_\Pi\le \mbox{\sl ew}_{\Pi_{\trunc,\mathcal{M}}}$: To show that $\mbox{\sl ew}_\Pi\le \mbox{\sl ew}_{\Pi_{\trunc,\mathcal{M}}}$, it suffices to observe that $\mbox{\sl ew}_{\Pi_{\trunc,\mathcal{M}}}$ satisfies $\ewt_{\Pi}(\mbox{\sl ew}_{\Pi_{\trunc,\mathcal{M}}})=\mbox{\sl ew}_{\Pi_{\trunc,\mathcal{M}}}$. 
By applying \cref{thm:puwf-normalizing} with given polynomial bounds for $\mbox{\sl ew}_\Pi,\mbox{\sl ew}_{\Pi_{\trunc,\mathcal{M}}}$, we obtain that $\mbox{\sl ew}_\Pi\le \mbox{\sl ew}_{\Pi_{\trunc,\mathcal{M}}}$. 
	\item $\mbox{\sl ew}_\Pi\ge \mbox{\sl ew}_{\Pi_{\trunc,\mathcal{M}}}$: To show that $\mbox{\sl ew}_\Pi\ge \mbox{\sl ew}_{\Pi_{\trunc,\mathcal{M}}}$, it suffices to observe that 
	$\mbox{\sl ew}_{\Pi}$ (extended with the $\sharp$ location whose score function is $1$) satisfies the higher-order equation of $\Pi_{\trunc,\mathcal{M}}$. Thus, we directly have that $\mbox{\sl ew}_\Pi\ge \mbox{\sl ew}_{\Pi_{\trunc,\mathcal{M}}}$ by \cref{thm:puwf-normalizing} with given polynomial bounds for $\mbox{\sl ew}_\Pi,\mbox{\sl ew}_{\Pi_{\trunc,\mathcal{M}}}$. 
\end{itemize}

Now similar to the proof for Theorem~\ref{thm:upperlower}, it suffices to show that $\lim\limits_{n\rightarrow \infty} \ewt_\Pi^n(0)=\mbox{\sl ew}_{\Pi}$. First, one observes that $\lim\limits_{n\rightarrow \infty} \ewt_\Pi^n(0)=\mathbb{E}_{\pv}(W\cdot [T\le n])$. Hence, $\mbox{\sl ew}_{\Pi}-\ewt_\Pi^n(0)=\mathbb{E}_{\pv}(W\cdot [T> n])$. By the bounded update property, we have that $\pv'$ is linearly bounded by $n$ for program valuations $\pv'$ reachable after exactly $n$ steps of execution of 
the underlying WPTS. Hence, since the truncation approximation $f$ is polynomial, we have that $\mathbb{E}_{\pv}(W\cdot [T> n])\le f(n)\cdot \mathbb{P}(T>n)$. As we have the exponential-decreasing concentration property for $\mathbb{P}(T>n)$, it follows that 
$\lim\limits_{n\rightarrow\infty}\mathbb{E}_{\pv}(W\cdot [T> n])=0$. 
\end{proof}

The theorem above states that if the truncation approximation gives correct bounds for the expected weights of the original WPTS outside the bounded range, then the bounds for the expected weights of the truncated WPTS are also correct bounds for the expected weights of the original WPTS.

\subsection{Application of Putinar's Positivstellensatz}\label{app:putinar}

We recall Putinar's Positivstellensatz below.

\begin{theorem}[Putinar's Positivstellensatz~\cite{putinar}] \label{thm:putinar} Let $V$ be a finite set of real-valued variables and $g, g_1, \ldots, g_m \in \mathbb{R}[V]$ be polynomials over $V$ with real coefficients. Consider the set $\mathcal{S}:=\{\mathbf{x} \in \mathbb{R}^V\,\mid\, g_i(\mathbf{x}) \geq 0  \mbox{ for all }1\le i\le m \}$ which is the set of all real vectors at which every $g_i$ is non-negative. If (i)~there exists some $g_k$ such that the set $\{ \mathbf{x} \in \mathbb{R}^V ~\mid~ g_k(\mathbf{x}) \geq 0  \}$ is compact and (ii)~$g(\mathbf{x})>0$ for all $\mathbf{x} \in \mathcal{S}$, then we have that 
	\begin{equation} \label{eq:putinar}
	\textstyle g = f_0 + \sum_{i=1}^m f_i \cdot g_i
	\end{equation}
	for some polynomials $f_0,f_1\dots, f_m\in \mathbb{R}[V]$ such that each polynomial $f_i$ is the  a sum of squares (of polynomials in $\mathbb{R}[V]$), i.e.~$f_i = \sum_{j=0}^{k} q_{i,j}^2$ for polynomials $q_{i,j}$'s in $\mathbb{R}[V]$.
\end{theorem}

In this work, we utilize the sound form in \eqref{eq:putinar} for witnessing a polynomial $g$ to be non-negative over a polyhedron $P$  for each constraint $\forall \pv\in P.(g(\pv)\ge 0)$ from \textbf{Step A5} of our algorithm. Let $\forall \pv\in P.(g(\pv)\ge 0)$ be such a constraint for which the polyhedron $P$ is defined by the linear inequalities $g_1\ge 0,\dots,g_m\ge 0$. Let $\pvars= \{ v_1, v_2, \ldots, v_t\}$ be the set of program variables and define $\monomials_d(\pvars)= \{m_1, m_2, \ldots, m_r\}$ as the set of all monomials of degree at most $d$ over $\pvars$, i.e.~$\monomials_d(\pvars) :=  \{ \prod_{i=1}^t v_i^{\alpha_i} ~\mid~ \forall i~~\alpha_i\in \mathbb{N} ~\wedge~ \sum_{i=1}^t \alpha_i \leq d \}$. The application of Putinar's to $\forall \pv\in P.(g(\pv)\ge 0)$ has the following steps. 

\begin{itemize}
		\item First, represent each $f_i$ in \cref{eq:putinar} as the positive semidefinite form $f_i=\mathbf{v}^\mathrm{T} \mathbf{Q}_i \mathbf{v}$ subject to the positive semidefinite constraint where each $\mathbf{Q}_i$ is a real matrix whose every entry is an unknown parameter. 
        \item Second, compute an equation in the form \eqref{eq:putinar} whose coefficients are affine expressions in the unknown coefficients from our templates and the unknown entries in the matrices $\mathbf{Q}_i$'s. 
        \item Third, establish the affine constraints between the unknown coefficients in the templates and the unknown entries in the matrices $Q_i$'s by matching the coefficients at the LHS and the RHS of the equation obtained from the previous step. 
\end{itemize}
The overall application processes all such constraints from \textbf{Step A5} of our algorithm by (i) collecting all the affine and the semidefinite constraints from the first and the third steps above and (ii) solve them by semidefinite programming. 

\subsection{Application of Handelman's Positivstellensatz}\label{app:Handelman}

To present Handelman's Positivstellensatz, we need the notion of monoid as follows. 
Below we consider an arbitrary finite collection $\Gamma=\{g_1,\dots, g_k\}$ ($k\ge 1$) of linear functions (i.e., degree-$1$ polynomials) in the program variables.  

\begin{definition}[Monoid]
The \emph{monoid} of $\Gamma$ is defined by:
\[
\mbox{\sl Monoid}(\Gamma):=\left\{\prod_{i=1}^k h_i \mid k\in\Nset_0\mbox{ and }h_1,\dots,h_k\in\Gamma\right\}~~.
\]
\end{definition}

Then in our context, Handelman's Positivstellensatz can be formulated as follows. 

\begin{theorem}[Handelman's Positivstellensatz~\cite{handelman1988representing}]
\label{thm:handelman}
Let $g$ be a polynomial in the program variables such that $g(\pv)>0$ for all program valuations $\pv\in 
P:=\{\pv'\in \Rset^{|\pvars|}\mid g_1(\pv')\ge 0, \dots, g_k(\pv')\ge 0\}$. 
If $P$ is compact, then we have  
\begin{equation}\label{eq:handelman} 
g=\sum_{i=1}^d a_i\cdot u_i 
\end{equation}
for some $d\in\Nset$, real numbers $a_1,\dots,a_d\ge 0$ and $u_1,\dots,u_d\in\mbox{\sl Monoid}(\Gamma)$. 
\end{theorem}

To apply Handelman's Positivstellensatz, we consider a natural number which serves as a bound on the number of multiplicands allowed to form an element in $\mbox{\sl Monoid}(\Gamma)$.
Then~\cref{eq:handelman} results in a system of linear equalities that involves $a_1,\dots,a_d$ and the coefficents of $g$. The application of Handelman's Positivstellensatz to each $\forall \pv\in P.(g(\pv)\ge 0)$
is simpler than that of Putinar's Positivstellensatz, and is as follows. 

\begin{itemize}
        \item First, compute an equation in the form \eqref{eq:handelman} whose coefficients are affine expressions in the unknown coefficients from our templates and the fresh variables $a_1,\dots,a_d$ from \cref{eq:handelman}. 
        \item Second, establish the affine constraints between the unknown coefficients in the templates and the fresh variables $a_1,\dots,a_d$ from \cref{eq:handelman} by matching the coefficients at the LHS and the RHS of the equation obtained from the previous step. 
\end{itemize}
The overall application processes all such constraints from \textbf{Step A5} of our algorithm by (i) collecting all the affine constraints from the second steps above and (ii) solve them by linear programming. 

\section{Supplementary Materials for Section~\ref{sec:experiment}}\label{app:experiments}

\subsection{Possible Approaches for Computing $M_{\mathrm{up}}$ and $M_{\mathrm{low}}$ of Score-recursive WPTS's}

Fix a score-recursive WPTS $\Pi$, and assume it has (1) the concentration property, i.e., $\probm(T>n) \leq c_1 \cdot e^{-c_2 \cdot n}$ for $c_1,c_2>0$, (2) the bounded-update property, and (3) the stepwise weight is bounded by $e^{c_3}$ for $0<c_3<c_2$. Then given the bounded ranges $B$ and $B'$ as computed in \cref{sec:algorithm}, we derive the upper bound $M_{\mathrm{up}}$ and the lower bound $M_{\mathrm{low}}$ for the expected weight from $B'\backslash B$ as follows.

For any $\pv\in B'\backslash B$,
\begin{eqnarray*}
\llbracket \Pi\rrbracket (\pv) &=& \expectdist{\pv}{w_T}   \\
&=& \sum_{n=1}^{\infty} \probm(T=n)\cdot w_n\\
&\le & \sum_{n=0}^{\infty} \probm(T>n)\cdot w_n \\
&\le& 1+\sum_{n=1}^{\infty} \probm(T>n)\cdot w_n \\
&=& 1+\sum_{n=1}^{n^*-1} \probm(T>n)\cdot w_n +\sum_{n=n^*}^{\infty} \probm(T>n)\cdot w_n \\
&\le& 1+M+\sum_{n=n^*}^{\infty} \probm(T>n)\cdot w_n \\
&=& M'+\sum_{n=n^*}^{\infty} \probm(T>n)\cdot w_n \\
&\le& M'+\sum_{n=n^*}^{\infty} c_1\cdot e^{-c_2\cdot n} \cdot w_n \\
&\le& M'+\sum_{n=n^*}^{\infty} c_1\cdot e^{-c_2\cdot n} \cdot e^{c_3\cdot n} \\
&=& M'+c_1\cdot \sum_{n=n^*}^{\infty}(e^{c_3-c_2})^n \\
&=&M'+ c_1\cdot \frac{a}{1-q} \\
&=& M_{\mathrm{up}} \\
\end{eqnarray*}
where 
\[
\sum_{n=1}^{n^*-1} \probm(T>n)\cdot w_n\le \sum_{n=1}^{n^*-1} w_n\le \sum_{n=1}^{n^*-1} (e^{c_3})^n = \frac{a'\cdot (1-(q')^{(n^*-1)})}{1-q'}=:M
\]
and $a'=e^{c_3}$, $q'=e^{c_3}$.
The first inequality is obtained from the fact that
\[
\probm(T>n)=\probm(T\ge n+1)=\probm(T=n+1)+\probm(T=n+2)+\dots,
\]
thus,
\[
\probm(T=n+1)\le \probm(T>n).
\]
The second inequality is derived by the fact that $\probm(T>0)\le 1$ and $w_0=\win=1$. The third inequality is obtained by the definition of $M$ above. The fourth inequality is obtained by the concentration property, while the fifth inequality is derived by the bounded stepwise weight condition.

For $M_{\mathrm{low}}$, we trivially set $M_{\mathrm{low}}=0$. We can refine it heuristically, e.g., according to the monotonicity of the scoring function.

\subsection{Overapproximation via Polynomial Interpolations}
Given a non-polynomial function $f(x)$ over the interval $I=[a,b]$, we aim to approximate $f(x)$ by polynomials $p(x)$'s. The correctness of approximation is based on a classical theorem called Weierstrass' Theorem \cite{jeffreys1988weierstrass}.

\begin{theorem}[Weierstrass' Theorem]\label{thm:Wierstrass}
Let $f(x)$ be a continuous function on the (closed) interval $[a,b]$. Then there is a sequence of polynomials $p_n(x)$ (of degree $n$) such that
\[
\mathop{lim}_{n\to\infty} ||f-p_n||_{\infty}=0.
\]
	
\end{theorem}

We also need the following theorem to measure the derived polynomials. The property of Lipschitz continuity supports the following theorem easily.
\begin{theorem}\label{thm:errorbound}
	Suppose $r(x)$ is a continuous and differentiable function on a compact convex set $\Psi\subseteq \Rset$. Assume that a collection of points $\{x_1,x_2,\dots,x_k\}$ are sampled uniformly from $\Psi$ and $s\in\Rset_{>0}$ is the sampling spacing. Let $r_0=\mathop{max}\{|r(x_1)|,|r(x_2)|,\dots,|r(x_k)|\}$, and $\beta=\mathop{sup}_{x\in\Psi} ||\nabla r(x)||$, then 
	\begin{equation}
	|r(x)|\le \beta\cdot s+ r_0,\ \forall\ x\in \Psi.
	\end{equation}
\end{theorem}

Then our scheme is as follows.
\begin{itemize}
	\item Split the interval $I=[a,b]$ uniformly into $m$ partitions, i.e., $I_1=[a_1,b_1],I_2=[a_2,b_2],\dots,I_m=[a_m,b_m]$.
	\item For each partition $I_i=[a_i,b_i]$, define a $n$-degree polynomial $p^i_n(x):=\sum_{j=0}^n c_{ij}\cdot x^j$.
	\begin{enumerate}
		\item Pick a non-negative integer $k>n$ and sample $k$ points uniformly from $f$ over $I_i$. That is,  \[
		D=\{(x_1,f(x_1)),(x_2,f(x_2)),\dots,(x_k,f(x_k))\}
		\]
		where $x_l\in I_i$ for all $1\le l\le k$.
		\item Let $p_n^i(x_l)=f(x_l)$ for all $1\le l\le k$, then we have a linear system $\mathbf{V}\cdot \mathbf{c}=\mathbf{f}$ where 
		\begin{equation*}
		\mathbf{V}=
		\begin{bmatrix}
		   1 & x_1 & x_1^2& \cdots &x_1^n \\
		   1 & x_2 & x_2^2& \cdots &x_2^n \\
		   \vdots & \vdots & \vdots&   &\vdots \\
		   1 & x_k & x_k^2& \cdots &x_k^n \\
		\end{bmatrix},
		\end{equation*}
		$\mathbf{c}=[c_{i0},c_{i1},\dots,c_{in}]^T$ and $\mathbf{f}=[f(x_1),f(x_2),\dots,f(x_k)]^T$.
		\item By solving the above overdetermined system, we obtain $p_n^i(x)$ as the approximation of $f(x)$ over the interval $I_i=[a_i,b_i]$. 
		\item Having $p_n^i(x)$, evaluate an error bound $\gamma_i$ such that
		\begin{equation}
		\forall\ x\in I_i,\ |f(x)-p_n^i(x)|\le\gamma_i .
		\end{equation} 
		Let $r(x)=f(x)-p_n^i(x)$ and $\Psi=I_i$, then we obtain $r_0=\mathop{max}\{|r(x_1)|,|r(x_2)|,\dots,|r(x_k)|\}$ by \cref{thm:errorbound}. To derive the Lipschitz constant $\beta$ of $r(x)$ over the interval $I_i$, we pick a non-negative integer $q=10k$, and sample $q$ points uniformly from $f$, i.e., we have another collection of points $\{x'_1,x'_2,\dots,x'_q\}$. Let $\beta=\mathop{max}\{|\nabla r(x'_1)|,\dots,|\nabla r(x'_q)|\}$, then
		\[
		\gamma_i:=\beta\cdot s+r_0
		\]
		where $s$ is the corresponding sampling spacing of the $q$ points.
	\end{enumerate}
	\item Now we have a set $D_p$ of tuples of intervals, polynomials and error bounds, i.e., 
	\begin{equation}\label{eq:D_p}
	D_p=\{(I_1,p_n^1(x),\gamma_1),\dots,(I_m,p_n^m(x),\gamma_m)\}
	\end{equation}
 
\end{itemize}

The approximation error bounds $\gamma_i$'s are taken into account when we synthesize the polynomial template $h$. Given a non-polynomial function $f(x)$ such that ${\tt score}(f(x))$ occurs in the program, we obtain a set $D_p$ in the form of \eqref{eq:D_p}.  For each interval $I_i$, we introduce a new variable $r_i$ and approximate $f(x)$ over $I_i$ as $p_n^i(x)+r_i$ with $r_i\in [-\gamma_i,\gamma_i]$. That is, for $1\le i\le m$, we have
\begin{equation}
\forall\ x\in I_i,\ f(x)\approx p_n^i(x)+r_i\ \text{with }r_i\in [-\gamma_i,\gamma_i].
\end{equation}

For a state $(\loc,\pv)$ such that $\loc$ is the location before the command ${\tt score}(f(x))$, there is the unique transition $\tau=\langle \loc, true, F  \rangle$ such that $F=\langle \loc',1, \textbf{1}, f \rangle$ and $\loc'$ is the location that follows the command ${\tt score}(f(x))$. Then for all valuations $\pv\in I(\loc) \wedge \Phi_B$ and $1\le i \le m$, it should hold that 
\begin{itemize}
	\item  for all $\pv[x] \in I_i$ and $r_i\in [-\gamma_i,\gamma_i]$, we have that $\ewt(h)(\loc,\pv) \le h(\loc, \pv)$ (for upper bounds) and $\ewt(h)(\loc,\pv) \ge h(\loc,\pv)$ (for lower bounds) where
	\[
	\ewt(h)(\loc,\pv)=(p_n^i(x)+r_i)\cdot h(\loc',\pv).
	\]
\end{itemize}

\subsection{Other Experimental Results}

\begin{figure}
	\centering
	\subfigure[Pedestrian v1]{
		\begin{minipage}{5cm}
			\centering
			\includegraphics[width=2.5in,height=2in]{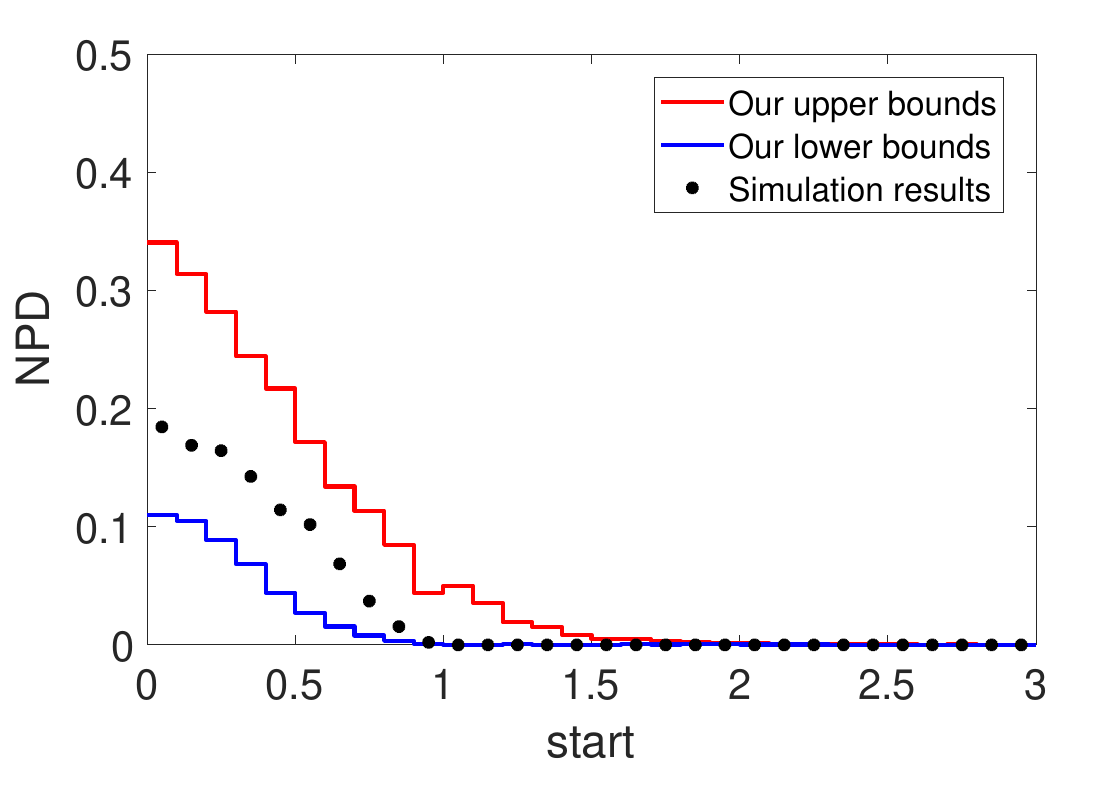}
		\end{minipage}
	}
	\quad\quad\quad\quad\quad 
	\subfigure[Hare Tortoise Race v2]{
		\begin{minipage}{5cm}
			\centering
			\includegraphics[width=2.5in,height=1.9in]{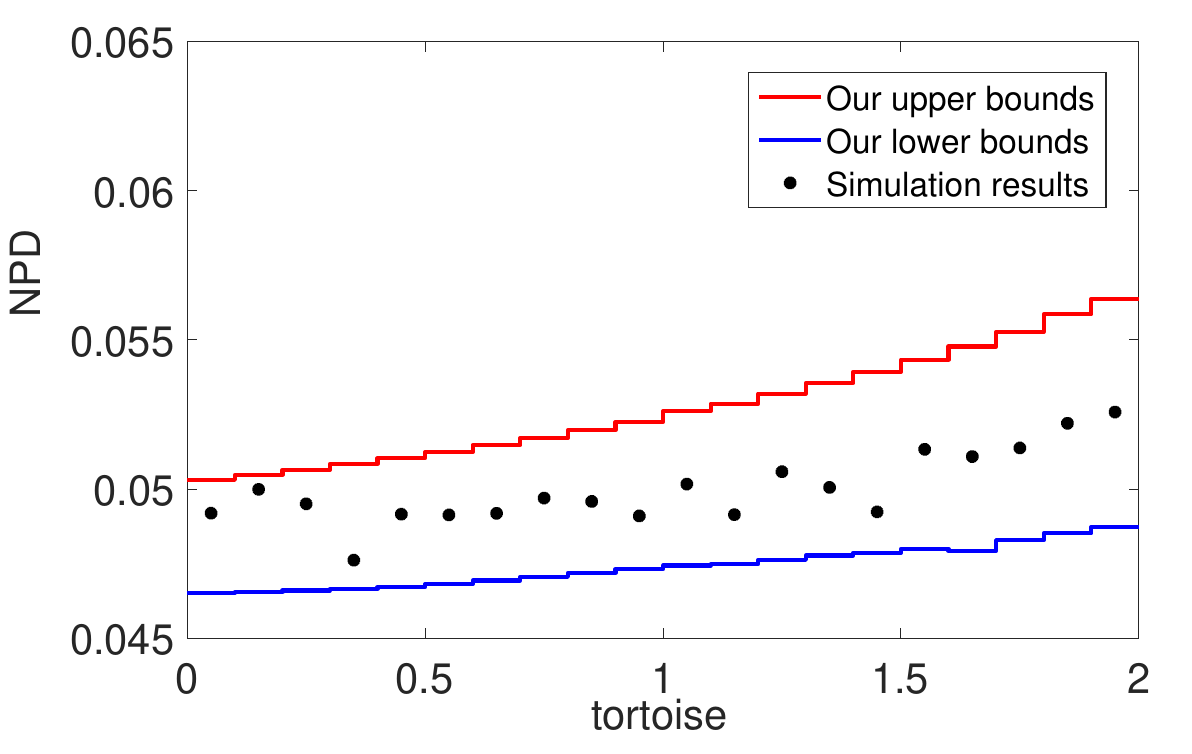}
		\end{minipage}
	}
	\subfigure[Random Walk v1]{
		\centering
		\begin{minipage}{5cm}
			\includegraphics[width=2.5in,height=2in]{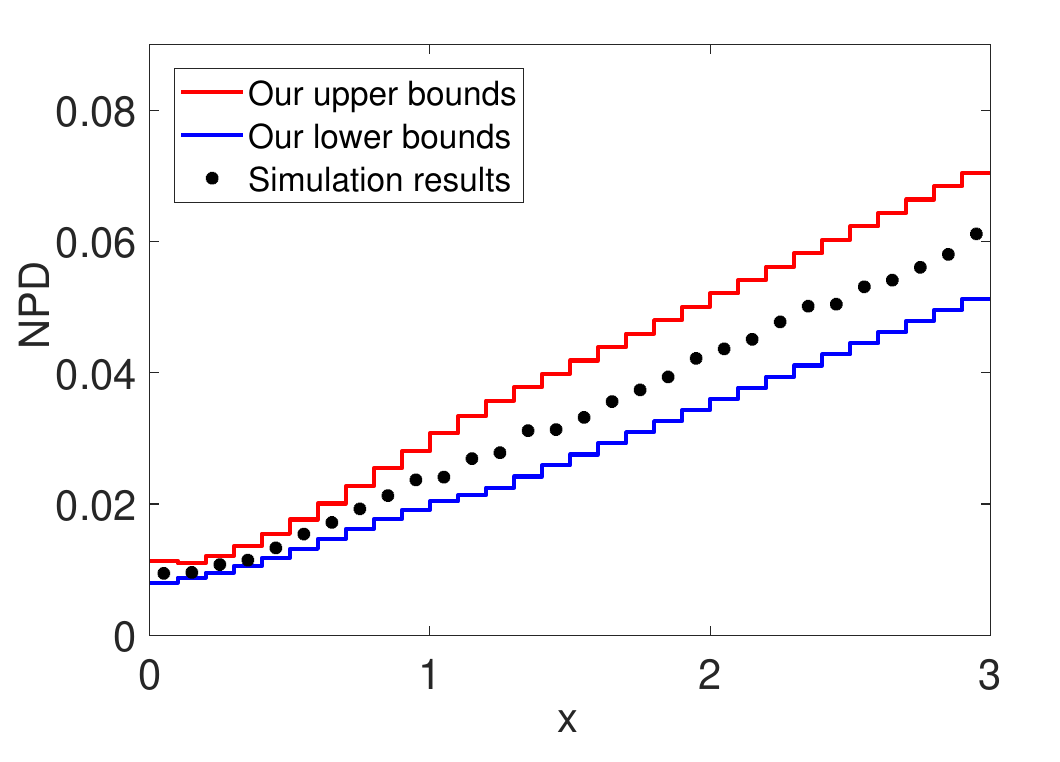}
		\end{minipage}
	}
	\quad\quad\quad\quad\quad 
	\subfigure[Pedestrian Multi-branches v3]{
		\centering
		\begin{minipage}{5cm}
			\includegraphics[width=2.5in,height=2in]{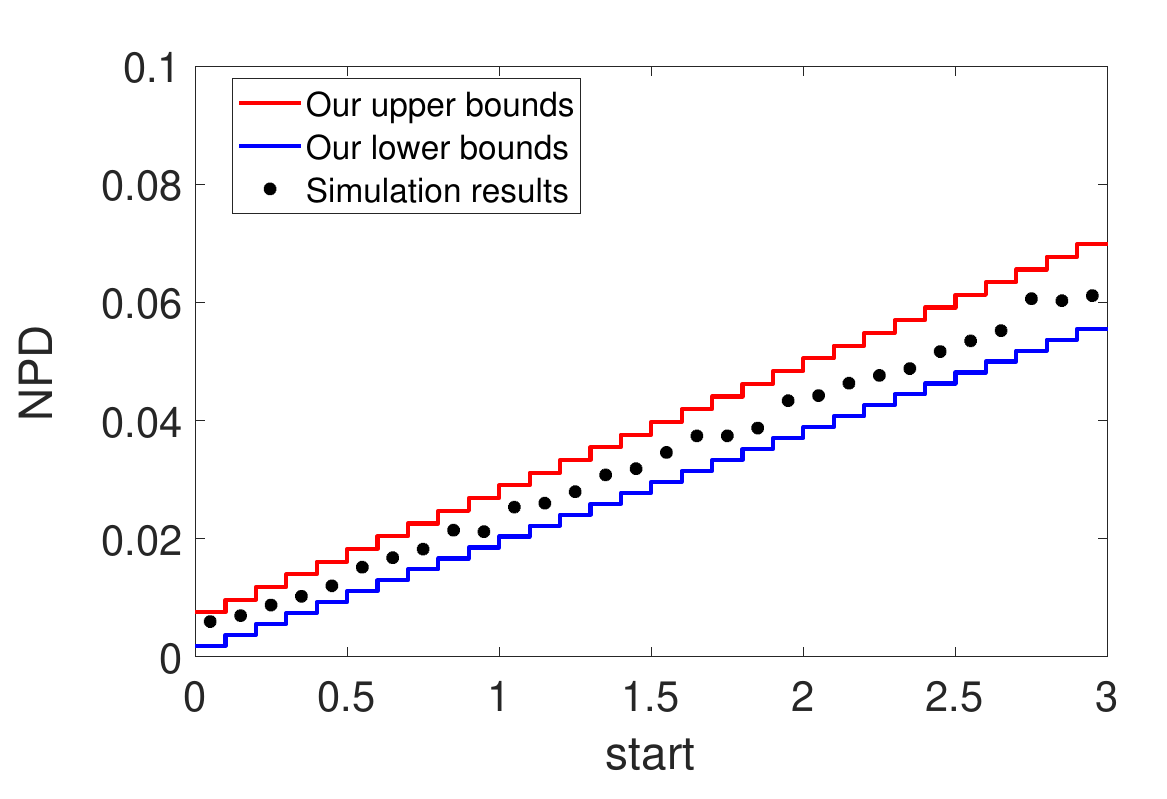}
		\end{minipage}
	}
 
	\caption{Part 1: NPD Bounds of Novel Examples}
	\label{fig:part1-results1}
\end{figure}

\begin{figure}
	\subfigure[Hare Tortoise Race v1]{
		\centering
		\begin{minipage}{5cm}
			\includegraphics[width=2.5in,height=2in]{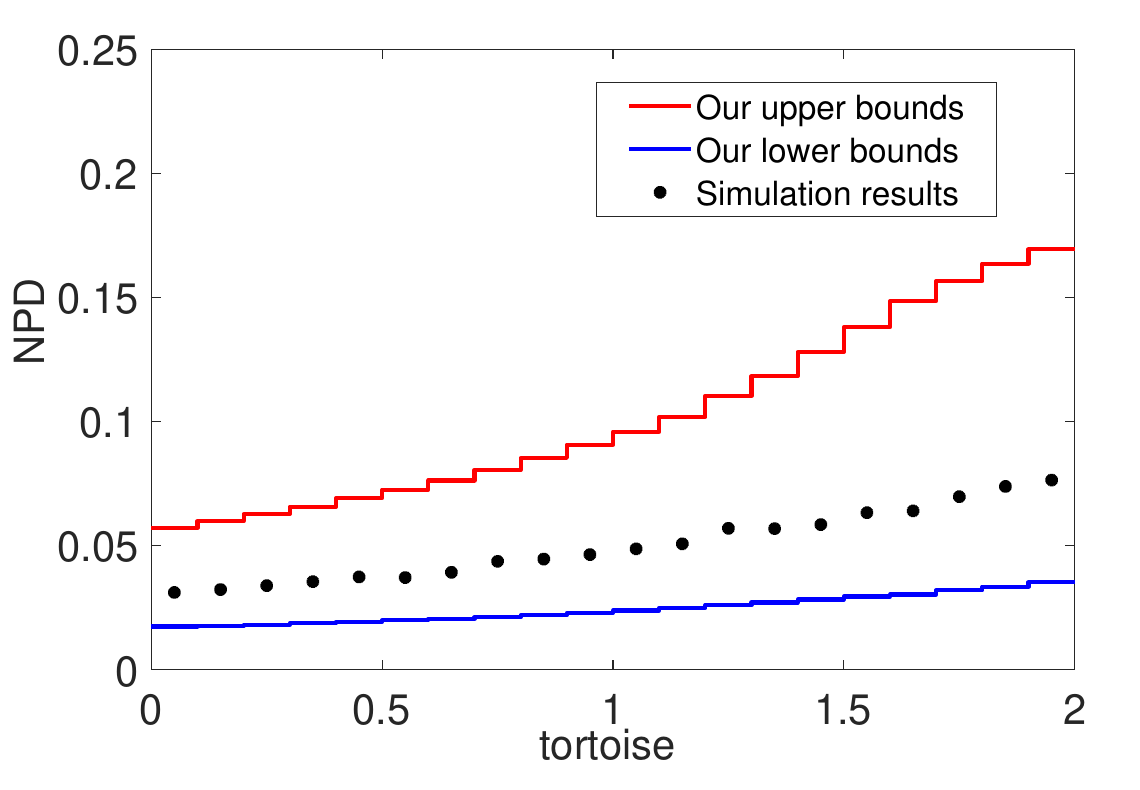}
		\end{minipage}
	}
	\quad\quad\quad\quad\quad 
	\subfigure[Random Walk v2]{
		\centering
		\begin{minipage}{5cm}
			\includegraphics[width=2.5in,height=2in]{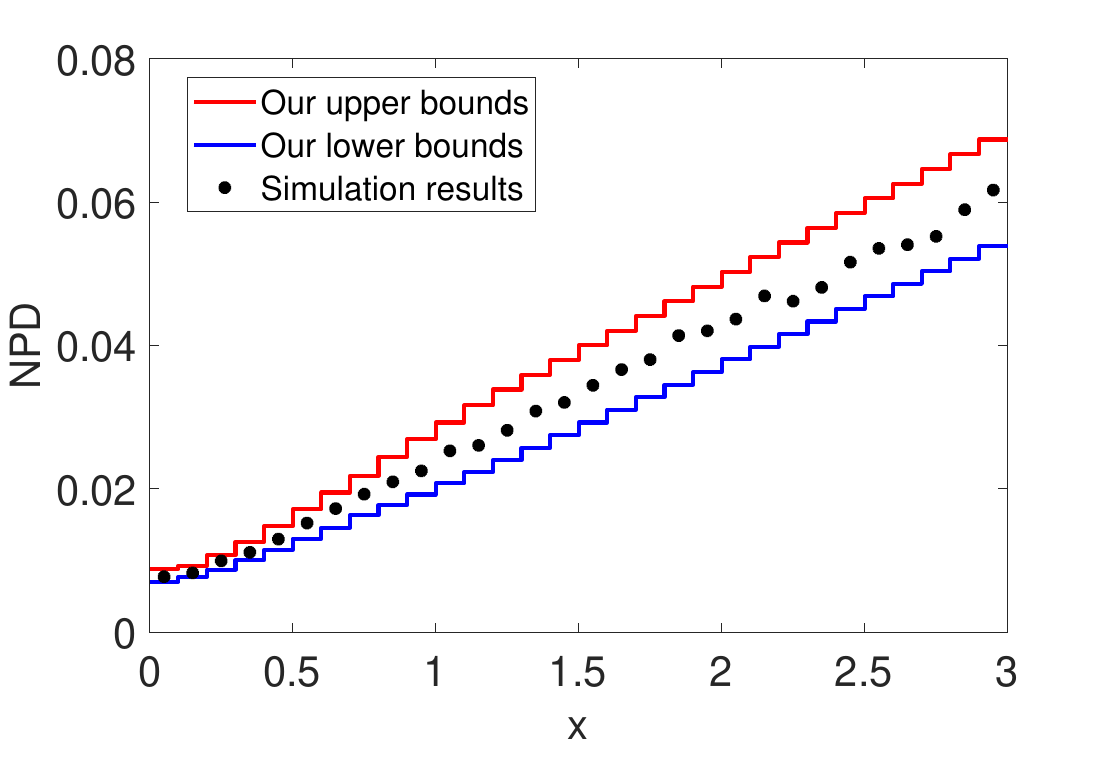}
		\end{minipage}
	}
	\subfigure[Random Walk v3]{
		\centering
		\begin{minipage}{5cm}
			\includegraphics[width=2.5in,height=2in]{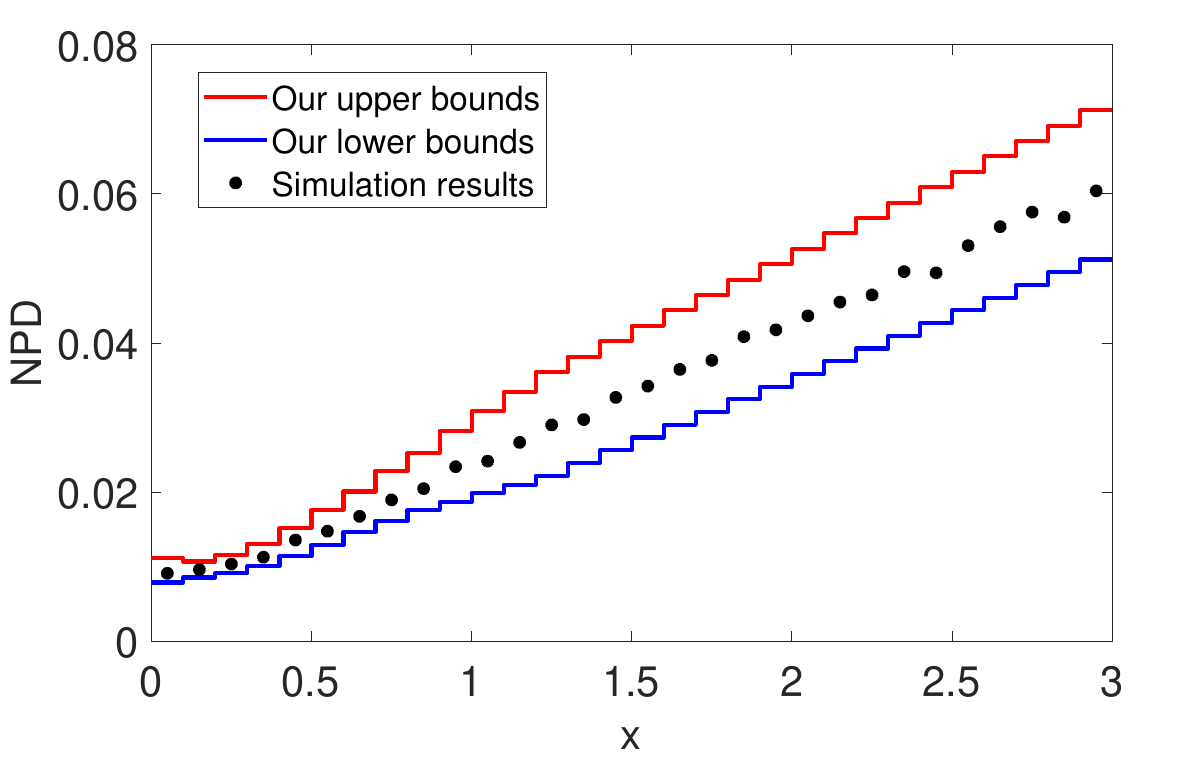}
		\end{minipage}
	}
	\quad\quad\quad\quad\quad 
	\subfigure[Random Walk v4]{
		\centering
		\begin{minipage}{5cm}
			\includegraphics[width=2.5in,height=2in]{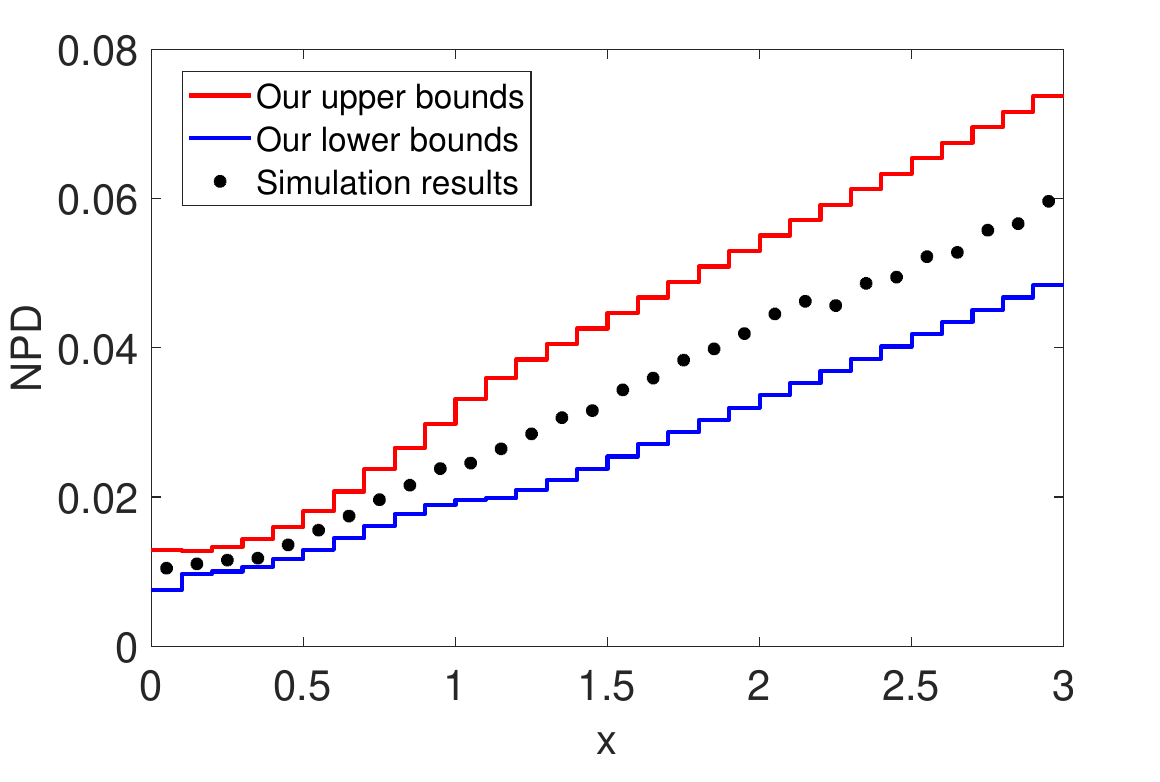}
		\end{minipage}
	}
	\subfigure[Pedestrian Multi-branches v4]{
	\centering
	\begin{minipage}{5cm}
		\includegraphics[width=2.5in,height=2in]{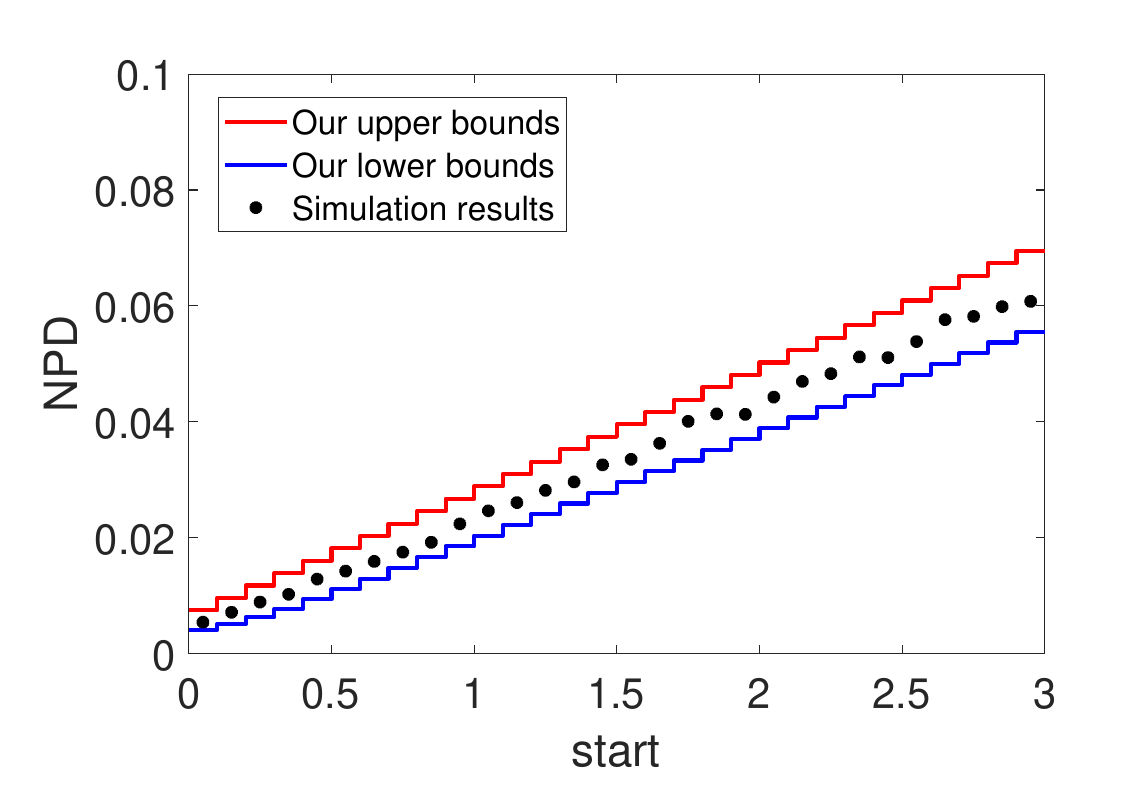}
	\end{minipage}
}
\quad\quad\quad\quad\quad 
\subfigure[Phylogenetic Birth Model]{
	\centering
	\begin{minipage}{5cm}
		\includegraphics[width=2.5in,height=2in]{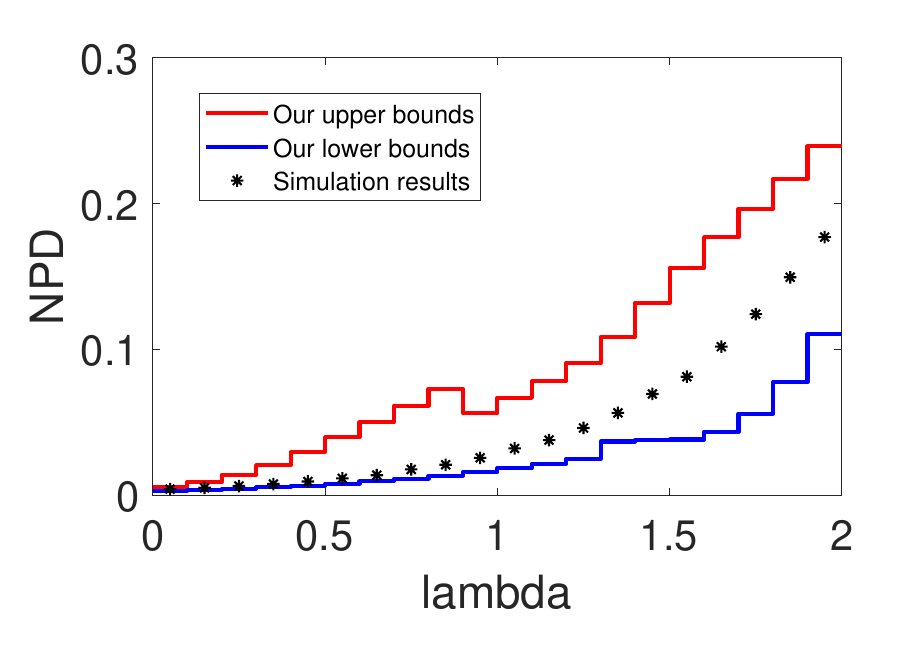}
	\end{minipage}
}
	\caption{Part 2: NPD Bounds of Novel Examples}
	\label{fig:part2-results1}
\end{figure}

\begin{figure}
	\centering
	\subfigure[Pedestrian]{
		\begin{minipage}{5cm}
			\centering
			\includegraphics[width=2.5in,height=2in]{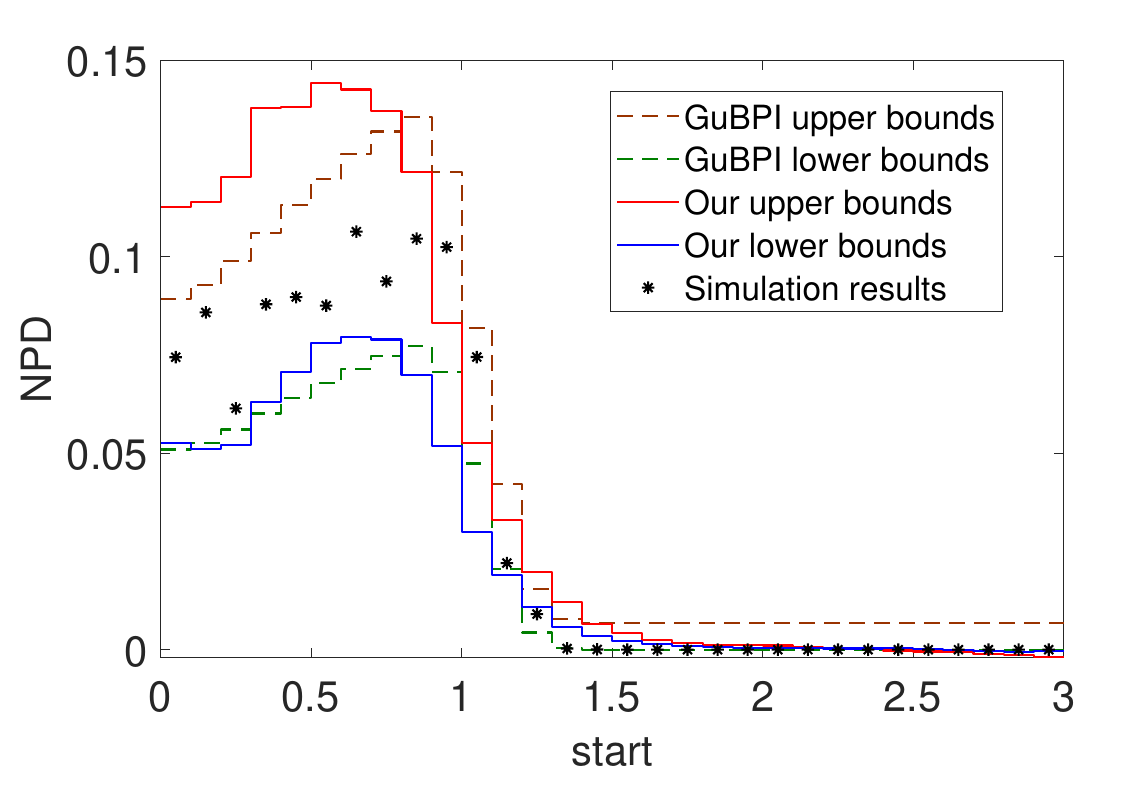}
		\end{minipage}
	}
	\quad\quad\quad\quad\quad 
	\subfigure[Pedestrain Large Deviation]{
		\begin{minipage}{5cm}
			\centering
			\includegraphics[width=2.5in,height=1.9in]{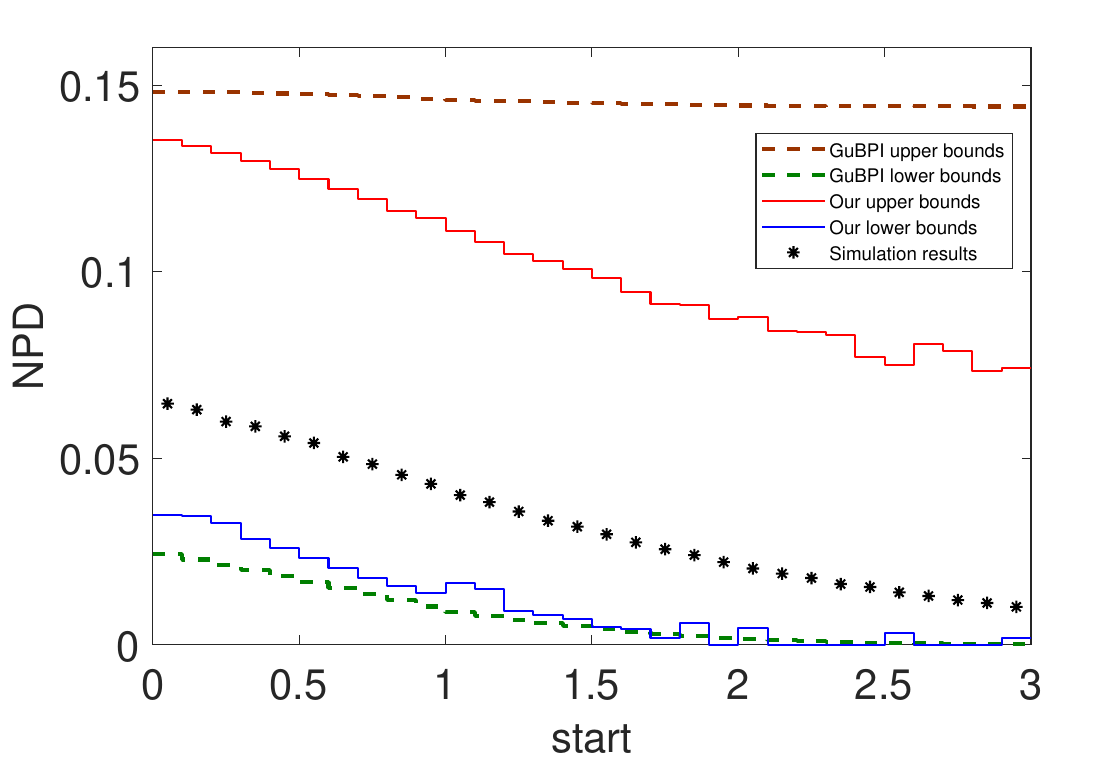}
		\end{minipage}
	}
	\subfigure[Pedestrian Beta v1]{
		\centering
		\begin{minipage}{5cm}
			\includegraphics[width=2.5in,height=2in]{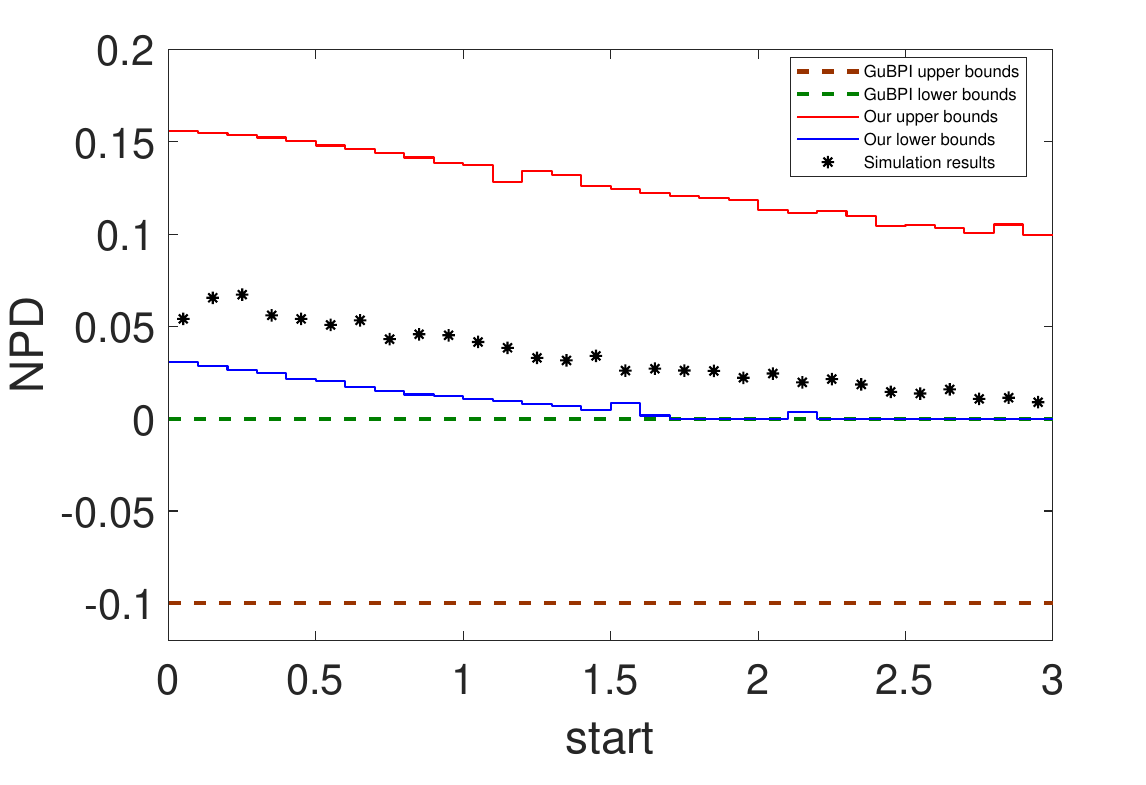}
		\end{minipage}
	}
	\quad\quad\quad\quad\quad 
	\subfigure[Pedestrian Beta v2]{
		\centering
		\begin{minipage}{5cm}
			\includegraphics[width=2.5in,height=2in]{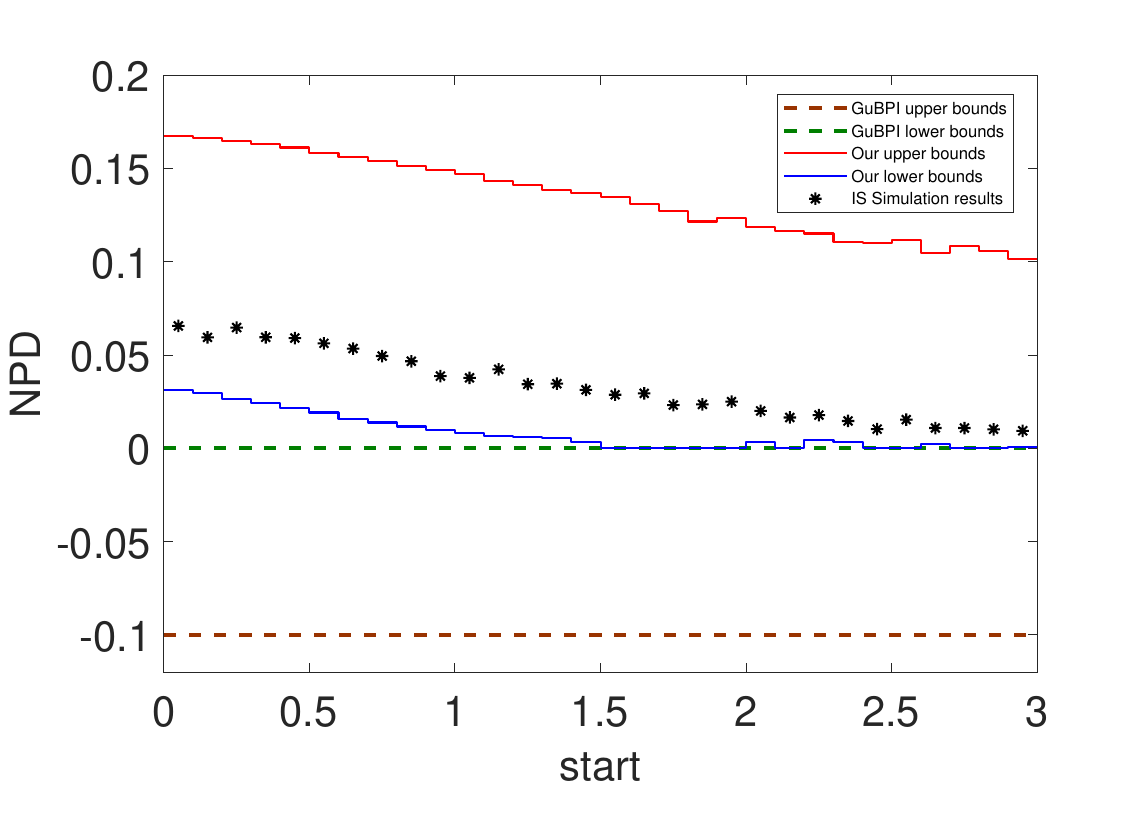}
		\end{minipage}
	}
	\caption{Part 1: NPD Bounds of Our Approach and GuBPI}
	\label{fig:part1-results2}
\end{figure}

\begin{figure}
	\centering
		\subfigure[Pedestrian Beta v3]{
		\centering
		\begin{minipage}{5cm}
			\includegraphics[width=2.5in,height=2in]{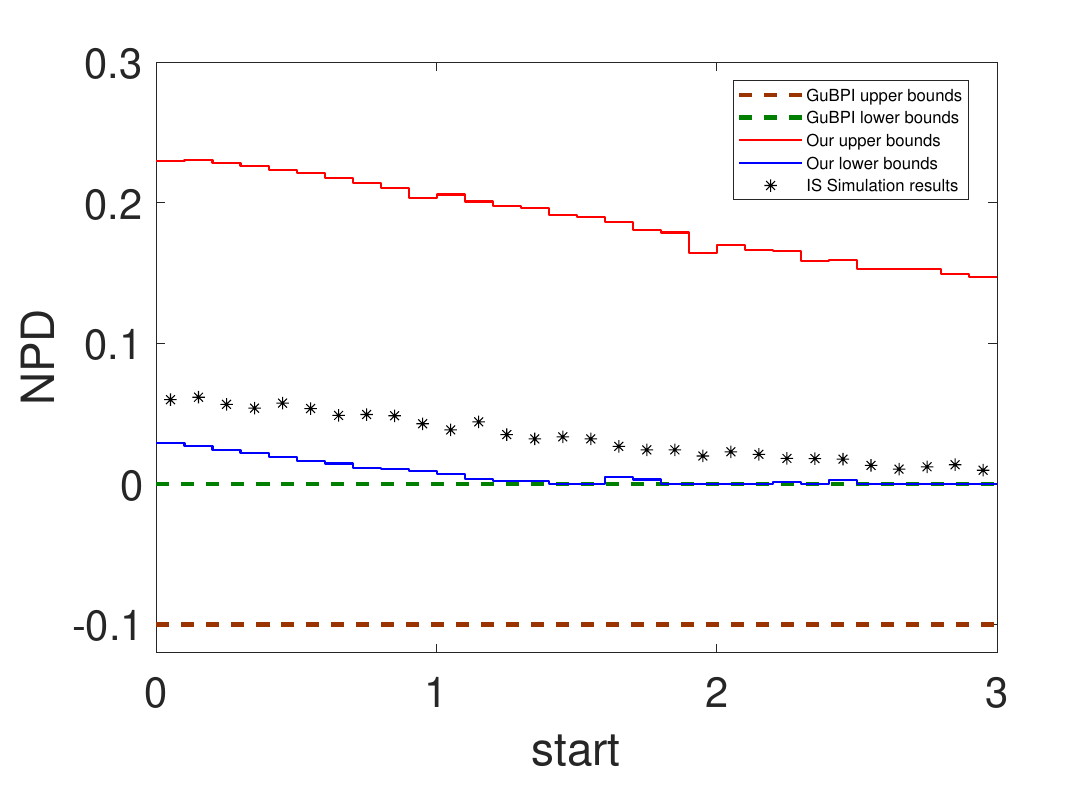}
		\end{minipage}
	}
	\quad\quad\quad\quad\quad 
	\subfigure[Pedestrian Beta v4]{
		\centering
		\begin{minipage}{5cm}
			\includegraphics[width=2.5in,height=2in]{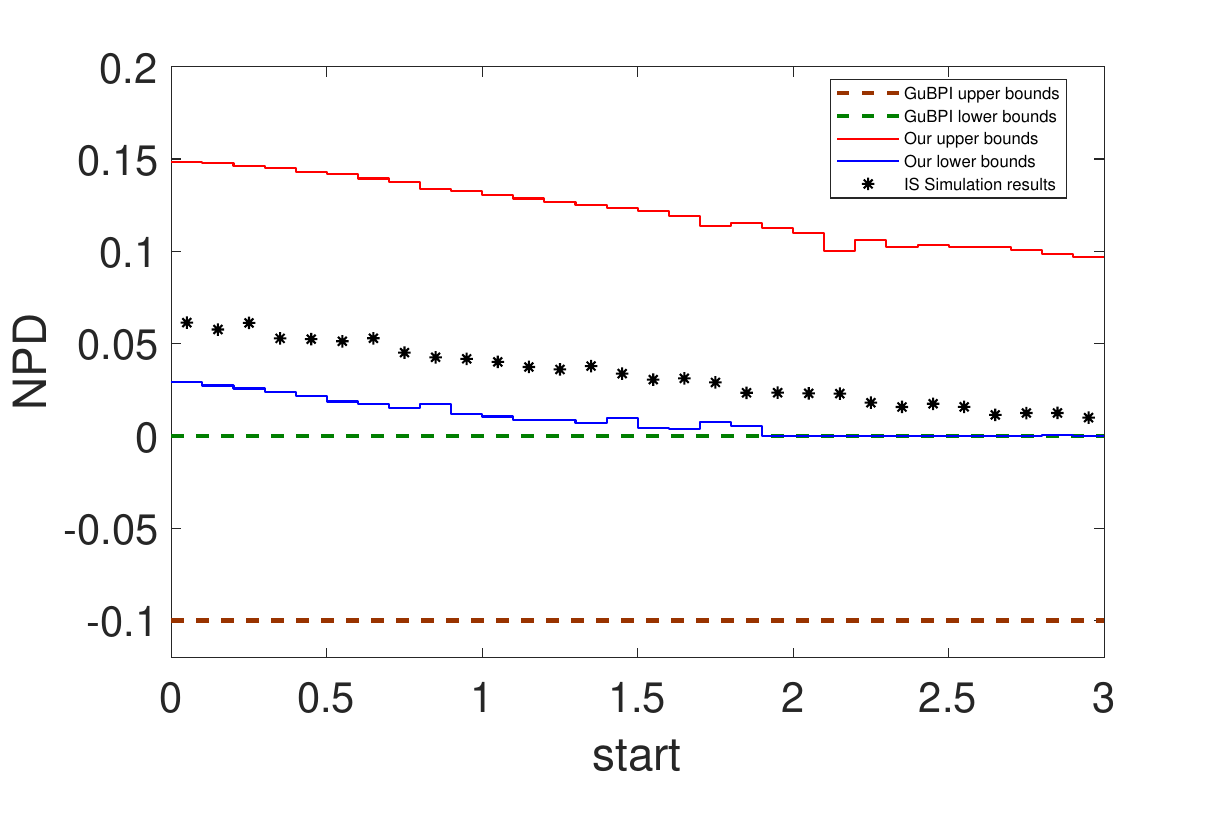}
		\end{minipage}
	}
	\subfigure[Pedestrian Multi-branches v5]{
		\begin{minipage}{5cm}
			\centering
			\includegraphics[width=2.5in,height=2in]{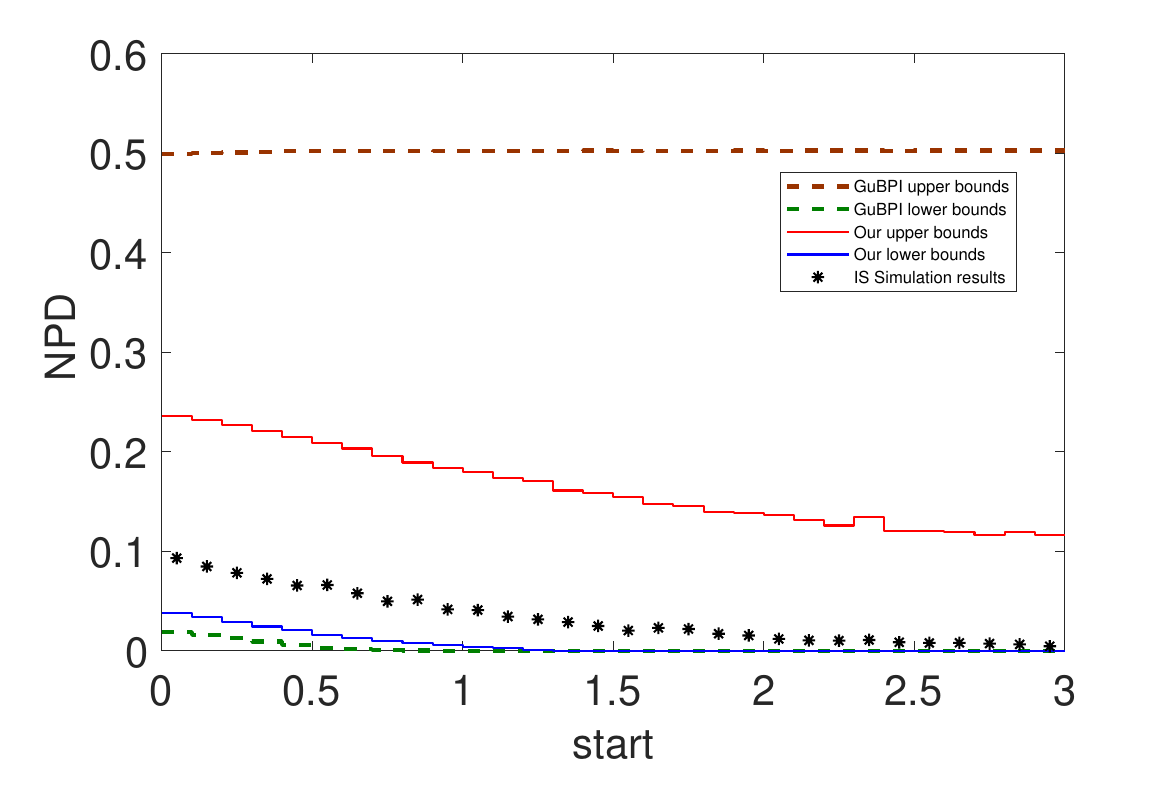}
		\end{minipage}
	}
	\quad\quad\quad\quad\quad 
	\subfigure[Para Estimation Recursive]{
		\begin{minipage}{5cm}
			\centering
			\includegraphics[width=2.5in,height=1.9in]{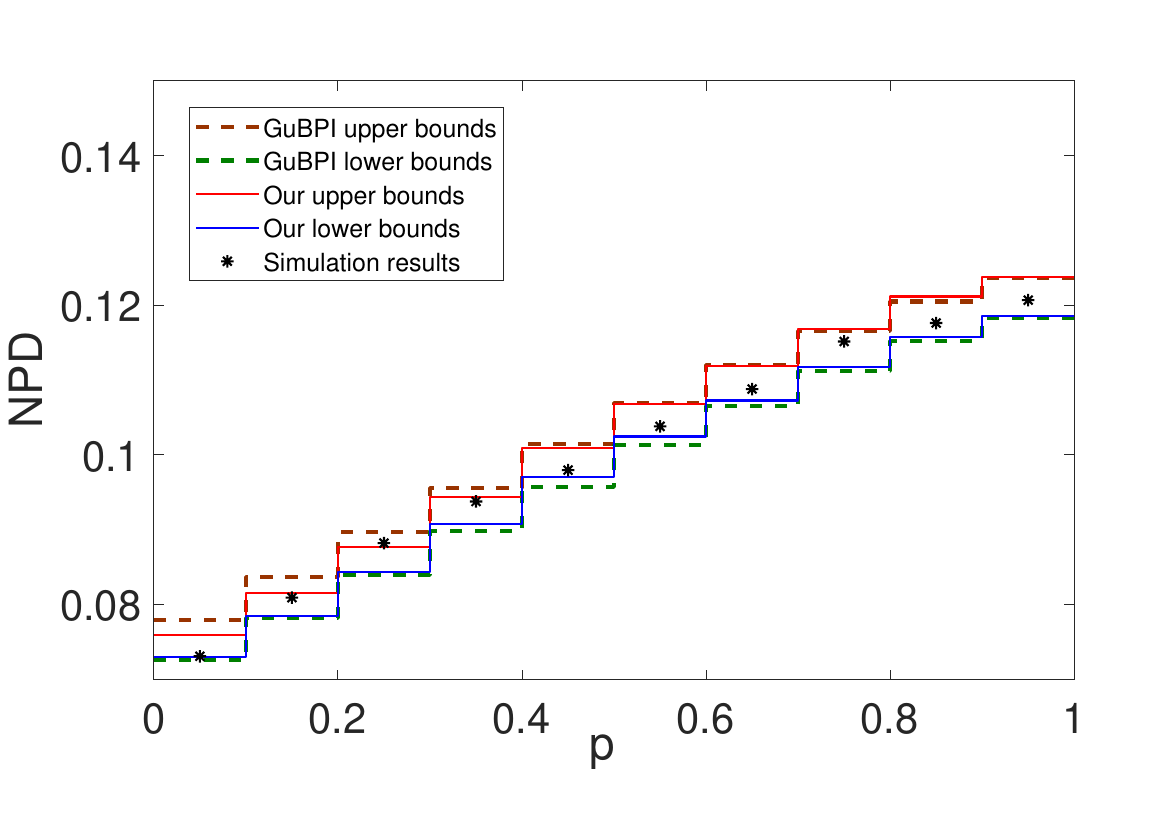}
		\end{minipage}
	}

	\caption{Part 2: NPD Bounds of Our Approach and GuBPI}
	\label{fig:part2-results2}
\end{figure}

\end{document}